\documentclass[a4paper,12pt,twoside,reqno]{amsbook}
\usepackage{enumerate}
\usepackage{graphicx}
\usepackage{amsfonts}
\usepackage{color}
\usepackage{bbm}
\usepackage{amsmath}
\usepackage{amsthm}
\usepackage{amssymb}
\usepackage[latin1]{inputenc}
\usepackage{hyperref}
\usepackage{csquotes}
\usepackage{tikz}
\usetikzlibrary{calc}
\usepackage{currvita}
\theoremstyle{change}

\newcommand{\IR}{{\mathbb{R}}}
\newcommand{\IN}{{\mathbb{N}}}
\newcommand{\IZ}{{\mathbb{Z}}}

\newcommand{\IK}{{\mathbb{K}}}

\newcommand{\IS}{{\mathbb{S}}}

\newcommand{\FA}{{\mathfrak{A}}}

\newcommand{\FE}{{\mathfrak{E}}}

\newcommand{\FP}{{\mathfrak{P}}}
\newcommand{\FU}{{\mathfrak{U}}}
\newcommand{\FV}{{\mathfrak{V}}}

\newcommand{\FX}{{\mathfrak{X}}}
\newcommand{\FY}{{\mathfrak{Y}}}
\newcommand{\FZ}{{\mathfrak{Z}}}

\newcommand{\Fh}{{\mathfrak{h}}}
\newcommand{\Fp}{{\mathfrak{p}}}

\newcommand{\CB}{{\mathcal{B}}}
\newcommand{\CC}{{\mathcal{C}}}

\newcommand{\CD}{{\mathcal{D}}}
\newcommand{\CE}{{\mathcal{E}}}

\newcommand{\CG}{{\mathcal{G}}}
\newcommand{\CH}{{\mathcal{H}}}
\newcommand{\Ch}{{{h}}}

\newcommand{\CM}{{\mathcal{M}}}
\newcommand{\CQ}{{\mathcal{Q}}}
\newcommand{\CS}{{\mathcal{S}}}

\newcommand{\loc}{{\mathrm{loc}}}
\renewcommand{\rho}{\varrho}
\renewcommand{\theta}{\vartheta}
\renewcommand{\tilde}{\widetilde}

\DeclareMathOperator{\Deg}{deg}
\DeclareMathOperator{\dom}{dom}
\DeclareMathOperator{\supp}{supp}
\DeclareMathOperator{\dist}{dist}

\DeclareMathOperator{\diam}{diam}
\DeclareMathOperator{\kap}{cap}
\newcommand{\Id}{\operatorname{Id}}
\newcommand{\tr}[1]{\operatorname{tr}_{ #1 }} % Vektor mit allen Vektoren f(v), d.h. (f(v))_{v\in V}
\theoremstyle{plain}
\newtheorem{lemma}{Lemma}[chapter]
\newtheorem{coro}[lemma]{Corollary}
\newtheorem{theorem}[lemma]{Theorem}
\newtheorem{prop}[lemma]{Proposition}
\theoremstyle{definition}
\newtheorem{definition}[lemma]{Definition}
\newtheorem{example}[lemma]{Example}
\theoremstyle{remark}
\newtheorem*{remark}{Remark}
\newtheorem*{reminder}{Reminder}
\newtheorem*{assumption}{Assumption}
 \newlength\headseptemp

\newcommand{\Hmm}[1]{\leavevmode{\marginpar{\tiny%
$\hbox to 0mm{\hspace*{-0.5mm}$\leftarrow$\hss}%
\vcenter{\vrule depth 0.1mm height 0.1mm width \the\marginparwidth}%
\hbox to 0mm{\hss$\rightarrow$\hspace*{-0.5mm}}$\\\relax\raggedright #1}}}
\makeindex
\begin{document}
% -------------------------------------------------------------%
\title[]{Analysis of Dirichlet forms on graphs}
% -------------------------------------------------------------%

\author[]{Sebastian Haeseler}
\date{\today}
\address{Mathematisches Institut\\ Friedrich Schiller Universit\"at Jena\\
  D-07743~Jena, Germany}
% sebastian.haeseler@uni-jena.de.}

% -------------------------------------------------------------%
% -------------------------------------------------------------%

% -------------------------------------------------------------%
\maketitle
January 20, 2014
\frontmatter

\tableofcontents
\chapter*{Preface}
The title of this thesis is a synthesis of "Analysis on Dirichlet spaces" and "Dirichlet forms on graphs". The first is coming from a series of papers by Sturm \cite{St-94,St-95,St-96} with the title "Analysis of local Dirichlet spaces". There certain global properties are under consideration with the help of the so-called intrinsic metric. The other part of the title comes from the series of papers by Lenz et al, \cite{KL-10,KL-11,HKLW-12}. There Dirichlet forms on discrete spaces are considered. It turns out that in the case of regular forms, there is always a weighted graph involved. So the idea of this thesis is to combine both aspects. What we get is in the sense of the Beurling-Deny formulae for regular Dirichlet forms a very general explicit Dirichlet form. Given this explicit form, we can study certain analytical properties. We can rephrase this in the following leitmotif:
\begin{itemize}
\item[(A)] We want to obtain qualitative results for general Dirichlet forms by studying the simplified Dirichlet forms on graphs.
\end{itemize}
The theory Dirichlet forms started with the pioneering work of Beurling and Deny in 1958, \cite{BD-58,BD-59}. Its roots can be seen in the development of classical potential theory in the beginning of the 20th century. The discovery of the close connection between classical potential theory and probability theory also had a huge impact on Dirichlet form theory. A milestone appeared in terms of the book of Fukushima \cite{F-80}. Dirichlet forms could be considered as a generalization of Laplacians on manifolds, certain pseudodifferential operators and discrete Laplacians on graphs. There has been a huge development on the subject since 1980 and dozens of articles appear each year. This leads to our second leitmotif:
\begin{itemize}
\item[(B)] We use Dirichlet form theory to obtain results for Laplacians on metric graphs.
\end{itemize}
Metric graphs have received a lot of attention in recent years, both from the point of view of mathematicians and application (see e.g. the conference proceedings \cite{EKKST-08} and the survey \cite{Ku-04}). A metric graph is by definition a combinatorial graph where the edges are considered as intervals, glued together according to the combinatorial structure. This combination of a continuum and combinatorial world, gives the metric graph a standing between manifolds and discrete graphs, and makes it thus an interesting subject to be studied. Closely connected with the theory of metric graphs is the theory of quantum graphs. Basically, those are metric graphs equipped with some singular behavior in the vertices. The latter makes it impossible to apply the machinery of regular Dirichlet forms directly to questions on recurrence, regularity and stochastic completeness to it. Fortunately our approach allows us to fill this gap. Also for many results in the theory of quantum graphs, a lower bound on the edge lengths is needed. Due to our approach, we can remove this assumption. As metric graphs inherit a nature between the discrete and continuous, one may ask which impact has its combinatorial part on metric graphs. Thus our third leitmotif is given as:
\begin{itemize}
\item[(C)]
We want to reduce questions concerning metric graphs to combinatorial problems.
\end{itemize}
Attempts in this direction can be found in \cite{Cattaneo} where spectral properties where reduced to a combinatorial problem. This was developed further by using so-called boundary triples, \cite{Pa-06,BGP,Po-12}. In order to be successful, one needs to put strong assumptions on the graph in terms of equal edge lengths and  an appropriate boundary space. Here, our approach gives a boundary space for a general graph, without any essential condition on the edge lengths. As an application, we deduce connections of global properties between metric and discrete graphs.\medskip\\

The thesis is organized as follows:\smallskip\\

We begin with the definition of topological graphs. They form the state space on which the Dirichlet forms will be defined. The definition will be analogous to the definition of topological manifolds. Usually metric graphs are defined in a constructive way, by clueing intervals together according to a graph structure. Our approach is more natural as we view topological graphs as topological space which locally look like a finite union of rays in the plane. This gives an intrinsic description of topological graphs and analogous concepts from the theory of manifolds can be also introduced in our setting. We are thus left with showing that this concept agrees with the usual one. This is done by associating a graph to the topological space. By the already mentioned gluing procedure, we show that the associated $CW$-complex agrees with the topological graph in Theorem 1.9. As we want to introduce Dirichlet forms on this space, we continue with the introduction of several weights. The first and a priori most prominent weight will play the role of the length of the edges. After its introduction we form a path metric space out of the topological graph. In particular certain notions from metric spaces need to be considered. As a first highlight, this leads to a Hopf-Rinow type theorem. Other weights which are introduced there are not connected with the metric properties. However they give birth to the discreteness as they will be related to jumps and killing of the stochastic process which is associated with the Dirichlet form. Before we can introduce them, we are left with the definitions of continuity, measurability and differentiability. Here the manifold like approach pays out, as we can use charts in those definition. The only problem occurs in the differentiability in the vertices. We solve this problem by ignoring, namely our substitute for differentiability is absolute continuity, and there we will only impose continuity in the vertices.\medskip\\
After this foundation chapter, we introduce in chapter 2 Dirichlet forms on metric graphs. This is done in the spirit of the leitmotifs $(A)$ and $(B)$. We begin with the so-called diffusion part. This form could be imagined as an elliptic form defined on a weighted Sobolev space. Here comes an important tool into play, viz, the length transformation. Using this we introduce subsequently normal forms of the graph, in particular metrics, which helps us in the analysis of the forms. The first analysis concerns embedding properties of the domain of the graph. The second analysis concerns the regularity of the Dirichlet space. Although both problems seem to be related, its treatment uses two metrics which are different in general. One metric, the canonical metric, could be seen as a generalization of a concept due to Feller \cite{Feller}, see also \cite{CF-12}, and the other one is the intrinsic metric coming from the theory of strongly local Dirichlet spaces. After the analysis  of the diffusion part we consider the graph Dirichlet forms, which is basically a perturbation of the diffusion part by a discrete graph. After its introduction and its basic properties, we consider the case when the discrete part is a bounded perturbation. In this case no further tools as the ones for the diffusion part are needed for the questions treated so far. We finish this second chapter with an reprise concerning the regularity of general graph Dirichlet forms.  This is done in terms of harmonic functions, a concept which becomes the melody of the leitmotif $(C)$ in chapter 3.\medskip\\
Before we continue with our program, in the interlude we have look on the connection of quantum graphs with the graph Dirichlet forms. More precisely we study Dirichlet quantum graphs, i.e. quantum graphs whose matching conditions lead to Dirichlet forms. At the first look, one might think that quantum graphs are more general than graph Dirichlet forms. This however is not the case, as we show in theorem 2.50 that the underlying graph structure for quantum graphs is somehow wrong. In particular, we construct a new metric graph and a Dirichlet form on it which are isometric isomorph in the sense of Dirichlet space to the given Dirichlet quantum graph. Using this result, we conclude that Dirichlet quantum graphs fit into our setting.\medskip\\
The third chapter is devoted to the study of the associated operators. We continue the study of the Dirichlet space there, with the introduction of the local space, or rather the space of functions being locally in the domain. This gives rise to the definition of a distributional operator, which we will call the weak Laplacian. After a characterization of this operator in Theorem 3.6, we already see the strong connection with a discrete problem, which will be strengthened to a pre-version of the Krein formula in Theorem 3.18. In the subsequent section we continue the deep analysis of this observation, which culminates in Theorems 3.28 and 3.29, where we connect the discrete and the continuous setting by means of operator theory. The main tools are harmonic functions there. The discretization procedure leads to what we call the Kirchhoff Dirichlet form. In the spirit of the  third leitmotif, we obtain a lot of results connecting discrete and continuous graphs. Among them are prominent properties as regularity in Theorem 3.30, stochastic completeness in Theorem 3.32 and a Sobolev inequality in Theorem 3.37. In the final section of this chapter, in particular Theorem 3.47, we relate the generators of the graph Dirichlet forms to the generators of the Kirchhoff Dirichlet forms. This could be seen as a version of the Krein resolvent formulae.\medskip\\
In the last chapter we turn back to the graph Dirichlet form. The aim there is to extend the form to all functions where the energy functional is finite. This leads in general Dirichlet form theory to the concepts of reflected and extended Dirichlet spaces. In our case we can actually equip the first one with a norm, and show afterwards in Theorem 4.6 that the latter one is a closed subspace with respect to this norm. Both spaces are strongly connected with the notion of recurrence and transience. Due to this concepts, we include a short discussion of potential theoretic machinery. Afterwards we concern the special case of a diffusion and also relate recurrence with the canonical boundary in Proposition 4.18. In the last section we calculate the trace Dirichlet form of a graph Dirichlet form. This gives a characterization of recurrence in terms of a discrete Dirichlet form in Theorem 4.24. \medskip\\
\begin{remark}
Inequalities and estimations from above and below play a crucial role. There a constant $C>0$ will appear frequently. Note that we follow tradition that this constant may change from line to line and will depend only on the values mentioned explicitly in the corresponding result.
\end{remark} 
\cleardoublepage
\thispagestyle{empty}
\vspace*{60mm}
\begin{center}
\textbf{{Acknowledgements}}\smallskip\\
I want to thank my advisor Prof. Dr. Daniel Lenz for his support over the last years and his help concerning this thesis. I am also very thankful to his group in Jena. In particular, I want to thank Xueping Huang, Marcel Schmidt and Carsten Schubert for reading the manuscript carefully and for their continuing interest in the subject. I am very grateful to my parents. Without their permanent support and encouragement this thesis would not have been possible. Finally, I want to express my love and gratitude to Susann.
\end{center}
\mainmatter
\chapter{Foundations}
This chapter is devoted to the study of topological graphs. After defining them intrinsically, we show that they equal their associated $CW$-complex. In section 2, we introduce weights and discuss their specific role which leads to the concept of metric graphs. Using this approach we introduce certain classical function spaces and notions related to them.
\section{Topological and discrete graphs}
The prototype of a topological graph is a set of intervals in $\IR^2$ having one point in common. This leads to the following definition which plays the same role for topological graphs as open subsets of Euclidean spaces as for manifolds.
\begin{definition}\index{star graph}
\begin{enumerate}
\item The star-graph of degree $d$, $d\in \IN$, is the set
\[\bigstar_d := \{ r(\cos k \tfrac{2\pi}{d},\sin k \tfrac{2\pi}{d})\in \IR^2 \mid 0\leq r<1, k=1,\dots, d\}\]
endowed with the induced topology from $\IR^2$.\\
\item Let $d\in \IN$, then a $d$-star-shaped chart with center $x\in \FX$ on a topological space $\FX$ is a pair $(U,\varphi)$ consisting of an open subset $U\subset \FX$ with $x\in U$ and a homeomorphism $\varphi: U \to \bigstar_d$ such that $\varphi(x) = 0$.
\end{enumerate}
\end{definition}
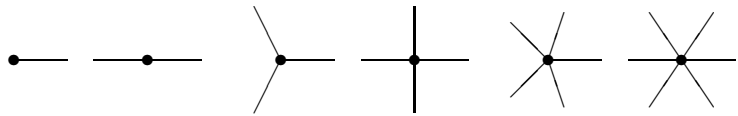
\begin{figure}[h]
\setlength{\unitlength}{2pt}
\begin{center}
\begin{picture}(135,20)(0,-10)
  \put(0,0){\line(1, 0){10}}
  \put(0,0){\circle*{2}}

  \put(25,0){\line(1, 0){10}}
  \put(25,0){\circle*{2}}
  \put(25,0){\line(-1, 0){10}}

    \put(50,0){\line(1, 0){10}}
  \put(50,0){\line(-1, 2){5}}
  \put(50,0){\circle*{2}}
  \put(50,0){\line(-1, -2){5}}

    \put(75,0){\line(0, 1){10}}
  \put(75,0){\line(1, 0){10}}
  \put(75,0){\circle*{2}}
  \put(75,0){\line(0, -1){10}}
  \put(75,0){\line(-1, 0){10}}

    \put(100,0){\line(1, 0){10}}
  \put(100,0){\circle*{2}}
  \put(100,0){\line(1, 3){3}}
  \put(100,0){\line(-1, 1){7}}
  \put(100,0){\line(-1, -1){7}}
  \put(100,0){\line(1, -3){3}}

    \put(125,0){\line(1, 0){10}}
  \put(125,0){\circle*{2}}
  \put(125,0){\line(-1, 0){10}}
  \put(125,0){\line(2, 3){6}}
  \put(125,0){\line(-2, 3){6}}
  \put(125,0){\line(-2, -3){6}}
  \put(125,0){\line(2, -3){6}}
\end{picture}
\end{center}
\caption{The star graphs $\bigstar_1$,$\bigstar_2$,$\bigstar_3$,$\bigstar_4$, $\bigstar_5$ and $\bigstar_6$}
\end{figure}
In analogy to manifolds we define topological graphs as topological spaces which look locally like a star-graph.  This was done for compact graphs in \cite{BR-07}. We generalize this concept for non-compact graphs as follows.
\begin{definition}\index{topological graph}
Let $(\FX,\CG)$ a topological space. We call $\FX$ a topological graph if
\begin{itemize}
\item[(i)] $(\FX,\CG)$ is Hausdorff,
\item[(ii)] $(\FX,\CG)$ has a countable base,
\item[(iii)] for each $x\in \FX$ there exists a $d$-star-shaped chart $(U,\varphi)$ with center $x$ for some $d\in\IN$.
\end{itemize}
\end{definition}
\begin{figure}[h]
\setlength{\unitlength}{2pt}
\begin{center}
\begin{picture}(90,63)(-10,-17)
%hier knoten
\put(0,0){\circle*{2}}
\put(0,40){\circle*{2}}
\put(30,0){\circle*{2}}
\put(70,0){\circle*{2}}
\put(50,20){\circle*{2}}
\put(50,-20){\circle*{2}}
\put(20,50){\circle*{2}}
\put(50,30){\circle*{2}}
\put(70,20){\circle*{2}}
%hier kanten
\put(0,0){\line(1,0){30}}
\put(0,0){\line(0,1){40}}
\put(0,40){\line(3,-4){30}}
\put(30,0){\line(1,1){20}}
\put(30,0){\line(1,-1){20}}
\put(50,20){\line(1,-1){20}}
\put(50,-20){\line(1,1){20}}
\put(50,20){\line(-1,1){30}}
\put(50,20){\line(0,1){10}}
\put(50,20){\line(1,0){20}}
\end{picture}
\end{center}
\caption{A compact graph}
\end{figure}
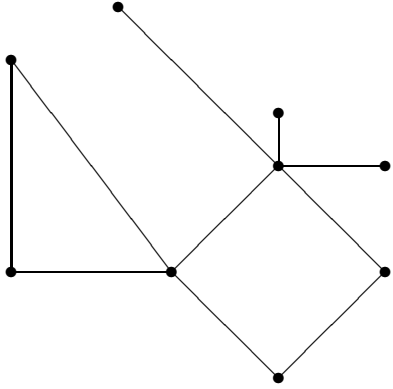
One may wonder why we include $(i)$ and $(ii)$ in the definition of a topological graph. The answer is the same as in the manifold case, to avoid pathological examples. If, for instance, we drop $(i)$, we would include graphs with multiple origins, which can be seen as follows. Consider the disjoint union of $n$ copies of a star graph $\bigstar_d$
\[\bigstar := \bigcup_{i=1}^n \{i\}\times \bigstar_d\]
and consider the equivalence relation
\[(i_1,r_1,k_1)\sim  (i_2,r_2,k_2) \quad:\Leftrightarrow \quad r_1=r_2\neq 0, k_1=k_2.\]
Then $(ii)$ and $(iii)$ hold true for the quotient space $\bigstar\slash_\sim$, but not $(i)$. To see that second countability makes sense in the definition of a topological graph, a construction similar to the so-called long line, see for instance \cite{SS-95}, would lead to topological spaces satisfying $(i)$ and $(iii)$, but not $(ii)$.\medskip\\
Next, let us introduce more concepts analogous to the manifold setting. For any $d$-star-shaped chart $(U,\varphi)$ on $\FX$ the local coordinate system $(r,k)$ is defined in $U$ by taking the $\varphi$-pullback of the polar coordinates in $\IR^2$ restricted to $\bigstar_d$. Basically a chart is an open set $U\subset \FX$ with local coordinates $(r,k)\in\IR_+\times\{1,\dots,d\}$. With a little abuse of notation we will identify $U$ with its image such that the coordinates $(r,k)$ can be identified with the polar coordinates in $\bigstar_d$.\medskip\\
Given two charts $(U,\varphi)$ and $(V,\psi)$ on a topological graph $\FX$, two coordinate systems, say $(r_U,k_U)$ and $(r_V,k_V)$, are defined in the intersection $U\cap V$. It is clear that, whenever $U$ and $V$ have a different center, each $x\in U\cap V$ has either a neighborhood homeomorphic to $\bigstar_2$ or the intersection is empty. Hence in the non-disjoint case $U\cap V$ is homeomorphic to a disjoint union of open intervals where the change-of-coordinates is given by $\varphi\circ \psi^{-1}: \psi(U\cap V) \to \varphi(U\cap V)$.\medskip\\
A family $\FA$ of charts on a topological graph is called an atlas if the charts from $\FA$ cover $\FX$ and the change-of-coordinates maps are $\CC^\infty$. Below we will introduce a special atlas, where the intersection of two charts is connected.\medskip\\
By definition it is easy to see that a topological graph is locally path-connected and hence connectedness and path-connectedness agree. Throughout this chapter we will assume that the underlying topological graph is connected. Furthermore, it follows from the  definition that we can also choose a base of the topology consisting of relatively compact sets, hence the topological space is locally compact.\medskip\\
We continue with some examples.
\begin{example}
The sets $\IR$, $\IR_+$, $\IS=\{x\in \IR^2 \mid |x|=1\}$ and each interval inherit the structure of a topological graph and all points which are not boundary points have a neighborhood homeomorphic to $\bigstar_2$, whereas boundary points have neighborhoods homeomorphic to $\bigstar_1$.\\
Consider the set $\FX_{\IZ^2}:=\{(x,y)\in \IR^2 \mid \exists z\in \IZ: x=z \mbox{ or } y=z\}$. Then each point $(x,y) \in \FX_{\IZ^2}$ with integer coordinates has a neighborhood homeomorphic to $\bigstar_4$, whereas all other points have one homeomorphic to $\bigstar_2$. We can easily extend this example to $\IZ^d$.\\ As an example which leads to a more general class of topological graphs one may take the union of the two sets $\{(x,y)\in \IR^2 \mid y= x^2 -1\}$ and $\{(x,y)\in \IR^2 \mid y= -x^2+1\}$. Note that one has to careful with examples derived by the union of two graphs of real-valued functions. For instance, take the union of $\{(x,y)\in (0,\infty)\times \IR \mid y=\sin \frac{1}{x} \}$ with the closed semi-axis $[0,\infty)$. Then the point $(0,0)$ has no neighborhood homeomorphic to a star-graph. However, excluding this point gives a topological graph.
\end{example}
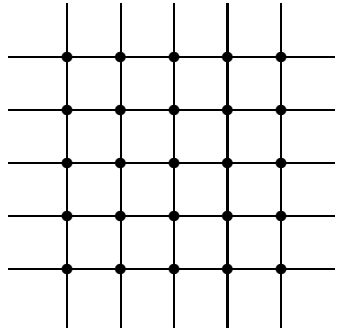
\begin{figure}[h]
\setlength{\unitlength}{2pt}
\begin{center}
\begin{picture}(60,60)(-12,-8)
\multiput(0,0)(0,10){5}{ \multiput(0, 0)(10, 0){5}{
  \put(0,0){\line(0, 1){10}}
  \put(0,0){\line(1, 0){10}}
  \put(0,0){\circle*{2}}
  \put(0,-1){\line(0, -1){10}}
  \put(-1,0){\line(-1, 0){10}}}}
\end{picture}
\end{center}
\caption{The metric graph $\FX_{\IZ^2}$}
\end{figure}
In the preceding examples all topological graphs were somehow generated by the union of curves. A different way of generating topological graphs will be discussed in detail, below. There, the idea is to start with a discrete graph and to consider the edges as intervals.\medskip\\
We next introduce some more topological concepts and notations.
\begin{definition}
We call a connected subset $\FY \subset \FX$ equipped with the trace topology a (topological) subgraph.
\end{definition}
We now turn our attention to the combinatorial aspects of topological graphs. Before doing so we recall basic definitions and facts of discrete graphs.
\begin{reminder}\index{discrete graph}
Given a countable set $V$, a subset $E \subset V\times V \setminus \{(v,v)\mid v\in V\}$ is called an edge set if $(v,w)\in E $ implies $(w,v)\in E$. In other words, $E$ defines a symmetric relation on $V$. The pair $(V,E)$ is referred to as a discrete graph. If $(w,v)\in E$ we say that $u,w$ are adjacent and denote this by $w\sim v$. If an edge $e\in E$ is given by $(v,w)$ we may also write $e\sim v$ and say $e$ is incident to $v$. The number of different edges incident to a fixed vertex is called its vertex degree. If this number is finite for all vertices, then the graph is called locally finite. To be precise, a graph defined in such a way is usually called a simple graph, since multiple edges and loops are excluded.
\end{reminder}
Our next task is to define a graph structure on a certain subset of a given topological graph. There is a natural choice of this, though as we will see, this set will be too small in general.
\begin{definition}\index{star graph!star-shaped chart}\index{star graph!degree}\index{vertex}
Let $\FX$ be a topological graph and $x\in\FX$.
\begin{itemize}
\item[(a)] If $x$ has a $d$-star-shaped chart, we call the number $d$ the degree of $x$ and denote it by $\Deg(x)$.
\item[(b)] The elements of the set $\FV_0:=\{ x\in \FX \mid \Deg (x)\neq 2\}$ are called proper vertices. Any closed, discrete set $\FV$ which contains $\FV_0$ is called a vertex set.
\item[(c)] For a fixed vertex set $\FV$ and $x\in \FX$ consider the set $\FU_\FV (x)$ consisting of all $U \subset \FX$ open, such that there is a $d$-star-shaped chart $(U,\varphi)$ with center $x$ and  $U\setminus\{x\}\cap \FV = \varnothing$. Then the subgraph
\[\bigstar_{\FV}(x) := \bigcup_{U\in \FU_\FV(x)} U\]
is called the maximal star-shaped neighborhood of $x\in \FX$, where we omit the subscript $\FV$ in general.
\end{itemize}
\end{definition}
From the definition of a topological graph, it is easy to see that the set of proper vertices is at most countable and has no accumulation points. But, as already noted, the set of proper vertices will in general be too small. This will become clear in the next proposition, as we are interested in vertex sets with certain properties.
\begin{prop}
There exists a vertex set $\FV$ with the following properties.
\begin{itemize}
\item[(i)] For all $x,y\in \FV$, $x\neq y$ we have either $\bigstar(x)\cap\bigstar(y) = \varnothing$ or $\bigstar(x)\cap\bigstar(y)$ is connected.
\item[(ii)] For all connected components of $\FX\setminus \FV$ the boundary consists of two elements of $\FV$.
\end{itemize}
\end{prop}
\begin{proof}
The proof will be constructive. If we have $\FV_0 = \varnothing$, then the topological graph is homeomorphic to $\IS$ or $\IR$, depending on whether it is compact or not. This follows by covering arguments, as each neighborhood of an arbitrary $x\in \FX$ is homeomorphic to an interval and by using that the graph is assumed to be connected, Hausdorff and second countable. In this case the choice of $\FV$ is obvious.\\
Assume $\FV_0 \neq \varnothing$. Denote the set of all connected components of $\FX\setminus \FV_0$ by $\FE_0$, then we can represent this set as disjoint union
\[\FE_0 = \FE \cup \FE_m \cup \FE_l \cup \FE_r,\]
where
\begin{eqnarray*}
\FE   &=& \{ \FY\in \FE_0 \mid \partial \FY = \{x,y\} \mbox{ and } \bigstar_{\FV_0}(x)\cap\bigstar_{\FV_0}(y)= \FY\}\\
\FE_m &=& \{ \FY\in \FE_0 \mid \partial \FY = \{x,y\} \mbox{ and } \bigstar_{\FV_0}(x)\cap\bigstar_{\FV_0}(y)\neq \FY\}\\
\FE_l &=& \{ \FY\in \FE_0 \mid \partial \FY = \{x\} \mbox{ and } \bar{\FY} \mbox{ compact}\}\\
\FE_r &=& \{ \FY\in \FE_0 \mid \partial \FY = \{x\} \mbox{ and } \bar{\FY} \mbox{ not compact}\}.
\end{eqnarray*}
Note that for $\FY\in \FE_r$ its closure is homeomorphic to $\IR_0^+$. Fix an atlas $\FA$ on $\FX$ such that $\FE_0\subset\FA$, and denote for $\FY\in \FE_0$ the associated homeomorphism to $(0,1)$ by $\pi_\FY$. We define a selection mapping $\Phi$ on $\FE_0$ as follows:
\[\Phi(\FY) :=
\begin{cases}
\partial\FY & ,\FY\in \FE\\
\partial\FY \cup \varphi_\FY^{-1}(\{\frac{1}{2}\}) &, \FY \in \FE_m\\
\partial\FY \cup \varphi_\FY^{-1}(\{\frac{1}{3}\},\{\frac{2}{3}\}) & ,\FY \in \FE_l\\
\partial\FY \cup \tilde{\varphi}_\FY^{-1}(\IN) & ,\FY \in \FE_r
\end{cases},\]
where $\tilde{\varphi}_\FY$ denotes the homeomorphism to $\IR_0^+$.\\
We now set $\FV = \Phi(\FE_0)$. By construction this set satisfies $(i)$ and $(ii)$.
\end{proof}
The set $\FV$ in the proposition is chosen such that all components of $\FX\setminus \FV$ have the desired properties
\begin{itemize}
\item[(i)] that for each $x,y\in \FV$ there exists at most  one connected component $\FY$ of $\FX\setminus\FV$ with $\{x,y\}\in \overline{\FY}$, and
\item[(ii)] each connected component has exactly two boundary points.
\end{itemize}
The set $\FV_0$ will be too small in general, since the set of multiple edges $\FE_m$, the set of loops $\FE_l$ and the set of rays $\FE_r$ will not be empty in general. This is essentially the reason why we have to enlarge the set of vertices.
\begin{definition}\index{vertex!proper}
We call a specific choice of a vertex set proper, if it satisfies $(i)$ and $(ii)$ of the previous proposition.
\end{definition}
Having fixed a proper vertex set of a given topological graph, we say $x,y\in \FV$ are in relation, if $\bigstar(x)\cap \bigstar(y)\neq \varnothing$. This defines a graph structure on $\FV$, and we denote this graph by $\Gamma_\FX = (V,E)$. The so constructed graph will be simple, due to the choice of $\FV$. Furthermore, we consider the special atlas $\FA_\FX = \FV_\FX \cup \FE_\FX$, where $\FV_\FX=\{\bigstar (x) \mid x\in \FV\}$ and $\FE_\FX$ is the set of connected components of $\FX\setminus \FV$. Since vertices can be identified with their maximal star-shaped neighborhoods and edges with connected components, we see that the atlas $\FA_\FX$ encodes all information on topology and graph structure. For this reason, we will denote this atlas also by $\Gamma_\FX$ to put emphasis on this fact and call the elements of $\FV_\FX$ vertex-stars and the elements of $\FE_\FX$ edges. The mapping $\FX\mapsto \Gamma_\FX$ will be called a graph representation of $\FX$.\medskip\\
The next question is, given a discrete graph how can we construct a topological graph from it. The idea is obvious: we consider each edge as an interval and glue them together at the endpoints according to the graph structure. In particular, given a discrete graph $\Gamma =(V,E)$, we associate the continuous edge  $\FX_e = \{e\}\times (0,1)$ to each edge $e\in E$ and consider
\[\FX_\Gamma = \bigcup_{e\in E} \FX_e \cup V.\]
In order to define an appropriate topology, fix an orientation on $\Gamma$, i.e. a map $(\partial^+,\partial^-):E \to V\times V$, such that $e=(\partial^+(e),\partial^-(e))$ where $\partial^+(e)$ is called the initial vertex and $\partial^-(e)$ the terminal vertex. Define the function $\Pi: \bigcup\limits_{e\in E} \{e\} \times [0,1] \to \FX_\Gamma$ as
\[\Pi(x) := \begin{cases} x &, \mbox{ if } x\in \bigcup\limits_{e\in E} (\{e\}\times (0,1))\\
 \partial^+(e) &, \mbox{ if } x = (e,0) \mbox{ for some } e\in E \\
 \partial^-(e) &, \mbox{ if } x = (e,1) \mbox{ for some } e\in E \end{cases}\]
and denote the induced equivalence relation by $\stackrel{\Pi}{\sim}$, i.e. for $x,y\in \bigcup\limits_{e\in E} \{e\} \times [0,1] $
\[ x\stackrel{\Pi}{\sim} y :\Leftrightarrow \Pi(x) = \Pi(y).\]
Via this relation we can identify $\FX_\Gamma$ with the quotient space $\Bigl(\bigcup\limits_{e\in E} \{e\}\times [0,1] \Bigr) /_{\stackrel{\Pi}{\sim}}$. Note that the canonical projection onto the equivalence classes agrees with $\Pi$. We equip the latter space with the quotient topology, i.e. the finest topology such that the canonical projection $\Pi:x\mapsto [x]$ onto the equivalence classes is continuous. Note that the topology on $\FX_\Gamma$ is independent of the choice of a specific orientation, since two different orientations define the same equivalence relation.
\begin{lemma}
Let $\Gamma = (V,E)$ be a locally finite discrete graph. Then the space $\FX_\Gamma$ equipped with the quotient topology is a topological graph.
\end{lemma}
\begin{proof}
In order to show the defining properties of a topological graph, we construct a base of the quotient space. A set $U\subset \FX_\Gamma$ is open, if $\Pi^{-1}(U)$ is open in the disjoint union. It suffices to consider all open sets such that there is at most one vertex in it, since this gives a base of the topological space. Let $x\in \FX_e \subset\FX_\Gamma$ for some $e\in E$, let $U$ be a neighborhood of $x$ such that $U\subset \FX_e$. Then $\Pi^{-1} (U) \subset \{e\}\times (0,1)$ is connected and therefore homeomorphic to $\bigstar_2$.\\
For vertices $x\in V \subset \FX_\Gamma$, consider an open neighborhood $U$ of $x$ such that $U \subset \bigstar(x) = \{x\}\cup \bigcup\limits_{e\sim x} \FX_e$. Then $\Pi^{-1} (U)$ is the disjoint union $\bigcup\limits_{e:\partial^+(e)=x} \{e\}\times [0,r_e) \cup \bigcup\limits_{e:\partial^-(e)=x} \{e\}\times (r_e,1]$ (for some numbers $r_e$ depending on $U$), which is open in the disjoint union $\bigcup\limits_{e\in E} \{e\}\times [0,1]$. Enumerate the edges as $e_1,\dots,e_{\Deg(x)}$ and define the mapping $\varphi: \bigstar(x) \to \bigstar_{\Deg(x)}$ by
\[\varphi(y)=
\begin{cases} 0& ,y= x\\
t(\cos(k\frac{2\pi}{\Deg(x)}),\sin(k\frac{2\pi}{\Deg(x)})) & ,y = (e_k,t) \mbox{ and } \partial^+(e_k)=x\\
(1-t)(\cos(k\frac{2\pi}{\Deg(x)}),\sin(k\frac{2\pi}{\Deg(x)})) & ,y = (e_k,t) \mbox{ and } \partial^-(e_k)=x
\end{cases}.\]
This mapping is one-to-one, and, as shown above, it is a homeomorphism.\\
Is is clear that $\FX_\Gamma$ has a countable base and that it is Hausdorff, thus we are done with the proof.
\end{proof}
Due to the previous lemma we have a mapping $\Gamma \mapsto \FX_\Gamma$ and we call the topological graph $\FX_\Gamma$ the associated topological graph to $\Gamma$. We are now in a position to state the main theorem of the first section, which is a compatibility result of the two constructions from above and gives a representation of a topological graph as a $CW$-complex. The interested reader in $CW$-complexes is referred to \cite{M}.
\begin{theorem}
Let $\FX$ be a topological graph, then it is homeomorphic to $\FX_{\Gamma_\FX}$. Moreover $\FX$ is a $CW$-complex, that is $\FX$ is a Hausdorff space equipped with a subspace $\FV$ satisfying
\begin{itemize}
\item[(i)] $\FV$ is a discrete, closed subspace of $\FX$,
\item[(ii)] $\FX\setminus\FV$ is the disjoint union of open subsets $\FX_e$, where each $\FX_e$ is homeomorphic to an open interval,
\item[(iii)] for each $\FX_e$, its boundary $\partial \FX_e$ is a subset of $\FV$ consisting of two points, and in particular the pair $(\bar{\FX}_e, \FX_e)$ is homeomorphic to $([0,1],(0,1))$,
\item[(iv)] $\FX$ is equipped with the so-called weak topology, that is a subset $\FY\subset \FX$ is closed (open) if and only if $\FY\cap \bar{\FX}_e$ is closed (open) for all edges $\FX_e$.
\end{itemize}
\end{theorem}
\begin{proof}
The proof follows from the remark that we can identify the set $\bigstar (x)$ with $\{x\}\cup \bigcup\limits_{e\sim x} \FX_e$ as in the previous proof. This gives a base of the topology. The properties of the $CW$-complex are obvious.
\end{proof}
The theorem gives us the opportunity to translate topological properties into graph theoretical properties. For instance, a topological graph is compact if and only if the discrete graph is finite.\medskip\\
From now on, whenever we have a topological graph, we denote it by $\FX_\Gamma=(\FV_\FX,\FE_\FX)$, to put emphasis that we have chosen a proper vertex set which is associated to the discrete graph $\Gamma_\FX$ and which serves at the same time as the special atlas from above encoding the topological graph.

\section{Weights, metrics and combinatorics}
In this section we introduce several functions defined on the product of the vertex set of a topological graph with itself and investigate their meaning. Recall that each $\tilde{\FV} \supset \FV$ gives a proper vertex set, such that weights may have a rather general meaning.\medskip\\
Let $\FX_\Gamma$ be a topological graph with vertex set $\FV$.
\begin{definition}\index{weight}
A symmetric function $w:\FV\times \FV \to [0,\infty)$ is called a weight over~$\FV$.
\end{definition}
By definition we can associate a graph to $\FV$ if we set $x\sim y$ whenever $w(x,y)>0$. As vague interpretation, we have in mind that $w(x,y)$ measures the distance between $x,y$ in a certain sense. Recall from graph theory, that a combinatorial path is a finite sequence of vertices $(x_1,\dots,x_n)$ such that $x_k\sim x_{k+1}$. The number $n$ is then called the combinatorial length of the path. We define the combinatorial metric with respect to a weight of two vertices $x,y\in \FV$ to be the infimum of the combinatorial lengths of all paths connecting $x,y$. This defines a metric on the set of vertices whenever the graph associated to the weight is connected, that is each pair of vertices can be connected by a path.\medskip\\
Turning back to topological graphs, recall that a famous result of Urysohn says, that each regular Hausdorff space which has a countable base is metrizable. Hence, by definition each topological graph is metrizable. Let $d$ be a metric that induces the topology of the graph. Then given an edge $\FX_e$ in $\FX$, its length is defined as
\[L(\FX_e) :=  \sup_{\FZ} \sum_{x_k\in \FZ} d(x_k,x_{k+1}),\]
where the supremum is taken over all partitions $\FZ$ of the set $\FX_e$. Thus we have a mapping $l:\FV\times \FV \to [0,\infty)$, $l(x,y) = L(\FX_e)$ when $e=(x,y)$ and $l(x,y)=0$ otherwise. This gives rise to the following definition.
\begin{definition}\index{weight!edge}
A weight $w:\FV\times \FV \to [0,\infty)$ is called an edge weight if $w(x,y)~\neq~0$ if and only if $x\sim y$. Therefore we will write $w(e)$ for $w(x,y)$ if $e=(x,y)$.
\end{definition}
From the discussion above, we see that $l(x,y)$ defines an edge weight. Conversely, given an edge weight we can define a metric on $\FX_\Gamma$. This is done by introducing a special atlas which involves an edge weight. We first introduce vertex and edge coordinates. Let $l$ be an edge weight and let $\Gamma =(V,E)$ be the associated graph of $\FX_\Gamma$. We will use $V$ and $E$ as index sets below. We start with the description of vertex coordinates. For $x\in \FV$ take the set $\bigstar(x) \in \FV_{\FX_\Gamma}$. Then as seen above we have $\bigstar(x)=\{x\}\cup\bigcup\limits_{e\sim x} \FX_e$. Let $(e_k)_{i=1}^{\Deg(x)}$ be an enumeration of the edges $e\sim x$. Denote by $\varphi_x$ the homeomorphism from $\bigstar(x)=\{x\}\cup\bigcup\limits_{i=1}^{\Deg(x)} \FX_{e_i}$, defined by
\[\varphi_x(y) := \begin{cases} (0,0) &, \mbox{ for } y=x\\ l(e_k) t (\cos k \tfrac{2\pi}{\Deg(x)},\sin k \tfrac{2\pi}{\Deg(x)})&,\mbox{ for } y=(e_k,t) \end{cases},\]
to the set
\[\{r_k(\cos k \tfrac{2\pi}{\Deg(x)},\sin k \tfrac{2\pi}{\Deg(x)})\in \IR^2 \mid 0\leq r_k<l(e_k), k=1,\dots, \Deg(x)\},\]
which is equipped with the trace topology coming from $\IR^2$,
In the latter set we can use polar coordinates, i.e. each point is given by the pair $(r,k)$ where $r$ denotes the distance from $(0,0)$ and $k$ relates to the angle $(\cos k \tfrac{2\pi}{\Deg(x)},\sin k \tfrac{2\pi}{\Deg(x)})$. With a little abuse of notation we denote this set also by $\bigstar(x)$. It is clear that the homeomorphism depends on the enumeration of the incident edges. Thus the change-of-coordinate mapping for two homeomorphisms corresponding to the same vertex, but with different enumerations of the incident edges, gives a permutation of the edges as set in $\IR^2$. Let now $y\sim x$, denote by $e$ the unique edge connecting $x$ and $y$. Then for two homeomorphisms $\varphi_x$ and $\varphi_y$ associated to $x$ and $y$ the change-of-coordinate mapping is simply given by
\[\varphi_x \circ \varphi_y^{-1}: ((0,l(e)),k_1) \to ((0,l(e)),k_2),\quad (t,k_1)\mapsto (l(e) - t, k_2),\]
where $k_1,k_2$ are coming from the enumeration of the edges. Thus given $z\in \FX_e$ we have $\varphi_x(z)= (t,k_1)$ and $\varphi_y(z)=(l(e)-t,k_2)$ for some $t\in (0,l(e))$.\medskip\\
Secondly, we use vertex-coordinates for the introduction of edge coordinates. As vertex coordinates depend on an enumeration of the edges, edge coordinates will depend on an orientation. Fix an orientation, let $\FX_e\in \FE_{\FX_\Gamma}$ and let $\varphi_{\partial_+(e)}$ be a homeomorphism corresponding to $\bigstar(\partial_+(e))$ as above. We then define the mapping $\varphi_e: \FX_e \to (0,l(e))$ for $x\in \FX_e$ by
\[x\mapsto \varphi_{\partial_+(e)}|_{\FX_e}(x) = (t,k) \mapsto t,\]
and therefore $\varphi_e$ is independent from the choice of enumeration of edges for $\varphi_{\partial_+(e)}$. In order to get that the mappings $(\phi_e)_{e\in E}$ cover $\FX_\Gamma$, we extend each $\phi_e$ to a mapping from $\FX_e\cup\{\partial^+(e),\partial^-(e)\}$ to $[0,l(e)]$ by setting $\phi_e(\partial^+(e))=0$ and $\phi_e(\partial^-(e))=l(e)$. We assume from now on that the special atlas, introduced in the previous section, is equipped with the homeomorphisms defined above.\medskip\\
We get now back to the introduction of a metric using an edge weight $l$. For this reason we need the following definition.
\begin{definition}\index{path}\index{path!graph representation}
A subgraph $\Fp_x^y$ is called a path connecting $x,y\in \FX$ if there exists a homeomorphism $\varphi: \Fp_x^y \to [0,1]$, such that $\varphi(x)= 0$ and $\varphi(y)=1$. A graph representation of a path is given by the ordered set
$(p_0, p_1, \dots, p_{n-1}, p_n)$, where $p_0=x$, $p_n=y$ and $\{p_1,\dots, p_{n-1}\} = \FV \cap p_x^y$ with the order $p_i< p_j$ whenever $\varphi(p_i)< \varphi(p_j)$.
\end{definition}
It is clear that the closure of each edge defines a path. Moreover let $\FX_\Gamma$ be connected, then given $x,y\in \FX_\Gamma$ there exist $p_1,\dots,p_{n-1}$ with $p_{j} \sim p_{j+1}$, $p_1\in \overline{\bigstar(x)}$ and $p_{n-1}\in \overline{\bigstar(y)}$ such that
\[\Fp_x^y := \bigcup_{j=1}^n \Fp_{p_{j-1}}^{p_j}\]
defines a path from $x$ to $y$ with graph representation $(p_0,\dots,p_n)$. Now given an edge weight $l:\FV\times \FV \to [0,\infty)$ the length of a path connecting $x$ and $y$ with graph representation $(p_0,\dots,p_n)$ is defined as
\[L(\Fp_x^y)= \sum_{k=1}^n |\varphi_{e_k} (p_k) - \varphi_{e_k}(p_{k-1})|\]
where $e_k$ denotes the unique edge with $p_{k-1},p_k \in \FX_{e_k}$ and the mappings $\varphi_e$ are the homeomorphisms associated to $\FX_e$ constructed above. We define the metric as
\[d_l(x,y) := \inf \{L(\Fp) : \mbox{ $\Fp$  connects $x$ and $y$}\}.\]\index{metric!path metric}
By connectedness and local finiteness, we see that this defines a metric on $\FX_\Gamma$. A particular choice of an edge weight $l$ will be called metrization and a graph equipped with a metrization will be called a metric graph. Note that a given metric and the metric defined by weights coming from that metric are in general different from each other. If it is clear from the context which metrization is chosen we will suppress the subscript in general and just write $d$.\medskip\\
As in each metric space we can define the closed ball with center $x\in \FX_\Gamma$ radius $r>0$ to be the set
\[B_r(x)=\{y\in \FX_\Gamma \mid d(x,y)\leq r\}\]
and the associated sphere
\[S_r(x)=\{y\in \FX_\Gamma \mid d(x,y)= r\}.\]
Note that closed balls need not to be compact in general. This will also be discussed below. As both are defined intrinsically via the metric, graphs will be nice if the graph structure and the metric structure have a certain symmetry relation. This gives rise to the following definition.
\begin{definition}\index{vertex!radial vertex set}
Let $\FX_\Gamma$ be a metric graph and fix a root $o\in \FX_\Gamma$. The vertex set is called radial with respect to $o$, if there exists a strictly increasing sequence $(r_n)_n\subset (0,\infty)$ with $r_o=0$ and such that for all $n\in \IN$
\[S_{r_n}(o) \subset \FV \]
and
\[\bigcup_{n\in \IN_0} S_{r_n}(o) =\FV.\]
\end{definition}
Is is clear that there are graphs which are not radial. However, sometimes it is possible to obtain a radial graph by inserting improper vertices.\medskip\\
If $\FX_\Gamma$ is a radial graph with root $o\in \FX_\Gamma$, we can equip the graph with a canonical orientation. Let $e\in E$ be given, $e=(x,y)$, we then say $x$ is the initial vertex if $d(o,x) < d(o,y)$, otherwise we say $y$ is the initial vertex. For $n\in \IN_0$ we say a vertex $x\in \FV$ belongs to the $n$-th generation if $x\in S_{r_n}(o)$, and for $n\in \IN$ we say an edge belongs to the $n$-th generation if $\partial^+(e) \in S_{r_n}(o)$. It is clear that for all $r>0$ the set $S_r(o)$ is countable, and thus the number $\# S_r(o)$ is well defined. We immediately obtain that for $r=r_n$ the number $\#S_r(o)$ equals the number of vertices of the $n$-th generation, and for $r\in (r_n,r_{n+1})$ the number $\#S_r(o)$ equals the number of edges of the $n$-th generation. Thus the function $\#S_r(o)$ contains important combinatorial information. We can also use this function to obtain the degree as
\[\deg(x)= \lim_{r\to 0} \#S_r(x).\]
Moreover, given a root $o\in \FX_\Gamma$ we can split the vertex degree into the number of edges pointing to and pointing away from a point $x\in \FX_\Gamma$, i.e. the inward degree with respect to $o$ is defined as \index{star graph!degree!inward} \index{star graph!degree!outward}
\[\deg_-(x) := \lim_{r\to 0} \# (B_{d(o,x)}(o) \cap S_r(x))\]
and the outward degree with respect to $o$ is defined as
\[\deg_+(x) := \lim_{r\to 0} \# (B_{d(o,x)}(o)^\complement \cap S_r(x)).\]
As a first consequence we get for all $x\in \FX_\Gamma$
\[\deg(x)= \deg_-(x) + \deg_+(x),\]
and as a second one we obtain for $r\in (r_n,r_{n+1})$ that
\[\# S_r(o) =\sum_{x\in S_{r_n}(o)} \deg_+(x) = \sum_{x\in S_{r_{n+1}}(o)} \deg_-(x).\]
Both quantities give rise to further classifications of radial graphs.
\begin{definition}\index{graph!spherically symmetric}
Let $\FX_\Gamma$ be a radial metric graph with root $o$. If $\deg_-(x)$ and $\deg_+(x)$ are constant on the spheres $S_r(o)$, then we call $\FX_\Gamma$ a spherically symmetric graph. In this case, we write $\deg(r)$, $\deg_+(r)$ and $\deg_-(r)$. If furthermore $\deg_-(x)=1$ on $S_{r_n}(o)$, then we call $\FX_\Gamma$ a radial tree.
\end{definition}
Spherically symmetric metric graphs are analogues of model manifolds, see \cite{Gr-99} and spherically symmetric graphs, see \cite{Wo-09}. For spherically symmetric graphs we immediately obtain the relation
\[\# S_r(o) = \# S_{r_n}(o) \deg_+(r_n) = \# S_{r_{n+1}}(o) \deg_-(r_{n+1})\]
for $r\in (r_n,r_{n+1})$. In the next section we will also treat edge weights coming from a measure and vice versa. We finish our discussion with another kind of weights, viz., the jump and the killing weights. Their meaning will become clear in the upcoming chapters.
\begin{definition}\index{weight!jump}
A weight $j:\FV\times \FV \to [0,\infty)$ is called a jump weight, if for all $x\in\FV$ the sum $\sum\limits_{y\in\FV} j(x,y)$ is finite.\\
A weight $k:\FV\times \FV \to [0,\infty)$ is called a killing weight, if $k(x,y)= 0$ for all $x\neq y$, i.e. it is supported on the diagonal of $\FV\times \FV$. We will also write $k(x)$ instead of $k(x,x)$.
\end{definition}
From the definition, the jump weight is independent of the graph structure and as we have noted above, it defines a graph structure on $\FV$ itself. In what follows, all graph-theoretical notions with respect to that graph will have the prefix jump.\medskip\\
If we are again in the radial setting, we see that the length of two edges of the same generation have to be equal, i.e. the length weights are radial as well. To adapt our notation, for an edge $e=(x,y)$ of the $n$-th generation we write
\[l_n := l(x,y)\]
and we have
\[l_n = r_n-r_{n-1}.\]
For general weights we impose the following definition.
\begin{definition}\index{weight!killing}
A weight $\omega: \FV\times \FV \to [0,\infty)$ on a radial metric graph with root $o\in \FX_\Gamma$ is called radial with respect to $o$ if
\[d(x_1,o)=d(x_2,o) \mbox{ and } d(y_2,o)=d(y_2,o) \Longrightarrow \omega(x_1,y_1)=\omega(x_2,y_2).\]
In this case we set for $d(x,o)=r_n$ and $d(y,o)=r_m$
\[\omega(x,y)= \omega_{n,m}.\]
\end{definition}
We now turn back to the metric which was introduced above by edge weights. By definition this metric is a so-called length metric, roughly speaking the metric is defined by the length of distance minimizing paths. For a more general study of length spaces we refer to \cite{BBI}. In what follows we will draw our attention on metrizations which are complete. A famous result in Riemannian geometry relating metric completeness with what is called geodesic completeness is known under the name of Hopf-Rinow theorem. This was generalized to length spaces, see \cite{BBI}, and says that metric completeness is equivalent to every closed metric ball being compact. We will use this to get a characterization of completeness in terms of the length function via rays. A subgraph $\Fp$ is called a ray if there exists a homeomorphism $\varphi:\Fp\to [0,\infty)$ such that for all compact subsets $K\subset \FX_\Gamma$ we have $K^\complement \cap \Fp \neq \varnothing$. As for finite paths, there exists a graph representation in terms of an infinite sequence of vertices $(p_0,p_1, \dots)$ with $p_k\sim p_{k+1}$.
\begin{prop}\index{theorem!Hopf-Rinow}
The metric space $(\FX_\Gamma, d_l)$ is complete if and only if for every ray $\Fp = (p_0,p_1,\dots)$ we have
\[\sum_{k=0}^\infty l(p_k,p_{k+1}) = \infty.\]
\end{prop}
\begin{proof}
Let $\FX_\Gamma$ be complete and assume that there exists a ray with finite length, i.e. $\sum\limits_{k=0}^\infty l(p_k,p_{k+1}) < \infty $. Since $d_l(p_k,p_{k+1}) \leq l(p_k,p_{k+1})$ we get that for $m_1,m_2\in \IN$ sufficiently large
\[d_l(p_{m_1},p_{m_2}) \leq \sum_{k=m_1}^{m_2-1} l(p_k,p_{k+1}) < \epsilon\]
and thus $(p_k)_k$ is a Cauchy sequence. By completeness there exists $p\in \FX_\Gamma$ such that $p_k \to p$. Note that for all $x\in \FX_\Gamma$ there exists $\epsilon>0$ such that $B_\epsilon(x)\setminus \{x\} \cap \FV =\varnothing$, giving that small neighborhoods of points of $\FX_\Gamma$ cannot contain infinitely many vertices. Thus $p\not\in \FX_\Gamma$ which contradicts the completeness assumption.\\
For the other direction assume that all rays have infinite lengths. We will show that all closed metric balls $B_r(x) = \{ y\in \FX_\Gamma \mid d_l(x,y) \leq r\}$ are actually compact for all $x\in \FX_\Gamma$ and $r>0$. The Hopf-Rinow theorem for length-spaces yields that the space $\FX_\Gamma$ is complete. So let $B_r(x)$ be such a closed metric ball, then either the number of edges having non-empty intersection with $B_r(x)$ is finite or infinite. In the first case we immediately obtain that the ball is compact. By constructing a ray with finite length, we show that the second case cannot happen . To do so we assume that $r := \sup \{R>0\mid B_R(x) \mbox{ is compact}\}$ is finite, as otherwise all closed balls are compact. Note that small balls lead to star graphs which are compact, and thus $r>0$. Let $(r_k)\subset \IR_+$  with $r_k \nearrow r$ and denote by $B_k$ the ball with center $x$ and radius $r_k$. Since all the balls $B_k$ are compact its boundary consists of finitely many points. Associate to each such boundary point a shortest path connecting it with $x$ and denote by $\FP_k$ this set of paths. Note that not each path in $\FP_{k-1}$ has a continuation to a path in $\FP_k$. We thus denote by $\FP_{k-1}^1$ the set of all paths in $\FP_{k-1}$ which have a continuation to a path in $\FP_k$. Moreover this set has to be non-empty. Note once more that not each path in $\FP_{k-1}^1$ has a continuation to a path in $\FP_k^1$, so that we define the set $\FP_k^2$ in an obvious manner. Continuing in this way we get for each $m\in \IN$ a family of paths $\FP_k^m$ in the ball $B_k$. Since the balls $B_k$ are compact there exists for all $k\in \IN$ a number $m_k\in \IN$ such that $\FP_k^m = \FP_k^{m_k}$ for all $m> m_k$. This gives that for given $k_1\in \IN$ each path in $\FP_{k_1}^{m_{k_1}}$ can be extended to a path in $\FP_{k_2}^{m_{k_2}}$ for arbitrary $k_2 > k_1$. Thus we define a ray as
\[\Fp = \bigcup \Fp_i\]
with $(\Fp_i)_i$ defined as $\Fp_i \in \FP_i^{m_i}$ and $\Fp_i=\Fp_{i+1}$ on $B_i$. By what we have proven such a ray exists. Thus $L(\Fp)= \lim\limits_{i\to \infty} L(\Fp_i) = \lim\limits_{i\to \infty} r_i = r$. By construction the ray $\Fp$ consists of infinitely many edges, and since the $\Fp_i$ are shortest paths we obtain
\[\sum_{e\in \Fp} l(e) = L(\Fp) = r < \infty\]
contradicting the assumption.
\end{proof}
If we deal with incomplete spaces we need the concept of completion.
\begin{reminder}\index{graph!completion}
For every metric space $M$ there exists a complete metric space $\widehat{M}$ which contains $M$ as dense subspace. This completion has the following universal property. If $N$ is a complete metric space and $f:M\to N$ is a uniformly continuous function,  then there exists a unique uniformly continuous function $\widehat{f}: \widehat{M}\to N$ which extends $f$. The space $\widehat{M}$ is determined up to isometry by this property.
\end{reminder}\index{boundary|see{graph completion}}
The completion of $\FX_\Gamma$ can be very difficult to handle. The interested reader is referred to \cite{Ge-10} and \cite{Ca-98}. So assume that the metric graph is not complete with respect to a metric and denote by $\widehat{\FX_\Gamma}$ its completion. Then the set
\[\partial \FX_\Gamma := \widehat{\FX_\Gamma} \setminus \FX_\Gamma\]
is called the (Cauchy-) boundary of the graph.
The boundary points can also be characterized in terms of rays, which was done in \cite{Ge-10}. We call two rays of finite lengths equivalent if there is a third ray of finite length that meets both rays infinitely often. The boundary of the graph can now be defined as the set of equivalent rays of finite length. This is an extension of our completeness criterion.\medskip\\
We finish this section with another piece of notation. This will be of interest in the next chapter when dealing with extensions of functions.
\begin{definition}\index{boundary!limit}
Let $\Fp\subset \FX_\Gamma$ be a ray and $\varphi:\Fp\to [0,\infty)$ the associated homeomorphism. We say that $u$ has a limit along the path if the function $u|_\Fp := u \circ \varphi^{-1}$ has this limit at infinity, and we write
\[\lim_\Fp u(x) := \lim_{t\to \infty} u|_\Fp (t).\]
\end{definition}
\section{Function spaces}
This section is devoted to the definition and study of certain function spaces on a metric graph $\FX_\Gamma$. As general principle, we try to generalize function spaces on intervals to the whole graph such that the restriction to an edge agrees with the usual function space, whereas functions on the whole graph have similar properties as the functions on the edges.\medskip\\
A function $u:\FX_\Gamma\to \IR$ could be written in vertex or edge coordinates due to the maximal atlas $\FA_{\FX_\Gamma}$ equipped with the homeomorphisms from the last section.
\begin{definition}\index{coordinates!vertex}
Let $u:\FX_\Gamma\to \IR$ be given and fix for each vertex an enumeration of its incident edges. The vertex-coordinates of $u$ are then given by the vector $(u_x)_{x\in \FV}$, where $u_x:\bigstar(x)\to \IR$,
\[u_x (r,k) := u(\varphi_x^{-1} (r,k)).\]\index{coordinates!edge}
Let $u:\FX_\Gamma\to \IR$ be given and fix an orientation of the edges. The edge-coordinates of $u$ are then given by the vector $(u_e)_{e\in E}$, where $u_e: \FX_e\to \IR$
\[u_e(t) := u ( \varphi_e^{-1} (t)).\]
\end{definition}
From the discussion on vertex coordinates in the last section we see that we have
\[u_x(t,k_1) = u_y(l(e)-t,k_2)\]
whenever $x\sim y$ and $k_1,k_2$ correspond to enumeration of $e=(x,y)$ in both vertex neighborhoods. Similarly, for two different orientations of an edge the change-of-coordinate mapping is simply given by $t \mapsto l(e)-t$.
\begin{remark}
In what follows, everything if not otherwise mentioned will be independent from the particular choice of an orientation or enumeration. Thus we will from now on assume that we have chosen one, without explicitly mentioning it.
\end{remark}\index{continuity}
We will use vertex- and edge-coordinates to get a better understanding of continuity, measurability and differentiability of functions defined on $\FX_\Gamma$. We start with a characterization of continuity in terms of vertex and edge coordinates. A function $u:\FX_\Gamma\to \IR$ is continuous if and only if $u_x:\bigstar(x)\to \IR$ is continuous for all $x\in \FV$. In particular $u_x$ is continuous if and only if for all $e\in \FE$ the edge components $u_e$ are continuous and the limits $\lim\limits_{t\to\varphi_e(x)} u_e(t)$ are equal for all edges incident to $x$, where $t\to\varphi_e(x)$ means $t\to 0$ if $\partial^+(e)=x$ and $t\to l(e)$ if $\partial^-(e)=x$.\\
We denote the space of continuous function by $\CC(\FX_\Gamma)$. Given $u\in \CC(\FX_\Gamma)$ we denote its support, i.e. the closure of the set $\{x\in\FX_\Gamma \mid u(x)\neq 0\}$, by $\supp u$ and by $\CC_c(\FX_\Gamma)$ the space of continuous function having compact support.
\begin{remark}
The topology on $\FX_\Gamma$ can also be defined as the initial topology of the set of functions which are continuous on each edge and directional continuous at each vertex, that is the limits towards each vertex agree for all incident edges. As $\CC(\FX_\Gamma)$ can be considered as a closed subspace of $\bigoplus\limits_{e\in E} \CC([0,l(e)])$, we can interpret this as a one-to-one correspondence between certain closed subspaces of $\bigoplus\limits_{e\in E} \CC([0,l(e)])$ and topological graphs.
\end{remark}
For convenience, let us introduce the following piece of notation. Given any space of functions $\CD$ with norm $\|\cdot\|_\CD$, we denote by $\CD_c$ the space of functions in $\CD$ having compact support and by $\CD_o$ the closure of $\CD_c$ with respect to the norm $\|\cdot\|_\CD$. Note that compactness means compactness coming from the chosen topology of the topological graph. Furthermore, we denote by $\CD_\loc$ the space of all functions which are locally in $\CD$, i.e. the space of all functions $u$ such that for all compact subsets $K\subset \FX_\Gamma$ there exists a function $\varphi\in\CD$ such that $u= \varphi$ on $K$. The space of bounded functions of $\CD$ will be denoted by $\CD_\infty$.
\begin{remark}
When dealing with metric graphs, as space of test functions one often chooses elements with compact support in the disjoint union of the open edges, see for instance \cite{KS-06,KPS-07,KPS-08}. We will not follow this tradition, since this space is in general too small as it ignores the structure of the graph.
\end{remark}
Given a jump weight $j$ on $\FX_\Gamma$ we define a subspace of $\CC(\FX_\Gamma)$ by
\[\CC_j(\FX_\Gamma)= \{ u\in\CC(\FX_\Gamma) \mid \forall x\in \FV : \sum_{y\in \FV} j(x,y)|u(y)| < \infty\}.\]
It is clear that $\CC_\infty(\FX_\Gamma)\subset \CC_j(\FX_\Gamma)$ for all jump weights $j$.
\begin{remark}
The space $\CC_j(\FX_\Gamma)$ equals in a certain sense the space $\tilde{F}$ defined in \cite{KL-10}. There it plays the role of the domain of definition of the formal Laplacian. Under certain conditions on $j$ (and the measure $m$ coming from $\ell^2(V,m)$) all Laplacians on $\ell^p$ are actually restrictions of $\tilde{F}$. As will be shown later, the space $\CC_j(\FX_\Gamma)$ will also be connected with the domain of certain operators, see section 3.
\end{remark}\index{integration}
Next we introduce a measure on $\FX_\Gamma$. Let $\CB(\FX_\Gamma)$ the Borel $\sigma$-algebra, i.e. the smallest $\sigma$-algebra containing all open sets of $\FX_\Gamma$. A set $\FY\subset \FX_\Gamma$ is called measurable, if for any star-shaped chart $\bigstar$ the set $\bigstar \cap \FY$ is a measurable set in $\IR^2$ with respect to the one-dimensional Lebesgue measure. Thus the family $\Lambda(\FX_\Gamma)$ of all measurable sets of $\FX_\Gamma$ forms a $\sigma$-algebra, in particular we have $\CB(\FX_\Gamma)\subset \Lambda(\FX_\Gamma)$ since by definition open sets are measurable. We define the measure on $\FX_\Gamma$ to be the pull-back of the Lebesgue measure defined on each chart. Using a partition of unity with respect to the open covering $\{\bigstar(x)\mid x\in\FV\}$ we get for a measurable function $u:\FX_\Gamma \to \IR$ that
\[\int\limits_{\FX_\Gamma} u(x) d\lambda(x) = \frac{1}{2} \sum_{x\in\FV} \int\limits_{\bigstar(x)} u_x(r,k) drdk = \frac{1}{2} \sum_{x\in \FV}\sum_{k=1}^{\Deg(x)} \int\limits_0^{l(e_k)} u_x(r,k) dr.\]
From the last expression we see that each edge is counted twice, but with different orientation. But since $\int\limits_0^{l(e)} u(r)\:dr = \int\limits_{l(e)}^0 -u(l-t) \:dt = \int\limits_0^{l(e)} \tilde{u} (r)\:dr$, we see that the integral is independent of the particular choice of orientation. Thus the integral above is given in edge-coordinates by
\[\int\limits_{\FX_\Gamma} u(x)\:d\lambda(x) = \sum_{e\in E} \int\limits_0^{l(e)} u_e(t) dt.\]
We will use both expressions for an integral of a function on $\FX_\Gamma$. In particular, edge-coordinates are useful when integrating along a path, i.e. for $x,y\in \FX_\Gamma$ let $\Fp_x^y=(p_0,\dots,p_n)$ be a path connecting them, then we immediately obtain
\[\int\limits_{\Fp_x^y} u(x)d\lambda(x) = \sum_{i=1}^n \int\limits_0^{l(e_i)} u_{e_i} (t)\:dt\]
where $e_{i+1} =(p_{i}, p_{i+1})$.\medskip\\
Given this measure we define $L^p(\FX_\Gamma)$ to be the set of all measurable functions $u$ such that $|u|^p$ is integrable, and as usual for $p=\infty$ we let $L^\infty(\FX_\Gamma)$ to be the space of all essentially bounded functions. By construction the $L^p$ spaces are isometric to the direct sum of $L^p$ spaces defined on intervals and of $L^p$ spaces defined on $\bigstar(x)$ with the induced Lebesgue measure from $\IR^2$, i.e.
\[L^p(\FX_\Gamma) \simeq \bigoplus_{e\in E} L^p(0,l(e)) \simeq \bigoplus_{x\in \FV} L^p(\bigstar(x)).\]\index{graph measure}
For later purposes we consider measures which are absolutely continuous with respect to the Lebesgue measure, with density $\nu:\FX_\Gamma\to \IR$ such that $\nu|_{\FX_e}$ is constant. We denote this constant by $\nu(e)$ and note that this function $\nu:\FV\times\FV \to (0,\infty)$ defines an edge weight. We introduce the weighted $L^p$ spaces as $L^p(\FX_\Gamma,\lambda_\nu)$ which are isometric to $\bigoplus\limits_{e\in E} L^p((0,l(e)),\nu(e) d\lambda)$.\medskip\\
By definition we obtain for an edge $e\in E$
\[\int\limits_{\FX_e} d\lambda_\nu = l(e)\nu(e).\]
and for a vertex $x\in \FV$
\[\int\limits_{\bigstar(x)} d\lambda_\nu = \sum_{e\sim x} l(e)\nu(e).\]
An important function in what follows is the volume growth function
\[V_{x_o}(r) = \int\limits_{B_r(x_o)} d\lambda_\nu\]
where $x_o\in \FX_\Gamma$ is a fixed point.\medskip\\
We have seen that there are natural choices for continuity and measurability of functions, due to topology and measure theory. Unfortunately there is no natural one of differentiability or Sobolev spaces, due to the singular role played by the vertices. We will introduce them via weak derivatives and absolute continuity. A function $u:[a,b]\to \IR$ is absolutely continuous if and only if there exists a constant $C\in \IR$ and a locally integrable function $v:[a,b]\to \IR$ such that $u(x)= C + \int\limits_a^x v(t)\: dt$. The function $v$ is then also called the weak derivative of $u$ and denoted by $u'$. This is known as the fundamental theorem of calculus for absolutely continuous functions, see \cite{Leoni}. We use this as the starting point for the definition of absolute continuity for metric graphs.
\begin{definition}\index{absolute continuity}
A function $u:\FX_\Gamma\to \IR$ is called absolutely continuous if there exists a locally integrable function $v:\FX_\Gamma \to \IR$ such that for all $x,y\in \FX_\Gamma$ and all paths connecting $x,y$ we have
\[u(y) = u(x) + \int\limits_{\Fp_x^y} v\: d\lambda.\]
In this case we call $u' := v$ the weak derivative of $u$.
\end{definition}
One can show that the weak derivative for functions on $\IR$ as defined above agrees with the distributional derivative. Using the lemma of Du Bois-Reymond, one obtains that the weak derivative is unique almost everywhere. In particular the weal derivative in our setting is unique almost everywhere. \\
By definition we have that a function $u:[a,b]\to \IR$ is absolutely continuous if and only if for all $c\in (a,b)$ the restrictions  $u|_{(a,c)}$ and $u|_{(c,b)}$ are absolutely continuous and $u\in \CC(a,b)$. This extends to our case and we obtain the following proposition.
\begin{prop}
A function $u:\FX_\Gamma \to \IR$ is absolutely continuous if and only if for all $e\in E$ the function $u_e$ is absolutely continuous and $u\in \CC(\FX_\Gamma)$.
\end{prop}
\begin{proof}
The only if part follows by considering paths which are subsets of the edges. This gives that each component is absolutely continuous. To show continuity, let $x\in \FV$ and $y\in \FX_e$ with $e\sim x$ and choose an orientation such that $x$ is the initial vertex of $e$. By absolute continuity we have
\[u(x)= u_e(0) + \int\limits_0^{\phi_e^{-1}(y)} u'_e(t)dt.\]
Thus we obtain that for $y\to x$ we have $u(y)\to u(x)$ as the edge $e\sim x$ was arbitrary. To show the if part let $\Fp=(p_0,p_1, \dots,p_n)$ be a path connecting $x,y\in \FX_\Gamma$. Denote by $u_{e_k}'$ the weak derivatives of the components of $u$ on the edges $e_k=(p_k,p_{k+1})$, with the obvious modification for $p_0=x$ and $p_n=y$. Then we have by continuity of $u$
\begin{eqnarray*}
u(x)- u(y) &=& \sum_{k=0}^{n-1} u(x_k) - u(x_{k+1}) \\
 &=& \sum_{k=0}^{n-1} \int\limits_{\FX_{e_k}} u_{e_k}'(t) dt\\
 &=& \int\limits_{\Fp} u'(t)dt.
\end{eqnarray*}
\end{proof}
\section*{Notes and remarks}
The idea of the definition of a topological graph is not new. It goes back to Kuratowski (see \cite{M}). Topological graphs as defined there, are in general defined in a constructive way, within the theory of $CW$-complexes. The presentation here goes along the idea of the definition of a compact metric graph in \cite{BF-06,BR-07}. We have decided for this intrinsic definition due to the analogy of manifolds. We closely followed \cite{Gr-09} in this progress. As we have seen, a lot of concepts can be transferred to our setting. In particular the concept of vertex coordinates seems to be new. It would also be possible to introduce concepts like tangent space and tangent bundle, when defining it as a disjoint union of one dimensional vector spaces in the vertices. From the definition we see that each point in the tangent bundle has a star-shaped neighborhood, and thus it is a topological graph as well. Our aim in the second part was to define several function spaces also in an intrinsic way, using only topological, measure theoretical and almost everywhere differentiability in our definitions. The introduction of weights will become fruitful in the upcoming chapters. It unites the continuous framework with the discrete one taken from \cite{KL-10,KL-11}. Apart from that, the function spaces related to weights seem also to be new in the one-dimensional setting. Our approach also opens the way to treat topological graphs having infinite vertex degree by means of the edge measure weights. 
\chapter{Dirichlet forms on graphs}
In this chapter we introduce the notion of Dirichlet forms on graphs. We first study the strongly local case via the tool of length transformation. This gives us certain embedding theorems and corresponding results in section 2. There we also introduce a new space, which is a Hilbert space sum of weighted Sobolev spaces on intervals. This leads to a new Sobolev type inequality, Theorem 2.14, which also seems to generalize the one-dimensional Sobolev inequality. A main point of investigation is the regularity problem. This is discussed in detail in section 3. There we also develop a boundary theory for our case which allows us to relate regularity with the appearance of a boundary, see Theorem 2.22, Theorem 2.28 and Proposition 2.29. Subsequently we complete the definition of a general graph Dirichlet form and investigate these new forms. An important case is investigated, viz, the case when the non-diffusion part is relatively bounded. These results, Proposition 2.35, Lemma 2.38 and Proposition 2.40, are mainly phrased in terms of a discrete Laplacian on a discrete space equipped with a measure coming from a capacity type quantity of the set of vertices. Some of these results, Proposition 2.41 and Corollary 2.42, improve the well known trace mapping properties taken from \cite{Ku-04}. At the end of the chapter we study the problem of regularity of general graph diffusion Dirichlet forms. Here, we have a first look on harmonic functions.\medskip\\
We start with the basics from the theory of Dirichlet forms. For further reading we recommend \cite{FOT-11}.
\begin{reminder}\index{Dirichlet form}
Let $\FX$ be a locally compact, separable metric space endowed with a positive Radon measure $\nu$ with $\mathrm{supp} \ \nu =\FX$. Let $\CD$ be a dense subspace of $L^2(\FX,\nu)$ and $\CE:\CD\times\CD \to \IK$ a non-negative sesquilinear form such that $\CD$ is closed w.r.t. the energy norm $\|\cdot\|_\CE$ given by
\[\|u\|_\CE^2 := \CE(u,u) + \|u\|^2_{L^2},\]
meaning that $\CE$ is a closed form in $L^2(\FX,\nu)$. We set $\CE(u):= \CE(u,u)$ in what follows.\\
A closed form is said to be a Dirichlet form if it is Markovian, that is if $\CD$ is stable under normal contractions, i.e. for all $T:\IK\to \IK$ with the property that $T(0)=0$ and $|T(x) - T(y)| \leq |x-y|$ for all $x,y\in \IK$ and for all $u \in \CD$ we have
\[T\circ u \in \CD \mbox{ and } \CE(T\circ u) \leq \CE(u).\]
In the real case it is known that the latter holds for all such $T$ if and only if it holds for the contraction $0\vee u \wedge 1$. In what follows we will only consider the real case.
\end{reminder}\index{Dirichlet form!strongly local}
A Dirichlet form is called strongly local if $\CE(u,v)=0$ whenever $u$ is constant on a neighborhood of the support of $v$. We first focus on strongly local Dirichlet forms and sketch the most important features. The interested reader may consult \cite{FOT-11} section 3.2. It can be shown that for each bounded $u\in \CD(\CE)$ there exists  a unique non-negative Radon measure $\mu_{\langle u \rangle}$ satisfying
\[\int\limits_{\FX_\Gamma} f d\mu_{\langle u \rangle} = 2\CE(uf,u) - \CE(u^2,f)\]
for all $f\in \CD(\CE)\cap \CC_c(\FX_\Gamma)$. This $\mu_{\langle u \rangle}$ is called the energy measure. Moreover we define the signed Radon measure $\mu_{\langle u,v \rangle}$ for bounded $u,v \in \CD(\CE)$ by polarization
\[\mu_{\langle u,v \rangle} =\frac{1}{2} (\mu_{\langle u+v \rangle} - \mu_{\langle u \rangle}- \mu_{\langle v \rangle}).\]
The form $\CE$ can then be written as\index{Dirichlet form!energy measure}
\[\CE(u,v) = \frac{1}{2}\mu_{\langle u,v \rangle} (\FX_\Gamma)=\frac{1}{2}\int\limits_{\FX_\Gamma} d\mu_{\langle u,v \rangle}.\]
A consequence of the strong locality of $\CE$ is the strong locality of its energy measure, i.e. if $U$ is an open set in $\FX$ on which the function $\eta \in \CD(\CE)$ is constant, then
\[\mathbf{1}_U d\mu_{\langle \eta,u\rangle}=0,\]
for any $u\in \CD(\CE)$. Another remarkable fact is that the energy measure satisfies the Leibniz rule, i.e. for bounded $u,v,w \in \CD(\CS)$ we have
\[d\mu_{\langle uv,w \rangle} = u d\mu_{\langle v,w \rangle} + v d\mu_{\langle u,w \rangle}.\medskip\]
As can be shown, see for instance \cite{FOT-11}, one can associate a stochastic process to a regular Dirichlet form. A stochastic process associated to a strongly local Dirichlet form is a diffusion. For this reason we will call a strongly local Dirichlet form also a diffusion Dirichlet form.
\section{The diffusion part}
In this section we introduce the diffusion part of what we will call a graph Dirichlet form later. Let $\FX_\Gamma$ be a metric graph as defined in the previous section.
\begin{definition}\index{graph diffusion Dirichlet form}
A graph diffusion Dirichlet form on $\FX_\Gamma$ with edge weights $(a(e))_{e\in E}$ and $(b(e))_{e\in E}$ is the form $\CE:\CD(\CE)\times \CD(\CE) \to \IR$ with domain
\[\CD(\CE) := \{u\in \CC(\FX_\Gamma) \mid u\in L^2(\FX_\Gamma, \lambda_a),\ u' \in L^2(\FX_\Gamma, \lambda_b)\}\]
and
\[\CE(u,v) = \int\limits_{\FX_\Gamma} u'v' d\lambda_b = \sum_{e\in E} \int\limits_0^{l(e)} u_e'(t) v_e'(t) b(e) dt.\]
\end{definition}
Our first task is, to show that this in fact defines a Dirichlet form. The main part is actually its closedness. This will follow immediately from the fact that
\[\CD(\CE) = \CC(\FX_\Gamma)\cap \bigoplus_{e\in E} W^{1,2}((0,l(e)),a(e)d\lambda,b(e)d\lambda),\]
where $W^{1,2}((0,l(e)),a(e)d\lambda,b(e)d\lambda)$ denotes the weighted Sobolev space consisting of all functions $u\in L^2((0,l(e)),a(e)d\lambda)$ such that the weak derivative $u'$ belongs to $L^2((0,l(e)),b(e)d\lambda)$. The equality is trivial, and we will show in the next section that the right hand side is a closed subspace of the direct sum of the one dimensional Sobolev spaces.\medskip\\
Using the one-dimensional chain rule for Lipschitz functions, see \cite{Leoni}, we immediately obtain that the form $\CE$ is Markovian and we see from the definition of $\CE$ that this form is strongly local with energy measure
\[d\mu_{\langle u\rangle}  =  |u'|^2 \: d\lambda_b.\]
From the definition above we see that three families of parameters enter the form $\CE$, viz, the edge lengths $l(e)$, the weights $a(e)$ associated with the measure of the $L^2$ space and the weights $b(e)$ which play the role of the ellipticity of the form. We will show that it is enough to study situations where basically only two of those three families of parameters are involved.\\
An important tool will be a transformation of the graph which induces a transformation of the corresponding Dirichlet spaces. Let $\FX_\Gamma$, $\widetilde{\FX}_\Gamma$ be metric graphs with the same sets of vertices and edges and with the same orientation of the edges. Denote the respective edge lengths by $l(e)$ and $\widetilde{l}(e)$. We define the length transformation \index{length transformation} $\Phi:\FX_\Gamma \to \widetilde{\FX}_\Gamma$ as
\[\Phi(x) = \begin{cases} x &, \mbox{ for } x\in \FV\\ \widetilde{\varphi}_e^{-1} (\frac{\widetilde{l}(e)}{l(e)} \varphi_e(x)) &. \mbox{ for } x\in \FE_e\end{cases}.\]
In particular this induces a transformation of functions $u:\FX_\Gamma \to \IR$ to functions $\widetilde{u}: \widetilde{\FX}_\Gamma \to \IR$ by
\[\widetilde{u}(\Phi(x)) = u(x).\]
Applying the substitution and chain rule on each edge separately we get
\[ \int\limits_{\FX_\Gamma} |u(x)|^2 d\lambda_a(x) = \int\limits_{\widetilde{\FX}_\Gamma} |\widetilde{u} (x)|^2 d\lambda_{\widetilde{a}}(x)\]
and
\[ \int\limits_{\FX_\Gamma} |u'(x)|^2 d\lambda_b(x) = \int\limits_{\widetilde{\FX}_\Gamma} |\widetilde{u}' (x)|^2 d\lambda_{\widetilde{b}}(x)\]
with
\[ \widetilde{a}(e) = \frac{a(e) l(e)}{\widetilde{l}(e)}\]
and
\[\widetilde{b}(e) =  \frac{b(e) \widetilde{l}(e)}{l(e)}.\]
Thus, the image Dirichlet form $\widetilde{\CE}$ induced by $\Phi$  on $L^2(\widetilde{\FX}_\Gamma, \lambda_{\widetilde{a}})$ is given as
\[ \widetilde{\CE}(\widetilde{u},\widetilde{v}) := \int\limits_{\widetilde{\FX}_\Gamma} \widetilde{u}'\widetilde{v}' d\lambda_{\widetilde{b}}.\]
In particular the Dirichlet spaces $\CE$ and $\widetilde{\CE}$ are equivalent in the sense that the $L^2$-norm, the energy-norm and the $L^\infty$-norm are preserved. We will be interested in transformations $\Phi$ such that the new metric gives certain information of the graph. This leads to certain normal forms, where we introduce the first one in the following definition.
\begin{definition}\index{canonical representation}\index{metric!canonical}
Let $\FX_\Gamma$ be a metric graph and $(\CE, \CD(\CE))$ a graph diffusion Dirichlet form with weight $(b(e))_{e\in E}$ on $L^2(\FX_\Gamma,\lambda_a)$. The image Dirichlet form associated with the length transformation $\Phi$ such that
$\widetilde{b}(e)=1$ is called the canonical representation of $\CE$. The associated measure $\lambda_{\widetilde{a}}$ is called the canonical measure and the associated edge weight $(\widetilde{a}(e))_{e\in E}$ is denoted by $(\nu(e))_{e\in E}$. We will also write $\nu$ for the measure $\lambda_\nu$. The associated length $\widetilde{l}$ is called the canonical length and denoted by $(l_c(e))_{e\in E}$.
\end{definition}
From the discussion above and the definition we obtain the following corollary.
\begin{coro}
Let $\FX_\Gamma$ be a metric graph and $(\CE, \CD(\CE))$ a graph diffusion Dirichlet form on $L^2(\FX_\Gamma,\lambda_a)$. Then the canonical measure is given by $\nu(e)=a(e) b(e)$ and the canonical length by $l_c(e)=\frac{l(e)}{b(e)}$.
\end{coro}
It is clear that a length transformation does not change the topology of a graph. However the metric properties can be changed dramatically. We illustrate this with the following example.
\begin{example}
We consider $\IR$ as the graph coming from $\IZ$, i.e. the vertex set is given as the set of integers and two vertices are connected if their Euclidean distance is $1$. Let $\CE$ be a graph diffusion Dirichlet form on $\IR$ with edge weights $(a(n,n+1))_{n\in \IN_0}$ and $(b(n,n+1))_{n\in \IN_0}$, then the canonical length of each edge $(n,n+1)$ is given as $\frac{1}{b(n,n+1)}$. The length of the path from $0$ to $\infty$ in the canonical metric is given as
\[\sum_{n=0}^\infty \frac{1}{b(n,n+1)}.\]
Thus, in the canonical length the graph might become incomplete.
\end{example}
\begin{remark}
Another normal form is given when setting all lengths to $1$. More precisely the image Dirichlet form has measure weights $a(e)l(e)$ and ellipticity weights $\frac{b(e)}{l(e)}$. This normal form will be not considered in the rest of this thesis, but we remark that the induced metric corresponds to the discrete metric when restricted to the set of vertices. In particular it gives a continuous version of the discrete metric.
\end{remark}
\section{Embedding theorems}
In this section we consider the behavior of functions at infinity and study important embedding properties. A key ingredient will be the following pre-version of the well known Poincaré inequality. In the whole section we will work with the canonical representation. Recall that a subgraph is a connected topological subspace equipped with the trace topology.
\begin{lemma}[Path Poincaré inequality]
Let $\FY \subset \FX_\Gamma$ be a compact, connected subgraph and $u:\FX_\Gamma\to \IR$ be absolutely continuous. Then for all $u^* \in u(\FY)$ there exists $c\in \FY$ such that for all $x\in \FY$ and each path $\mathfrak{p}_c^x$ connecting $c$ and $x$, we have
\[|u(x)-u^*| \leq \int\limits_{\mathfrak{p}_c^x} |u'(t)| \: dt.\]
Note that we can choose $u^*=\bar{u} = \nu(\FX_\Gamma)^{-1} \int u \:d\nu$.
\end{lemma}
\begin{proof}
Since $u$ is continuous it attains its maximum and minimum on $\FY$, which we denote as  $u(x_+)= \max u$ and $u(x_-) = \min u$. Thus we have $u(\FY)= [u(x_-), u(x_+)]$. In particular $u(x_-)-u^* \leq 0$ and $u(x_+) - u^* \geq 0$. Without restriction let the inequalities be strict, otherwise we choose the point $c$ below as $x_-$ or $x_+$ resp. Since the graph is connected, there is a path $\mathfrak{p}_{x_-}^{x_+}$ connecting $x_-$ and $x_+$. Note that $u-u^*$ restricted to this path is continuous and changes the sign, hence there exists $c\in \mathfrak{p}_{x_-}^{x_+}$ such that $u(c)= u^*$. Take now $x\in \FY$ arbitrary and a path $\mathfrak{p}_x^c$ connecting $c$ and $x$. Then since $u$ is absolutely continuous we have
\[|u(x)- u^*| \leq \int\limits_{\mathfrak{p}_c^x} |u'(t)| \: dt.\qedhere\]
\end{proof}
The path Poincaré inequality is of big use in what follows. From the point of view of Sobolev spaces in metric measure spaces this lemma says that the weak derivative is an upper gradient, see for instance \cite{HK-00}. A remarkable consequence is the following proposition.
\begin{prop}
Let $u\in \CD(\CE)$, then $u$ is $\frac{1}{2}$-Hölder continuous with respect to the canonical metric. In particular we have
\[|u(x) - u(y)|^2 \leq d_c(x,y) \CE(u).\]
\end{prop}
\begin{proof}
Let $x,y\in \FX_\Gamma$, then the path Poincaré inequality yields
\[|u(x)- u(y)|^2 \leq L(\Fp_x^y) \CE(u).\]
Taking the infimum over all paths connecting $x,y$ yields the claim.
\end{proof}
Note that we do not use that $u\in L^2$. We discuss this in the end of this section.\\
The path Poincaré inequality yields the desired Poincaré and Sobolev inequalities. The canonical diameter $\diam(\FY)$ of a subgraph $\FY$ of $\FX_\Gamma$ is defined as the diameter of the metric space $(\FY,d)$
\[\diam(\FY) := \sup\{d(x,y)\mid x,y\in \FY\},\]
with $d$ being the metric coming from the canonical length. We can now come to one of the most important inequalities in analysis, viz, the Poincaré inequality.
\begin{theorem}[Poincaré inequality]\index{inequality! Poincaré}
Let $\FX_\Gamma$ be a metric graph and let $\FY\subset \FX_\Gamma$ be a compact subgraph. Then for any $p\in [1,\infty)$, $q\in [1,\infty]$ we have for all absolutely continuous functions $u:\FX_\Gamma \to \IR$
\[\Bigl(\frac{1}{\nu(\FY)}\int\limits_\FY |u - \bar{u}|^q d\nu \Bigr)^\frac{1}{q} \leq \diam(\FY)^\frac{p-1}{p} \Bigl( \int\limits_\FY |u'|^p d\lambda\Bigr)^\frac{1}{p}.\]
\end{theorem}
\begin{proof}
Let $y\in \FY$, then using the previous lemma there is a path $\Fp$ such that
\begin{eqnarray*}
|u(y) - \bar{u} | &\leq& \int\limits_\Fp |u'| d\lambda\\
&\leq& L(\Fp)^\frac{p-1}{p} \Bigl(\int\limits_\Fp |u'|^p d\lambda\Bigr)^\frac{1}{p}.
\end{eqnarray*}
As $L(\Fp) \leq \diam(\FY)$ we obtain the desired inequality by integrating over $\FY$ with respect to $\nu$.
\end{proof}
\begin{remark}
In the case of $q=\infty$ the left hand side of the Poincaré inequality is given as $\|u-\bar{u}\|_\infty$.
\end{remark}
The importance of the Poincaré inequality is already expressed in the following results.
\begin{coro}\index{inequality!Sobolev}
Let $\FX_\Gamma$ be a metric graph and let $\FY\subset \FX_\Gamma$ be a compact subgraph. Then for any $p\in [1,\infty)$, $q\in [1,\infty]$ we have for all absolutely continuous functions $u:\FX_\Gamma \to \IR$ and all $x\in \FY$
\[|u(x)| \leq \nu(\FY)^{-\frac{1}{q}}\Bigl(\int\limits_\FY |u|^q d\nu\Bigr)^\frac{1}{q} + \diam(\FY)^{\frac{p-1}{p}}\Bigl(\int\limits_\FY |u'|^p d\lambda\Bigr)^\frac{1}{p}.\]
\end{coro}
\begin{proof}
The inequality follows from the Poincaré inequality by choosing $q=\infty$, the reverse triangle inequality and by noting that Hölder's inequality implies $\bar{u}\leq (\nu(\FY)^{-1} \int |u|^q d\nu)^\frac{1}{q}$.
\end{proof}
The last inequality is known as Sobolev inequality. Note that in the above statements the choice of $p$ and of $q$ are independent of each other. This implies the following embedding theorem.
\begin{theorem}\index{theorem!Sobolev embedding!local version}
Let $\FX_\Gamma$ be a metric graph and $\CE$ a graph diffusion Dirichlet form as above. Then the embedding
\[\CD(\CE) \hookrightarrow \CC(\FX_\Gamma)\]
is continuous, where the right hand side is equipped with the topology of uniform convergence on compact subsets.\\
Furthermore, for all compact subgraphs $\FY\subset \FX_\Gamma$ the mapping
\[\cdot|_\FY : \CD(\CE) \to \CC(\FY),\quad u\mapsto u|_\FY\]
is compact.
\end{theorem}
\begin{proof}
The continuous embedding follows from Sobolev's inequality with $p=2$ and $q=\infty$. Thus we are left to show the compact embedding, which will be done by using the Arzelà-Ascoli theorem. Let therefore $(u_n)$ be a bounded sequence in $\CD(\CE)$. By the Sobolev inequality this sequence is also uniformly bounded in $\CC(\FY)$, and thus the sequence is pointwise bounded in $\CC(\FY)$. Let now $x,y\in \FY$, we obtain again by the path Poincaré inequality
\[|u_n(x)-u_n(y)| \leq d_c(x,y)^\frac{1}{2} \CE(u_n)\]
giving that the sequence is equicontinuous. Now we can apply the Arzelà-Ascoli theorem and we are done.
\end{proof}
\begin{remark}
In the theorems above we only used that the $L^p$ mean of $u$ with respect to $\nu$ on $\FY$ is some value in $u(\FY)$. We will use this fact in a later section to derive analogous statements.
\end{remark}
We can now state and prove that the form $(\CE,\CD(\CE))$ is closed.
\begin{theorem}
Let $\FX_\Gamma$ be a metric graph. Then the space
\[\CC(\FX_\Gamma) \cap \bigoplus_{e\in E} \{u\in L^2((0,l_c(e)),d\nu)\mid u'\in L^2((0,l_c(e)),d\lambda)\}\]
is a Hilbert space. Consequently, $\CD(\CE)$ is a Hilbert space with respect to the energy norm $\|\cdot\|_\CE^2 = \|\cdot\|_2^2 + \CE(\cdot)$.
\end{theorem}
\begin{proof}
Since the direct sum is a Hilbert space by definition, we only need to show that the intersection with the continuous functions is a closed subspace. But this follows since Cauchy-sequences in this space are locally uniform Cauchy sequences from the previous theorem, and thus by the completeness of the direct sum, there is a limit function. But this limit function has to be continuous by uniform convergence on compact subsets.
\end{proof}
Above we have shown the continuous embedding of $\CD(\CE)$ into the space of continuous functions. An important observation is, that this is actually an embedding of $\CD_\loc(\CE)$ into $\CC(\FX_\Gamma)$. The following example shows that in general $\CD(\CE)$ is not continuously embedded into $\CC_\infty(\FX_\Gamma)$.
\begin{example}
Let $\FX_\Gamma=\IR^+_0$ with $\FV = \IN$ and let $\nu(n,n+1)$ be such that $\sum\limits_{n=0}^\infty \nu(n,n+1) < \infty$. Then the constant function $\mathbf{1}$ is in $\CD(\CE)$ and the functions defined by
\[\mathbf{1}_n(x) := \begin{cases} 1&, x\in [0,n]\\ -x+n+1 &, x\in [n,n+1] \\ 0&, x\in [n+1,\infty) \end{cases}\]
are in $\CD(\CE)$ for all $n\in \IN$ as well. Moreover an easy calculation shows $\mathbf{1}_n \to \mathbf{1} $ in $\CD(\CE)$ whereas $\|\mathbf{1} - \mathbf{1}_n\|_\infty =1$. Thus $\CD(\CE)$ is not continuously embedded in $\CC_\infty(\FX_\Gamma)$.
\end{example}
Now we introduce a subspace of $\CC(\FX_\Gamma)$ which has the property that $\CD(\CE)$ is embedded into it.
\begin{definition}
Let $\FX_\Gamma$ be a metric graph and $(a(e))_{e\in E}$ an edge weight. Then $\ell^2(\FX_\Gamma,a)$ is the space of all continuous functions satisfying
\[\|u\|_\nu := \Bigl(\sum_{e\in E} \lambda_a(\FX_e) \|u_e\|^2_\infty\Bigr)^\frac{1}{2} < \infty.\]
\end{definition}
Using the canonical form we have
\[\lambda_a(\FX_e)= a(e) l(e) = \nu(e)l_c(e).\]
One has to be careful with the symbol $\ell^2$ here. Note that it is not a typical small $L^2$ space. We have chosen this symbol, as the space is a Hilbert space sum.
\begin{prop}
The mapping $\|\cdot\|_\nu$ defines a norm on $\ell^2(\FX_\Gamma, \nu)$ and the space  $(\ell^2(\FX_\Gamma, \nu),\|\cdot\|_\nu)$ is a Banach space. Furthermore the embedding
\[\ell^2(\FX_\Gamma, \nu) \hookrightarrow L^2(\FX_\Gamma,\nu)\]
is continuous.
\end{prop}
\begin{proof}
It is straightforward to show that $\|\cdot\|_\nu$ defines a norm. To show completeness let $(u_n)_n$ be a $\|\cdot\|_\nu$ Cauchy sequence. Then by definition this sequence is a uniform Cauchy sequence on each compact subset of $\FX_\Gamma$. Thus there exists a continuous function $u$ on $\FX_\Gamma$ such that $(u_n)$ converges uniformly on each compact subset to $u$. Furthermore using the lemma of Fatou we obtain
\begin{eqnarray*}
\|u-u_n\|^2_\nu &=& \sum_{e\in E} \nu(e)l_c(e) \|u_e-(u_n)_e\|^2_\infty \\
&=& \sum_{e\in E} \nu(e)l_c(e) \|\lim_{m\to \infty} (u_m)_e-(u_n)_e\|^2_\infty \\
&\leq& \liminf_{m\to \infty} \sum_{e\in E} \nu(e)l_c(e) \|(u_m)_e-(u_n)_e\|^2_\infty \\
&\leq & \epsilon
\end{eqnarray*}
for $\epsilon>0$ arbitrary and $n$ large enough. Thus $u_n\to u$ in $\|\cdot\|_\nu$ norm and hence $u\in \ell^2(\FX_\Gamma, \nu)$. The second claim follows easily by observing
\[ \int\limits_{\FX_e} |u_e (t)|^2 \:\nu(e)dt \leq \|u_e\|^2_\infty l_c(e) \nu(e).\qedhere\]
\end{proof}
The significance of this space lies in the fact proven below that $\CD(\CE)$ is mapped continuously into $\ell^2(\FX_\Gamma,\nu)$. However, for our purposes we need the following assumption.
\begin{assumption}\index{assumption!intrinsic length}
Assume from now on that $\nu(e) l_c(e)^2 \in (0,1]$.
\end{assumption}
This assumption has a nice interpretation in terms of the intrinsic metric. This will be introduced in the next section. So far let us only mention that it is not a strong assumption as it could be achieved by inserting additional vertices of degree $2$.
\begin{theorem}\index{theorem!Sobolev embedding!$\ell^2(\nu)$ version}
Let $\FX_\Gamma$ be a metric graph and $\CE$ a graph diffusion Dirichlet form with canonical measure $\nu$. Then the embedding
\[\CD(\CE) \hookrightarrow \ell^2(\FX_\Gamma,\nu)\]
is continuous.
\end{theorem}
\begin{proof}
The Sobolev embedding applied to each edge gives
\[\|u_e\|^2_\infty \leq C(\tfrac{1}{(\nu(e)l_c(e))}  \|u_e\|^2_2 + l_c(e) \|u'\|_2^2).\]
After multiplication with $(\nu(e)l_c(e))$ and using that $\nu(e)l_c(e)$  is uniformly bounded from above we derive the claim by summing over $e\in E$.
\end{proof}
We discuss now, how this theorem is related to the original Sobolev embedding. To this aim, consider the graph $\FX_\Gamma=\IR$ with measure weights $a(n,n+1)$ and $b(n,n+1)$ being identically one. By the fundamental theorem of calculus one obtains for $u\in W^{1,2}(\IR)$ and any $x\in \IR$ that
\[u(x) - u(0) = \int\limits_0^x u'(t)dt.\]
Thus the limit for $u$ as $x\to \infty$ exists whenever the right hand side has a finite limit. Applying this to $u^2$ we get
\begin{eqnarray*}
|u^2(x) - u^2(0)| &\leq& 2 \int\limits_0^\infty |u(t) u'(t)| dt\\
&\leq & \|u\|_2^2  + \|u'\|_2^2.
\end{eqnarray*}
Thus we have that the limit $u(x)$ for $x\to \infty$ exists and equals zero, as $u\in L^2$. Summarizing we get that $W^{1,2}(\IR)$ is continuously embedded into $\CC_o(\IR)$. On the other hand we have shown that $W^{1,2}(\IR)$ is continuously embedded into $\ell^2(\IR,1)$. This gives a better embedding since
\[\|u\|_\infty^2 \leq \sum_{n\in \IZ} \|u|_{[n,n+1]}\|_\infty^2.\]
In this sense our result improves the usual Sobolev inequality. On the other hand, we let $(l(n))_n$ be a sequence with $\sum l(n) = \infty$ and choose as vertex set for $\IR$ the numbers $r(n)= \sum_{k=1}^n l(k)$. If we assume that the edge weights $a$ and $b$ are identically $1$ again, then the associated Dirichlet space is the same as above. Our Sobolev embedding gives that the Sobolev space $W^{1,2}(\IR)$ is embedded into $\ell^2(\FX_\IZ,\nu)$ with $\nu(\FX_e)=l(e)$. However, this space also contains unbounded functions and thus the embedding theorem is worse than the original. Nevertheless, we also see that constant functions do not belong to $\ell^2(\FX_\IZ,\nu)$ and we conclude that unbounded functions in $\ell^2(\FX_\IZ,\nu)$ have to oscillate strongly. We can thus improve the embedding by embedding the Sobolev space $W^{1,2}(\IR)$ into the subspace of $\ell^2(\FX_\IZ,\nu)$ of Hölder continuous functions, as the the Hölder condition gives a restriction on oscillation behaviour.\\
Analogous considerations also hold in the case of an interval, say $\FX_\Gamma = [0,1)$. The usual Sobolev embedding says that $W^{1,2}([0,1)) $ is continuously embedded into $\CC([0,1))$ and also it is equal to $W^{1,2}([0,1])$. Our result applied to a metrization $(l(n))_n$ of $[0,1)$ such that $\sum l(n)=1$ gives that $W^{1,2}([0,1))$ embedded into $\ell^2([0,1),\nu)$ with $\nu(\FX_e)=l(e)$. Again the latter space also contains unbounded functions.\medskip\\
We are interested when the space $\CD(\CE)$ is embedded into $\CC_\infty(\FX_\Gamma)$ or even $\CC_o(\FX_\Gamma)$. We already know that this is the case if the length and measure weights are bounded from below by a positive constant. A first result towards the general case is the following.
\begin{prop}
Let $\FX_\Gamma$ be a metric graph and $\CE$ a graph diffusion Dirichlet form with canonical measure $\nu$. Then $\ell^2(\FX_\Gamma, \nu)$ is continuously embedded into $\CC_\infty(\FX_\Gamma)$ if and only if there is some constant $c>0$ such that for all $e\in E$ we have $\nu(\FX_e)\geq c$.
\end{prop}
\begin{proof}
It is known that the weighted sequence space $\ell^2(E,\nu(\FX_e))$ is embedded into $\ell^\infty(E)$ if and only if $\sqrt{\nu(\FX_e)}^{-1} \in \ell^\infty$. As
\[\sup_{x\in \FX_\Gamma} |u(x)| = \sup_{e\in E} \|u_e\|_\infty\]
the claim follows.
\end{proof}
As the embedding properties are closely connected with the behavior of functions at infinity, we continue to study the behavior of functions along infinite paths.
\begin{theorem}\index{theorem!Sobolev embedding!extension version}
Let $\FX_\Gamma$ be a metric graph and $\CE$ a graph diffusion Dirichlet form with canonical measure $\nu$. Let $u\in \CD(\CE)$ and  $\Fp$ a ray. If the ray has finite canonical length, then $\lim_{\Fp} u(x)$ exists. If the ray has infinite length and additionally $\inf_{e\in \Fp} \nu(e)>0$ then the limit $\lim_{\Fp} u(x)$ exists and is zero. In particular, if for any ray one of these assumptions holds, then we have the continuous embedding
\[\CD(\CE)\hookrightarrow \CC_o(\widehat{\FX_\Gamma}),\]
where $\widehat{\FX_\Gamma}$ denotes the completion of the space $\FX_\Gamma$ with respect to the canonical metric.
\end{theorem}
\begin{proof}
We show that in both cases
\[\int\limits_{\Fp} |u'(t)| dt\]
is finite as this implies the existence of the limit of $u(x)$ along this path.\\
In the case that the path has finite length we use Cauchy-Schwarz and estimate this by
\[L_c(\Fp)^\frac{1}{2} \CE(u).\]
In the case that $\Fp$ has infinite length we use
\[|u^2(x) - u^2(y)| \leq \int\limits_\Fp |u(t)^2| \nu(t)dt + \int\limits_\Fp |u'(t)^2| \nu^{-1}(t) dt\]
and by assumption $\inf\limits_{e\in E} \nu(e) >0$ the right hand side could be estimate by $\|u\|_\CE^2$. As the measure of the path is infinite we immediately obtain that this limit is zero.\\
The embedding is obvious.
\end{proof}
We finish this section with a note on the so called resistance metric. We have actually seen from the path Poincaré inequality that for all absolutely continuous functions $u$ with $\|u'\|_2 <\infty$ there is a constant $C>0$ depending only on $x,y$ such that the estimate
\[|u(x)-u(y)|^2 \leq C \|u'\|_2^2\]
holds. Note that we do not assume $u\in L^2$. We know that $C \leq d_c(x,y)$. We set
\[R(x,y):= \sup\{ \frac{|u(x)- u(y)|^2}{\|u'\|_2^2} \mid u' \in L^2\},\]
which is known as the resistance metric, see for instance \cite{S-94}. From the estimate above we obtain
\[R(x,y) \leq d_c(x,y)\]\index{metric!resistance}
but in general the converse inequality does not hold. Among others, the previous result can be generalized by using the resistance metric instead of the canonical. This can be seen by observing that functions in $\CD(\CE)$ are $\frac{1}{2}$-Hölder continuous with respect to the resistance metric. Thus they can be continuously extended to the resistance boundary. In particular functions in $\CD_o(\CE)$ do vanish on the resistance boundary. One problem arising with the resistance metric is, that in general it is not induced by edge weights. This happens only in the case of a tree, as there it agrees with the canonical metric.
\begin{prop}
If $\FX_\Gamma$ is a tree then we have
\[R(x,y)=d_c(x,y).\]
\end{prop}
\begin{proof}
Let $x,y\in \FX_\Gamma$ then we construct a function $u$ as follows. We let $u(x)=0$, $u(y)= d_c(x,y)$ and linearly interpolated on the unique path $\Fp$ connecting $x,y$. On each connected component of $\FX_\Gamma\setminus\Fp$ we continuously extend the function $u$ by setting it constant. Then we have that $u'\equiv 0$ on the complement of the path and $u'=1$ on the path. As therefore $\CE(u)= d_c(x,y)$ we obtain
\[d_c(x,y)^2 = |u(x) -u(y)|^2 \leq R(x,y) \CE(u) = R(x,y) d_c(x,y).\]
The claim follows, as the other inequality always holds.
\end{proof}
In chapter 4 we come back to functions with finite energy functional $\CE$ but not necessarily finite $L^2$ norm.
\section{Regularity}
In this section we consider the question whether there are enough continuous functions with compact support in the domain of our Dirichlet form or not. Enough in this context means that the space
\[\CD_c(\CE) = \CD(\CE) \cap \CC_c(\FX_\Gamma)\]\index{Dirichlet form!regularity}
is uniformly dense in $\CC_c(\FX_\Gamma)$ and $\CE$-dense in $\CD(\CE)$. In the context of Dirichlet forms this property is known as regularity.\medskip\\
The uniform density in $\CC_c(\FX_\Gamma)$ follows easily by approximating a function on each edge. This is possible, since a function having compact support is non-zero only on a finite number of edges.\medskip\\
To study the $\CE$-density we introduce the space
\[\CD_o(\CE) = \overline{\CD(\CE)\cap \CC_c(\FX_\Gamma)}^{\|\cdot\|_\CE}.\]
By general principles the restriction $\CE_o$ of $\CE$ to $\CD_o(\CE)$ is again a Dirichlet form and by definition this form is regular. Thus the question of regularity is equivalent to
\[\CD(\CE)= \CD_o(\CE)\]
being true or not. In order to study this, we transform our Dirichlet space to another normal form.
\begin{definition}\index{intrinsic representation}\index{metric!intrinsic}
Let $\FX_\Gamma$ be a metric graph and $(\CE, \CD(\CE))$ a graph diffusion Dirichlet form on $L^2(\FX_\Gamma,\nu_a)$ with weights $(b(e))_{e\in E}$. The image Dirichlet form coming from the length transformation $\Phi$ such that $\widetilde{b}(e)=\widetilde{a}(e)$ is called the intrinsic representation. The associated measure $\lambda_{\widetilde{a}}$ is called the intrinsic measure and the associated weight is denoted by $\omega(e):=\widetilde{a}(e)$. We will also write $\omega$ for the measure $\lambda_\omega$. The associated length $\widetilde{l}$ is called the intrinsic length and denoted by $(l_i(e))_{e\in E}$.
\end{definition}
As we will see below, regularity is connected to completeness of the metric graph with respect to the intrinsic metric. An important observation is that the transformation $\Phi$ maps $\CC_c(\FX_\Gamma)$ to $\CC_c(\widetilde{\FX}_\Gamma)$ bijectively, and thus regularity is invariant under $\Phi$.\\
From the transformation rules connected with length transformations we see that we can always transform a graph into the intrinsic form. The next two propositions are proven by those rules.
\begin{prop}
For a graph diffusion Dirichlet form on a metric graph $\FX_\Gamma$, the intrinsic measure is given by $\sqrt{a(e)b(e)}$ and the intrinsic length is given by $l(e) \sqrt{\frac{a(e)}{b(e)}}$.
\end{prop}
As canonical and intrinsic scales are both of interest we include the following formula.
\begin{prop}
For the intrinsic and canonical scales the following relations hold true:
\begin{eqnarray*}
\omega(e)^2 &=& \nu(e)\\
l_i(e)^2&=& l_c(e)^2 \nu(e).
\end{eqnarray*}
\end{prop}
We see from the proposition that the assumption $l_c(e)^2 \nu(e)$ being uniformly bounded from above, means in the intrinsic representation that the intrinsic metric weights $l_i(e)$ are uniformly bounded.\medskip\\
The reason why we call it intrinsic comes from the theory of strongly local Dirichlet forms, see \cite{BM-95} and \cite{St-94}. For those forms one may associate a distance function to the underlying space by means of
\[\rho(x,y) := \sup\{|u(x) - u(y)|\mid u\in \CD_{\mathrm{loc}} \cap \CC(\FX_\Gamma) \mbox{ and } d\mu_{ \langle u,u\rangle} \leq dm\}.\]
A typical assumption on $\rho$ as a metric in general Dirichlet space is, that it induces the original topology, which is essential in the construction of sufficiently many cut-off functions.\medskip\\
The following lemma justifies the name intrinsic representation.
\begin{prop}
Let $\FX_\Gamma$ be a connected metric graph. Then the intrinsic metric is given by
\[\rho(x,y) = \sup \{ |u(x) - u(y)| \mid u\in \CD_{\loc}(\CE) \mbox{ and } |u'| \leq 1 \}.\]
Moreover, for all $x,y\in \FX_\Gamma$ we have
\[\rho(x,y) = d_i(x,y).\]
\end{prop}
\begin{proof}
It is clear that the intrinsic metric is given as the supremum defined above. We will only prove the second claim. Fix $y\in \FX_\Gamma$, we set $d_y (x) := d_i(x,y)$ and $\rho_y (x) := \rho(x,y)$. Note that we have $d_y \in \CC(\FX_\Gamma)$. For $x\in \FX_\Gamma$ choose $\FX_e\subset \FE$ such that $x\in \FX_e$, then we obtain
\[d_y(x) = \min\{ d_i(y,\partial^+(e)) + |\varphi_e(x) - \varphi_e(\partial^+(e))|, d_i(y,\partial^-(e)) + |\varphi_e(x) - \varphi_e(\partial^-(e))|\}.\]
Thus we have $|d_y'| = 1$ a.e. and therefore $d_y \in \CD_{\loc}(\CE) $. Now by definition of $\rho$ we have $d_y(x) \leq \rho_y (x)$ for all $x,y\in \FX_\Gamma$.\\
For the converse inequality, let $x,y \in \overline{\FX}_e$ and $u\in \CD_{\loc}(\CE) \cap \CC(\FX_\Gamma)$ an arbitrary function with $|u'| \leq 1$. Then we have
\begin{eqnarray*}
|u(x) - u(y)| & =& |u_e (\varphi_e(x)) - u_e(\varphi_e (y))|\\
&\leq &\int\limits_{\varphi_e(x)}^{\varphi_e(y)} |u_e'(t)|\: dt\\
&\leq& |\varphi_e(x)- \varphi_e(y)|,
\end{eqnarray*}
i.e. $\rho(x,y) \leq |\varphi_e(x) - \varphi_e(y)|$ on $\FX_e$. For arbitrary $x,y\in \FX_\Gamma$, choose a minimizing path $\Fp=(p_0,\dots,p_n)$ from $x$ to $y$ with $L(\Fp) = d_i(x,y)$. By triangle inequality we obtain
\begin{eqnarray*}
\rho(x,y) \leq \sum_{k=1}^{N-1} \rho(p_k,p_{k+1}) \leq \sum_{k=1}^{N-1} |\varphi_{e_k}(p_k) - \varphi_{e_k} (p_{k+1})| =  d_i(x,y),
\end{eqnarray*}
where $e_k$ denotes the unique edge with $p_k, p_{k+1} \subset \FX_{e_k}$.
\end{proof}
We can now obtain a result on regularity.
\begin{theorem}\index{graph diffusion Dirichlet form!regularity}
If $(\FX_\Gamma, d_i)$ is connected and complete, we have $\CD(\CE) = \CD_o(\CE)$. In particular, the form $\CE$ is regular.
\end{theorem}
\begin{proof}
By  regularity of $\CE$ on $\CD_o(\CE)$ and completeness of the graph with respect to the intrinsic metric we define cut-off functions as follows: fix $o\in \FX_\Gamma$ and set $\rho_n(x) := \inf \{\rho(x,y)\mid y\in B_n(o)\}$. Let $\zeta:[0,\infty)\to [0,1]$ be continuously differentiable and $\zeta(0)=1$, $\zeta \equiv 0$ on $[1,\infty)$ and $|\zeta'|\leq 2$. Then we have for $\eta_n := \zeta\circ \rho_n$, see \cite{BLS-09},
\[d\mu_{\langle \eta_n\rangle} \leq 2 d\omega.\]
Moreover, $\eta_n \in \CD(\CE)$ and $\eta_n u \in \CD(\CE)$ for all $u\in \CD(\CE)$ and $n\in \IN$. Furthermore $\eta_n u$ has compact support and is continuous by the local Sobolev embedding and we have
\begin{eqnarray*}
\CE(u-\eta_n u) &=& \int\limits_{B_n(o)^c} |(u - \eta_n u)'|^2 d\omega\\
&\leq & \int\limits_{B_n(o)^c} |u'|^2 d\omega + \int\limits_{B_n(o)^c}|\eta_n' u|^2 d\omega + \int\limits_{B_n(o)^c} |\eta_n u'|^2d\omega\\
&\leq& \int\limits_{B_n(o)^c} |u'|^2 d\omega + 4\int\limits_{B_n(o)^c} |u|^2 d\omega
\end{eqnarray*}
Since $u,u' \in L^2(\FX_\Gamma,\omega)$ the right hand side tends to zero. Similarly we obtain $\eta_n u \to u$ in $L^2(\FX_\Gamma,\omega)$.
\end{proof}
The theorem above tells us that non-regularity is connected with non-completeness with respect to the intrinsic metric. However the other direction is not true as the following example shows.
\begin{example}
Let $\FX_\Gamma=(0,1)$ and consider a metrization and a measure such that both boundary points have infinite measure, i.e. for all $\epsilon>0$ the sets $(0,\epsilon)$ and $(1-\epsilon, 1)$ have infinite measure. Using the latter one as intrinsic measure, the results below will show that this form is regular.
\end{example}
So we need to analyze intrinsic boundary points. Unfortunately it turns out that the intrinsic boundary is somehow not appropriate. One reason for this is that in general one can not expect to extend functions from the domain continuously to it. This follows as in general the intrinsic boundary and the canonical boundary are different, which can be seen as follows. Let $\FX_\Gamma$ be a metric graph and let $x$ be an intrinsic boundary point such that there are two rays having $x$ as intrinsic boundary point. Now, if we equip the graph with a canonical metric such that both rays have different boundary points, we can easily construct a continuous function which could be uniquely extended to a continuous function on the canonical completion, but which has no continuous extension to the intrinsic completion. Our plan is to analyze the intrinsic boundary via the canonical one, since for instance this extension property holds. We start with a Sobolev embedding for $\CD_o(\CE)$.
\begin{theorem}\index{theorem!Sobolev embedding!form version}
Let $\FX_\Gamma$ be a metric graph and $\CE$ a graph diffusion Dirichlet form with canonical measure $\nu(e)$. Assume that all rays have either finite canonical length or the canonical measure weights are uniformly positive on each ray. Then $\CD_o(\CE)$ is continuously embedded into $\CC_o(\FX_\Gamma)$.
\end{theorem}
\begin{proof}
We already know from the last section that the space is embedded into $\CC_o(\widehat{\FX_\Gamma})$, so it is sufficient to show that all functions $u\in \CD_o(\CE)$ vanish at the canonical boundary. But this follows easily as for all functions with compact support in $\FX_\Gamma$ we have for $y\in \partial \FX_\Gamma$ and $x\in \FX_\Gamma$
\[|u^2(x)|=|u^2(x)-u^2(y)| \leq d_c(x,y) \CE(u)\]
and thus uniform convergence along the ray forces each $u\in \CD_o(\CE)$ to vanish at $y\in \widehat{\FX_\Gamma}\setminus \FX_\Gamma$.
\end{proof}
Due to the theorem the defect of regularity can be measured in terms of the canonical boundary. However in general this one is too big. As functions in $\CD(\CE)$ can be uniquely extended to the canonical boundary we do not distinguish whether they are defined on $\FX_\Gamma$ or on its completion. The same holds for functions $u\in \CD(\CE)$ and negligibility of boundary points, what we will now introduce.
\begin{definition}\index{boundary!negligible}
Let $\FX_\Gamma$ be a metric graph and $\CE$ a graph diffusion Dirichlet form. A ray $\Fp$ is called ($\CE$-)negligible if for all $u\in \CD(\CE)$ we have
\[\lim_\Fp u(x)=0.\]
We denote the set of all negligible rays by $\partial_o\FX_\Gamma$
\end{definition}
The following proposition gives a link between negligibility and unique continuation.
\begin{prop}
Let $\FX_\Gamma$ be a metric graph and $\CE$ a graph diffusion Dirichlet form. A ray with finite length is ($\CE$-)negligible if and only if any $u\in \CD(\CE)$ can be uniquely extended to the associated boundary point of $\Fp$ by zero.
\end{prop}
The last proposition clarifies the notion of negligibility of boundary points coming from a metrization. Thus we can speak of negligible boundary points, which we associate with $\partial_o\FX_\Gamma$.\medskip\\
The reason for the definition of negligibility gives the following lemma.
\begin{lemma}
Let $\FX_\Gamma$ be a metric graph and $\CE$ a graph diffusion Dirichlet form. Then for all $u\in \CD(\CE)$ there exists a sequence $(u_n)\in \CD(\CE)$ with the property that for all negligible rays $\Fp$ the function $u_n|_\Fp$ has compact support and which converges to $u$ in $\CE$-norm.
\end{lemma}
\begin{proof}
Let $u\in \CD(\CE)$, then we know from the Sobolev embedding that $u$ is continuous. Let $\Fp$ be a  negligible ray with an associated homeomorphism $\varphi:\Fp \to [0,\infty)$. Then by continuity of $u$ there exists for all $\epsilon_n>0$ an $R>0$ such that $|u_n(x)| < \epsilon_n$ for all $x\in \varphi^{-1}(R,\infty)$. Define the function $\phi_n(u) = u - ((-\epsilon_n)\vee u \wedge \epsilon_n$. Then for all negligible rays $\Fp$ the function $\phi_n(u) |_\Fp$ has compact support. Furthermore applying Theorem 1.4.2 in \cite{FOT-11} we get that $\phi_n(u) \to u$ in $\CE$-norm.
\end{proof}
The lemma directly improves the regularity result.
\begin{theorem}
Let $\FX_\Gamma$ be a metric graph and $\CE$ a graph diffusion Dirichlet form. Assume that the intrinsic boundary is negligible, then $\CE$ is regular.
\end{theorem}
\begin{proof}
Let $u\in \CD(\CE)$ and let $(u_n)$ be the sequence from the previous lemma. Then the subgraph $\supp u_n$ is a closed set and by assumption this subgraph is intrinsically complete. Thus we can approximate $u_n$ by a sequence $u_n^k \in \CD_c(\CE)$ with support in $\supp u_n$. Now we can define a diagonal sequence which converges to $u$ in $\CE$-norm.
\end{proof}
Due to the theorem we are left with a description of negligibility.
\begin{prop}\index{boundary!characterization}
Let $\FX_\Gamma$ be a metric graph and $\CE$ a diffusion Dirichlet form. Then the following holds true.
\begin{itemize}
\item[(i)] If each neighborhood of a boundary point has infinite measure, then this point is negligible.
\item[(ii)] If an intrinsic boundary point has a neighborhood of with finite measure, then this point is not negligible.
\item[(iii)] If a canonical boundary point is not an intrinsic boundary point, then each neighborhood of this point has infinite measure and is thus negligible.
\item[(iv)] If for an intrinsic boundary point there is a path such that the intrinsic measure weights are bounded from below along this path, then it is also a canonical boundary point.
\end{itemize}
\end{prop}
Note that we know nothing in the case that an intrinsic boundary point is not a canonical boundary point -- except if each neighborhood has infinite measure. It would follow that $\inf_{e\in E} \omega(e)=0$. We know from the previous section that this case causes problems.
\begin{proof} (i) follows immediately by continuity of the functions in the domain $\CD(\CE)$. A function being nonzero there cannot be square integrable. The second claim follows by construction:  if $B_\epsilon (x)$ is an intrinsic ball with finite measure, we set $u$ as
\[u(y)= \frac{\dist(y,B_{\epsilon}(x)^\complement)}{\dist(y,B_{\epsilon}(x)^\complement) + \dist(y,B_{\frac{\epsilon}{2}}(x))}.\]
This function is absolutely continuous with $|u'| \leq \frac{C}{\epsilon}$, it is equal to $1$ in $B_{\frac{\epsilon}{2}}(x)$ and equal to $0$ outside of $B_{\epsilon}(x)$. We therefore get
\begin{eqnarray*}
\|u\|_2^2 + \CE(u) &=& \int\limits_{B_{\epsilon}(x)} |u|^2 + |u'|^2 d\omega \\
&\leq& (1+ \frac{C}{\epsilon})\omega(B_{\epsilon}(x)) < \infty
\end{eqnarray*}
by assumption on $B_{\epsilon}(x)$.\\
To prove (iii), for such a boundary point we find a path $\Fp$ with
\[\sum_{e\in \Fp} l_i(e) = \infty, \quad \sum_{e\in \Fp} \frac{l_i(e)}{\omega(e)} < \infty.\]
By Cauchy-Schwarz we get
\[ \sum_{e\in \Fp} l_i(e) = \sum_{e\in \Fp} \sqrt{l_i(e)\omega(e)} \sqrt{\frac{l_i(e)}{\omega(e)}} \leq \Bigl(\sum_{e\in \Fp} l_i(e)\omega(e) \Bigr)^\frac{1}{2} \Bigl(\sum_{e\in \Fp} \frac{l_i(e)}{\omega(e)}\Bigr)^\frac{1}{2}\]
and it follows that
\[\sum_{e\in \Fp} l_i(e) \omega(e) = \omega(\Fp) =\infty.\]
As each neighborhood of this boundary point contains a part of this path, the claim follows.\\
Finally (iv) follows easily from the Hölder inequality.
\end{proof}
A different approach to regularity was undertaken in \cite{GM-12}. There, regularity was measured in terms of the polarity of Cauchy-boundary of the non-complete manifold. Note that the arguments take over to our situation.\medskip\\
In a later section we will study embedding and regularity properties from a different point of view, viz, by reduction to the problem of a certain discrete graph. However this requires different techniques, which we will have to develop.\medskip\\
Till now it was sufficient to consider a connected metric graph. From  the next section on we will consider a countable disjoint union of metric graphs. It is clear that this still defines a topological graph, but not a metric space anymore since we also have $d(x,y)=\infty$ for some $x,y\in \FX_\Gamma$. For convenience we will still call such a disconnected graph a metric graph. As one easily checks, the results of a local nature in the previous section remain true, while the global ones need the additional assumption that the diameter of the connected components are uniformly bounded from below. Moreover by a diagonal sequence argument one can also show that an unconnected form is regular if and only if all its components are.
\section{The discrete part}
In this section we want to introduce the discrete part. We start with a reminder on general Dirichlet forms.
\begin{reminder}\index{Dirichlet form!Beurling-Deny formulae}
Given a general regular Dirichlet form $\CS$ on $L^2(\FX_\Gamma,\omega)$ the celebrated Beurling-Deny formula says that for $u,v\in \CC_c(\FX_\Gamma)\cap \CD(\CS)$ the form $\CS$ can be expressed as
\[\CS(u) = \CE(u) + \int\limits_{\FX_\Gamma\times \FX_\Gamma \setminus d} (u(x)-u(y))^2J(dx,dy) + \int\limits_{\FX_\Gamma} u(x)^2 K(dx),\]
where
\begin{itemize}
\item $\CE$ is a symmetric form on  $\CD(\CS)\cap\CC_c(\FX_\Gamma)$ which has the strong local property, i.e. $\CE(u,v)=0$ for all $u,v\in \CD(\CS)\cap \CC_c(\FX_\Gamma)$ such that $v$ is constant on a neighborhood of $\supp u$,
\item $J$ is a symmetric positive Radon measure on the product space $\FX_\Gamma\times \FX_\Gamma$ off the diagonal $d$,
\item $K$ is a positive Radon measure on $\FX_\Gamma$.
\end{itemize}
Furthermore $\CE$, $J$ and $K$ are uniquely determined by $\CS$.
\end{reminder}
In the previous sections we have intensively studied the strongly local part. To go one step towards the Beurling-Deny formula we will add a jump and a killing part to the previously defined strongly local form. They will be coming from weighted sums of $\delta$-measures on $\FX_\Gamma\times \FX_\Gamma$ and $\FX_\Gamma$ respectively, corresponding to a discrete graph. Therefore we continue with a review of regular Dirichlet forms on discrete sets. The interested reader may consider \cite{KL-10,KL-11} and related articles. Later we will also study connections of our setting to certain discrete Dirichlet forms.
\begin{reminder}\index{Dirichlet form!discrete}
Let $V$ be a countable set,  $m$ a measure and $\CQ: \CD \times \CD \to \IR$ a Dirichlet form on $\ell^2(V,m)$ with dense domain $\CD$. It was shown in \cite{KL-11} that for each regular Dirichlet form $\CQ$ on $\ell^2(V,m)$ there exists a graph defined by a jump-weight $j$ such that the restriction of $\CQ$ to the functions which are zero except at a finite number of points, is given by
\[\CQ(u) = \frac{1}{2} \sum_{x,y\in V} j(x,y)(u(x)-u(y))^2 + \sum_{x\in V} k(x) u(x)^2,\]
where $k$ is a killing weight on $V$. The operator associated to $\CQ$ is called the discrete (weighted) Laplacian and is given by
\[\triangle_{j,k,m} u(x) := \frac{1}{m(x)}\sum_{y\in V} j(x,y) (u(x)-u(y)) + \frac{k(x)}{m(x)} u(x),\]
with domain being a subspace of $\{u:V\to \IR \mid \sum\limits_{y\in V} j(x,y) |u(y)| <\infty \mbox{ for all } x\in V\}$.
Conversely, given a jump weight and a killing weight on $V$, we see by Fatou's lemma that the form $\CQ$ is closable on the space of all functions which are zero except at a finite number of points, and it is closed on its maximal domain, which is given by
\[\{u\in \ell^2(V,m) \mid \CQ(u)<\infty\}.\]
\end{reminder}
Our aim is to add a discrete Dirichlet form to the form $\CE$. Let $\FX_\Gamma$ be a metric graph with form $\CE$ and let $j$ be a jump weight over $\FV$. As mentioned at the end of the previous section we do not assume the graph to be connected. We consider an extension of the discrete form $\CQ$ from above in a certain sense. For $u\in \CC(\FX_\Gamma)$ the formal form $\tilde{\CQ}$ is defined as the mapping $\tilde{\CQ}: \CC(\FX_\Gamma) \to [0,\infty]$ given (as above) by
\[\tilde{\CQ}(u)= \frac{1}{2}\sum_{x,y\in \FV} j(x,y)(u(x)-u(y))^2+\sum_{x\in \FV} k(x) u(x)^2.\]
Unfortunately, the restriction $\CQ$ of $\tilde{\CQ}$ to the set
\[\CD(\CQ) =\{ u\in \CC(\FX_\Gamma) \mid \tilde{\CQ} (u) < \infty\},\]
is not closed or closable on $L^2(\FX_\Gamma,\omega)$. To see the latter, consider for $x\in \FV$ the sequence $\phi_n$, with $n\in \IN$ sufficiently large, given in vertex coordinates
\[(\phi_n)_x(r,k)=\begin{cases}  1-nr &, 0 \leq r \leq\tfrac{1}{n}\\ 0&, r>\tfrac{1}{n} \end{cases}.\]
Then $\phi_n\in \CD(\CQ)$, $\|\phi_n\|_2 \to 0$ and $\CQ(\phi_n - \phi_m) =0$.  If $\CQ$ were closable, we would have $\CQ(\phi_n)\to 0$. But as we have $\CQ(\phi_n)= const\neq 0$ for all $n\in \IN$, the form cannot be closable. Similarly one obtains that the form is not closed.\\
It is more reasonable to consider the perturbed form
\[\CS(u,v):= \CE(u,v)  + \CQ(u,v)\]
with domain
\[\CD(\CS):= \CD(\CE)\cap \CD(\CQ):= \{u \in \CD(\CE) \mid \tilde{\CQ} (u) < \infty\}.\]
Due to the local Sobolev embedding $\CD(\CE)\hookrightarrow \CC(\FX_\Gamma)$, this form is well-defined. The next lemma shows that this form is also closed.
\begin{lemma}
The form $\CS$ is closed on $L^2(\FX_\Gamma,\omega)$.
\end{lemma}
\begin{proof}
Let $(u_n)_n$ be a Cauchy sequence w.r.t. $\|\cdot\|_{\CS}$. Hence $(u_n)_n$ is a Cauchy sequence w.r.t. $\|\cdot\|_{\CE}$ and therefore there exists a function $u\in \CD(\CE)$ such that $u_n \to u$ in $\CE$-norm. By the local Sobolev embedding, Theorem 2.9, this sequence also converges pointwise, in particular
\[u(x) := \lim_{n\to\infty} u_n(x),\]
for all $x\in \FX_\Gamma$. Fatou's lemma now yields
\[\CQ(u - u_n) \leq \liminf_{m\to \infty} \CQ(u_m - u_n) \leq \liminf_{m\to\infty} \|u_m-u_n\|_\CS.\]
Since $(u_n)_n$ is also a $\|\cdot\|_\CS$ Cauchy sequence by assumption, the right hand side tends to zero for $n,m\to\infty$. This gives $u_n\to u$ w.r.t. $\|\cdot\|_\CS$.
\end{proof}
The discussion above and the previous lemma give rise to the following definition.
\begin{definition}\index{graph Dirichlet form}
Let $\FX_\Gamma$ be a metric graph. A graph Dirichlet form $\CS$ is given as $\CS= \CE + \CQ$ with domain $\CD(\CS)= \CD(\CE)\cap \CD(\CQ)$, where $\CE$ is a graph diffusion Dirichlet form and $\CQ$ is a discrete Dirichlet form associated with the jump weight $j(x,y)$ and the killing weight $k(x)$.
\end{definition}
Our next aim is to describe the domain of $\CS$. A sequence $(K_n)_n$ of finite subsets of $\FV$ is called a $\FV$-exhaustion, if $K_n \subset K_{n+1}$ and $\bigcup\limits_{n\in \IN} K_n = \FV$. Given a $\FV$-exhaustion $(K_n)_n$ and a discrete form $\CQ$ we can define the restricted forms
\[\CQ_n(u) := \frac{1}{2} \sum_{x,y\in K_n} j(x,y) (u(x)-u(y))^2 + \sum_{x\in K_n} k(x) u(x)^2,\]
and moreover we define the approximating forms
\[\CS_n := \CE + \CQ_n\]
with $\CD(\CS_n)= \CD(\CE)$, which is reasonable, since
\[\CQ_n(u) \leq C_n \|u\|_\CE\]
where $C_n >0$ depends on $j$, $l$, $k$ and $\nu$. In particular this gives that the forms $\CS_n$ are closed.
\begin{prop}
Let $u \in \CC(\FX_\Gamma)$, then $\sup\limits_{n} \CQ_n(u) < \infty$ if and only if $u \in \CD(\CQ)$. Furthermore
\[\CS_n \to \CS\]
on $\CD(\CE)\cap \CD(\CQ)$.
\end{prop}
\begin{proof}
The first claim follows from
\begin{eqnarray*}
u \in \CD(\CQ) &\Leftrightarrow& \sum_{x,y\in \FV} j(x,y) (u(x) - u(y))^2 + \sum_{x\in \FV} k(x) u(x)^2< \infty\\
&\Leftrightarrow& \lim_{n\to \infty} \sum_{x,y\in K_n} j(x,y) (u(x) - u(y))^2 + \sum_{x\in K_n} k(x) u(x)^2< \infty\\
&\Leftrightarrow& \lim_{n\to \infty} \CQ_n(u) <\infty.
\end{eqnarray*}
For the second claim, note that for $u\in\CD(\CE)\cap \CD(\CQ)$ it is clear that $\CS_n(u)$ is monotonic. Then $\CS_\infty := \lim\limits_{n\to \infty} \CS_n$ has the domain
\[\{u\in \CD(\CE)\mid \lim_{n\to\infty} \CQ_n(u) <\infty\}.\]
By the first statement, this set equals
\[\{u\in \CD(\CE)\mid u \in \CD(\CQ)\}\]
which is by definition $\CD(\CE)\cap \CD(\CQ)$.
\end{proof}
Next we introduce a substitute for connectedness of a graph Dirichlet form $\CS$ on $\FX_\Gamma$. Thus the following definition is very intuitive.
\begin{definition}
Let $\FX_\Gamma$ be a metric graph and $\CS$ a Dirichlet form on $\FX_\Gamma$. We call two connected components $\FY_1$, $\FY_2$ of $\FX_\Gamma$ $j$-connected if there exist $y_1\in \FY_1$ and $y_2\in \FY_2$ such that $j(y_1,y_2)>0$. This defines a graph structure on the set of connected components of $\FX_\Gamma$. We call the pair $(\FX_\Gamma, \CS)$ $j$-connected if this graph is connected.
\end{definition}
The definition could be done for any disjoint union of metric spaces and the jump measure of a Dirichlet form on it, as the jump measure connects the connected components.\medskip\\
It is known that in the case of vanishing $\CQ$, connectedness of $\FX_\Gamma$ is equivalent to the Dirichlet form $\CE$ being irreducible. The same holds for discrete Dirichlet forms, i.e. irreducibility is equivalent of the induced graph being connected. Before we show that $j$-connectedness is an appropriate replacement for connectedness, we recall the definition of irreducibility.
\begin{reminder}[Irreducibility]\index{Dirichlet form!irreducibility}
An $m$-measurable set $\FY\subset \FX_\Gamma$ is called invariant if for all $u,v \in \CD(\CS)$ we have that $\mathbf{1}_\FY u,\mathbf{1}_\FY v \in \CD(\CS)$ and
\[\CS(u,v) = \CS(\mathbf{1}_\FY u,\mathbf{1}_\FY v) + \CS(\mathbf{1}_{\FY^\complement}u, \mathbf{1}_{\FY^\complement} v).\]
A Dirichlet form is called irreducible if all invariant sets $\FY$ either satisfy $m(\FY)=0$ or $m(\FY^\complement)=0$.
\end{reminder}
We have the following proposition.
\begin{prop}
The graph Dirichlet form $\CS$ is irreducible if and only if $\FX_\Gamma$ is $j$-connected.
\end{prop}
\begin{proof}
We first assume that $\FX_\Gamma$ is $j$-connected. Let $\FY$ be an invariant set. Then for all $u\in \CD(\CS)$ the function $\mathbf{1}_\FY u$ belongs to $\CD(\CS)$ and in particular to $\CD(\CE)$. This gives that $\FY$ has to be a connected component of $\FX_\Gamma$. We thus have
\[\CS(u)- \CS(\mathbf{1}_\FY u)-\CS(\mathbf{1}_{\FY^\complement}u) = \CQ(u)- \CQ(\mathbf{1}_\FY u)-\CQ(\mathbf{1}_{\FY^\complement}u).\]
A short calculation shows that the right hand side equals
\[\sum_{x\in \FV}\sum_{y\in \FV} j(x,y)(\mathbf{1}_\FY(x)u(x)-\mathbf{1}_\FY(y)u(y)) ((\mathbf{1}_{\FY^\complement}(x)u(x)-\mathbf{1}_{\FY^\complement}(y)u(y))).\]
Note that all terms where $x,y\in \FY$ or $x,y\in \FY^\complement$ vanish, so that this term equals
\[-2\sum_{x\in \FV\cap \FY}\sum_{y\in \FV\cap \FY^\complement} j(x,y)u(x)u(y).\]
By $j$-connectedness, this cannot be zero for all $u\in \CD(\CS)$ and thus the equality
\[\CS(u) = \CS(\mathbf{1}_\FY u)+\CS(\mathbf{1}_{\FY^\complement}u)\]
holds if and only if $\FY=\emptyset$ or $\FY=\FX_\Gamma$. This gives that $\CS$ is irreducible.\\
Conversely, assume that $\FX_\Gamma$ is not $j$-connected and let $\FY$ be a connected component of $\FX_\Gamma$. It is immediate that then $\mathbf{1}_\FY u\in \CD(\CS)$ for all $u\in \CD(\CS)$. Assume additionally that there does not exist $x\in \FY$ and $y\in \FY^\complement$ with $j(x,y)>0$. Then, as above, we have
\[\CS(u)- \CS(\mathbf{1}_\FY u)-\CS(\mathbf{1}_{\FY^\complement}u) = -2\sum_{x\in \FV\cap \FY}\sum_{y\in \FV\cap \FY^\complement} j(x,y)u(x)u(y).\]
By assumption all $j(x,y)$ in this sum are equal to zero, and we obtain that $\FY$ is an invariant set. Thus $\CS$ is not irreducible.
\end{proof}
As usual, one can restrict to irreducible Dirichlet forms when one is interested in certain analytical properties. Thus from this point on, we make the following assumption.
\begin{assumption}\index{assumption!irreducibility}
The graph Dirichlet form $\CS$ is irreducible.
\end{assumption}
\section{Bounded perturbations}
Up to this point we have not made any link between the forms $\CE$ and $\CQ$. However if we are interested in certain analytical properties each parameter may enter the stage. In this section we will focus on the question, when the form $\CQ$ is small with respect to $\CE$, i.e. forms $\CS$ which are equivalent to $\CE$ in the form sense. We start with the following observation.
\begin{prop}\index{relative boundedness}
If $\CD(\CE)\subset \CD(\CQ)$, then the graph form $\CS$ is relatively bounded with respect to the form $\CE$, that is the identity map $$(\CD(\CE),\|\cdot\|_\CE) \to (\CD(\CE),\|\cdot\|_\CS)$$ is continuous, i.e. there exists $C>0$ such that for all $u\in \CD(\CE)$ we have
\[\|u\|_\CS \leq C  \|u\|_\CE.\]
In particular this inequality holds if and only if there exists $C>0$ such that for all $u\in \CD(\CE)$ we have
\[\CQ(u) \leq C (\|u\|_2^2 + \CE(u)).\]
\end{prop}
\begin{proof}
Since both spaces are complete, Banach's theorem on the bounded inverse gives the inequality. The second inequality is thus a consequence of the embedding.
\end{proof}
If we are in the position of the proposition above we just say that $\CQ$ is $\CE$-bounded for short.\medskip\\
We are now interested how the geometry of the graph and the discrete data arising from the form $\CQ$ are connected with $\CE$-boundedness of $\CQ$. To do so we need an appropriate concept of the size of the form $\CQ$.
\begin{definition}\index{Dirichlet form!discrete!support}
Let $\FX_\Gamma$ be a metric graph and $\CQ$ a discrete form. Then the support of $\CQ$ is defined as
\[\supp \CQ := \{x\in \FX_\Gamma \mid \forall \phi\in \CC_c(\FX_\Gamma), x\in \supp \phi : \CQ(\phi) >0\}.\]
\end{definition}
An immediate consequence of the definition of $\CQ$ is that $\supp \CQ$ is a discrete subset of $\FV$. Furthermore, we have the following characterization.
\begin{prop}
Let $\CQ$ be a discrete Dirichlet form on a metric graph $\FX_\Gamma$. Then we have
\[\supp \CQ = \{x\in \FX_\Gamma \mid \exists y\in \FX_\Gamma : j(x,y)>0\}\cup \{x\in \FX_\Gamma \mid c(x)>0\}.\]
\end{prop}
\begin{proof}
For $x\in \supp \CQ$ let $\phi_x$ be a continuous function having support in the maximal star-neighborhood of $x$ and $\phi_x(x)=1$. Then
\[0 < \CQ(\phi_x) = \sum_{y\in \FV} j(x,y) + c(x).\]
Thus $c(x)>0$ or for some $y\in \FV$ from the right hand side the weight $j(x,y)$ doesn't vanish. Conversely, let $x\in \FX_\Gamma$ with $c(x)>0$ or such that there exists $y\in \FX_\Gamma$ with $j(x,y)>0$. Then
\[0< c(x) \leq \CQ(\phi_x),\]
respectively
\[0< j(x,y) \leq \sum_{y} j(x,y) = \CQ(\phi_x),\]
and thus $x\in \supp \CQ$.
\end{proof}
We will now discuss that $\CE$-boundedness of $\CQ$ is measured in terms of the boundedness of $\CQ$ as form defined on an $\ell^2$ space with respect to a certain measure. This measure is related to the form $\CE$.
\begin{lemma}
Let $\FX_\Gamma$ be a metric graph and $\CS$ a graph Dirichlet form. If $\CQ$ is $\CE$-bounded, then there is a constant $C>0$ such that we have for all $x\in \supp \CQ$
\[\sum_{y\in \FV}  j(x,y) + k(x)  \leq C\ {\mathrm{cap}}(\{x\},\supp \CQ) .\]
where
\[{\mathrm{cap}}(\{x\},\supp \CQ) = \inf\{ \|\phi\|^2_2 + \CE(\phi) \mid \phi\in \CD(\CE),\ \phi(x)=1, \phi|_{\supp \CQ \setminus \{x\}}=0\}\]
is the relative capacity of $\{x\}$ with respect to $\supp \CQ$. Furthermore we have
\[{\mathrm{cap}}(\{x\},\supp \CQ)\leq \sum_{e\sim x} \omega(e) \coth l_i(e),\]
where equality holds if and only if for all neighbors $y$ of $x$ we have $y\in \supp \CQ$.
\end{lemma}
\begin{proof}
Fix $x\in \FV$ and let $\phi_x\in \CD(\CE)$ with $\phi(x)=1$ and $\phi|_{\FV\setminus\{x\}} =0$. Note that such a $\phi_x$ always exists for all $x\in \FV$. A short calculation gives
\[\CQ(\phi_x)= \sum_{y\in \FV} j(x,y) + k(x).\]
By $\CE$-boundedness of $\CQ$ we conclude
\[\sum_{y\in \FV} j(x,y) + k(x) \leq C( \|\phi_x\|^2_2 + \CE(\phi_x))\]
for all such $\phi_x$. Now taking the infimum on the right hand side over all such $\phi_x$ we get
\[\sum_{y\in \FV} j(x,y) + k(x) \leq C\ {\mathrm{cap}}(\{x\},\supp \CQ).\]
For the second statement, we first observe that the infimum is attained by a function $\phi$ which is zero outside the connected component of $x$ of $\{x\}\cup \FX_\Gamma \setminus \supp \CQ$. Furthermore one can show that this infimum is $1$-harmonic outside of $\supp \CQ$, i.e. its components satisfy the differential equation $-(\phi_x)_e'' + (\phi_x)_e =0$. Thus a natural candidate of $\phi_x$ is given with components
\[(\phi_x)_e(t) = \frac{\sinh(l_i(e)-t)}{\sinh l_i(e)}\]
for all edges $e\sim x$ being oriented such that $x=\partial^+(e)$, and $\phi_e=0$ for all other edges. Then we get
\[\CE(\phi_x) + \|\phi_x\|_2^2 = \sum_{e\sim x} \omega(e) \coth l_i(e),\]
and we obtain the estimate of the capacity. If in addition all neighbors of $x$ belong to the support of $\CQ$, the connected component of $x$ is exactly the star graph with center $x$ and the just constructed function is the infimum in question.
\end{proof}
The estimate in the previous lemma is sharp in the sense that we obtain pointwise bounds for a discrete Laplacian. However, if the distance between neighboring vertices tends to zero, one might expect that $\CQ(\phi_x)$ should tend to zero as well, in contrast to $\coth l_i(e)\to \infty$ as $l_i(e)\to 0$. We circumstantiate our expectation by the following example.
\begin{example}
Let $\FX_\Gamma = \bigcup_{n\in \IN} \FY_n$ be a disjoint union of compact graphs $\FY_n$ and let $\CS$ be a graph Dirichlet form on $\FX_\Gamma$ such that $\supp \CQ \cap \FY_n$ consists of only one element for all $n\in \IN$. Then $\mathbf{1}_{\FY_n} \in \CD(\CS)$ and has compact support. In particular for all $x\in \supp \CQ$ we have
\[\kap(\{x\},\supp \CQ) < \omega(\FY_n)\]
and therefore if $\CQ$ is $\CE$-bounded we obtain
\[\sum_{y\in \FV}  j(x,y) + k(x) \leq  C \omega(\FY_n).\]
If for instance $\omega\equiv 1$ and $\diam (\FY_n)\to 0$ as $n\to \infty$, then as expected also the left hand side has to tend to zero. More precisely, assume each $\FY_n$ is given as the interval $[0,l_n]$ and assume that $\supp \CQ$ consists of the initial vertices, then a short calculation shows that $\CE(u) + \|u\|_2^2$ is minimized on $\FY_n$ by the function
\[u_n(x)= \frac{\sinh(l_n - x)}{\sinh l_n} + \frac{\sinh x}{\sinh l_n}\]
and therefore we obtain
\[\kap(\{x\},\supp\CQ)= \tanh l_n\]
which tends to zero as $l_n$ does.
\end{example}
Another way to interpret the previous result is that the associated discrete Laplacian is a bounded operator on the space $\ell^2(\supp \CQ, m)$, where $m$ is a discrete measure. The question is whether for all bounded discrete Laplacians on some $\ell^2$ space the associated form $\CQ$ gives an $\CE$-bounded form. The following proposition gives us good candidates for the discrete measure $m$.
\begin{prop}
Assume there exists a function $m:\supp \CQ \to (0,\infty)$ such that $\triangle_{j,k}$ is a bounded operator on $\ell^2(\supp \CQ,m)$, i.e. there exists $C>0$ such that we have
\[\sum_{y\in \FV}  j(x,y) + k(x) \leq C m(x),\]
and the restriction map
\[\cdot|_{\supp \CQ} : \CD(\CE) \to \ell^2(\supp \CQ,m)\]
is continuous, then $\CQ$ is $\CE$-bounded.
\end{prop}
\begin{proof}
As was shown in \cite{KL-10} the first assumption is exactly the boundedness of the discrete Laplacian on $\ell^\infty$, which is in turn equivalent to the boundedness of the form on $\ell^2$. This means
\[\CQ(u) \leq C \|u\|^2_{\ell^2}\]
for all $u\in \CD(\CQ)$. Using the continuity of the restriction map we obtain
\[\CQ(u) \leq C \|u\|^2_{\ell^2} \leq C'(\|u\|_2^2 + \CE(u))\]
for all $u\in \CD(\CE)$.
\end{proof}
We are left with a discussion for which measures the restriction map is continuous. Note that the bigger the measure $m$ is, the more discrete Laplacians give $\CE$-bounded forms.
\begin{prop}
Let $(\FY_x)_{x\in \supp \CQ}$ be a collection of finite subgraphs with $x\in \FY_x$ such that
\[\sup_{y\in \FX_\Gamma} \# \{ \FY_x \mid y\in \FY_x\} < \infty,\]
then with $m(x):= \min\{\nu(\FY_x), (\diam \FY_x)^{-1}\}$ the mapping
\[\cdot|_{\supp \CQ} : \CD(\CE) \to \ell^2(\supp \CQ,m)\]
is continuous.
\end{prop}
\begin{proof}
We apply the Sobolev inequality to the subgraph $\FY_x$ and obtain
\[|u(x)|^2 \leq 2 \nu(\FY_x)^{-1} \int\limits_{\FY_x}|u|^2 d\nu + 2 \diam(\FY_x) \int\limits_{\FY_x} |u'|^2 d\lambda \]
and therefore
\[|u(x)|^2 m(x) \leq 2 ( \|u|_{\FY_x}\|_2^2 + \|u'|_{\FY_x}\|_2^2).\]
Summing now over $x\in \supp \CQ$ and taking into account that each point $y\in \FX_\Gamma$ is in less then $C:= \sup_{y\in \FX_\Gamma} \# \{ \FY_x \mid y\in \FY_x\}$ many subgraphs, we get
\[\sum_{x\in \supp \CQ} |u(x)|^2 m(x) \leq 2C (\|u\|^2_2 + \CE(u)).\]
\end{proof}
Note that we have not assumed that the collection $(\FY_x)_x$ is a covering of the graph $\FX_\Gamma$. Thus one natural choice is for instance to choose an edge adjacent to $x\in \supp \CQ$ as $\FY_x$. As $\nu(e)l_c(e)^2 \leq 1 $ by assumption we have
\[ \min\{ \nu(e)l_c(e), \tfrac{1}{l_c(e)}\} = \nu(e) l_c(e)\]
we could choose $m(x):= \max\limits_{e\sim x} \nu(e) l_c(e)$.\medskip\\
Certain special cases of the above proposition are already known. We present them in the following corollary. In particular, we can actually recover the result given in \cite{Ku-04} section 3.3, concerning quantum graphs. The precise connection between our setting and the quantum graph setting will be the treated in the interlude following this chapter.
\begin{coro}
Assume that the measure $\nu$ is uniformly bounded from below by a positive constant and assume that for $x\in \supp\CQ$ the numbers
\[r(x):= \inf\{d_c(x,y) \mid y\in \supp \CQ\}\]
are also uniformly bounded from below by a positive constant. Then any bounded Laplacian $\triangle_{j,k}$ on $\ell^2(\supp \CQ ,1)$ gives an $\CE$-bounded form $\CQ$. In particular this holds if $\nu \equiv 1$ and if the edge lengths are uniformly bounded from below.
\end{coro}
It is the in particular statement of the last corollary which was proven in \cite{Ku-04} section 3.3.\medskip\\
Later we will again treat the question when the restriction map is continuous, but with other methods having a stronger focus on the discreteness of the graph $\FX_\Gamma$.
\section{Regularity in the general case}
Investigating regularity properties of a graph Dirichlet form $\CS$ is harder than in the diffusion case. In order to study regularity we will follow another approach. Before we are doing so, we first consider the case when the diffusion governs the form.
\begin{coro}
Let $\FX_\Gamma$ be complete and $\CS$ be a graph Dirichlet form such that $\CQ$ is $\CE$-bounded. Then $\CS$ is regular.
\end{coro}
Recall that regularity means that $\CS \cap \CC_c$ is dense in both $(\CS, \|\cdot\|_\CS)$ and $(\CC_c(\FX_\Gamma), \|\cdot\|_\infty)$. In the general case this intersection can be rather small. Regularity is a property which says, that there are enough nice functions in the domain which are continuous and have compact support. Fortunately, due to the Sobolev embedding, functions in $\CD(\CS)$ are already continuous. In particular we have
\[\CD(\CS)\cap \CC_c(\FX_\Gamma) =\CD_c(\CE) \cap \CD(\CQ).\]
Regularity implies that this intersection is even bigger.
\begin{lemma}
If $\CS$ is a regular graph Dirichlet form, then $\CD_c(\CE)$ is contained in $\CD(\CS)$ and thus
\[\CD(\CS)\cap \CC_c(\FX_\Gamma) =\CD_c(\CE).\]
\end{lemma}
\begin{proof}
Note, that each function in $\CD_c(\CE)$ could be written as a sum of functions in $\CD_c(\CE)$ supported on edges or star neighborhoods of vertices. Moreover for $\varphi\in \CD_c(\CE)$ supported on an edge, we have $\CQ(\varphi)=0$ and thus $\varphi \in \CS$. Let now $\varphi\in \CD_c(\CE)$ be supported in a star neighborhood of a vertex $x\in \FV$.  Let $\psi$ be the function which is equal to $\varphi(x)$, supported on the support of $\varphi$ and edgewise affine linear. Then the function $\varphi-\psi$ is a sum of functions supported on edges. Thus it suffices to show that $\psi \in \CD(\CS)$. Choose $\psi(x)=2$, then there exists a function $\eta\in \CD(\CS)$ with $\eta(x) > 1$ and $|\eta (y) | <1$ for all $y\sim x$, since $\CC_c(\FX_\Gamma)\cap \CD(\CS)$ is dense in $\CC_c(\FX_\Gamma)$ with respect to the supremum norm. Moreover since $\CD(\CS)$ is invariant under taking modulus, we can assume that $\eta$ is non-negative and again we can assume that $\eta$ is edgewise affine linear. Since $\CS$ is a Dirichlet form, $\tilde{\eta} = \eta \wedge 1$ belongs to $\CD(\CS)$ as well. Hence, again since $\CD(\CS)$ is a vector space, it contains the function $\eta - \tilde{\eta}$.
\end{proof}
In order to treat regularity in the general case, we will have a first glimpse on harmonic functions. They will be defined and analyzed in the next chapter. Given a metric graph $\FX_\Gamma$ and graph Dirichlet form $\CS$, we introduce, as in the case of a diffusion Dirichlet form, the restriction of the form $\CS$ to the space
\[\CD_o(\CS)= \overline{\CD(\CS)\cap \CC_c(\FX_\Gamma)}^{\|\cdot\|_\CS}.\]
With similar arguments as in section 3, regularity is equivalent to the identicalness of the spaces $\CD(\CS)$ and $\CD_o(\CS)$. As the latter space is a closed subspace, we obtain the following theorem.
\begin{theorem}
Let $\FX_\Gamma$ be a metric graph and $\CS$ a graph Dirichlet form. Then the orthogonal complement $\CH$ of the space $\CD_o(\CS)$ consists of all functions $u\in \CD(\CS)$ such that for all $\phi\in \CD_c(\CS)$ we have
\[\CS(u,\phi) + (u,\phi) = 0.\]
In particular we have the orthogonal decomposition
\[\CD(\CS) = \CD_o(\CS) \oplus \CH_1.\]
\end{theorem}
\begin{proof}
The orthogonality claim follows, as the condition is nothing but orthogonality with respect to the scalar product coming from the energy norm $\|\cdot\|_\CS$. The claim on the orthogonal decomposition follows as the space $\CD_o(\CS)$ is a closed subspace.
\end{proof}
We refer to the elements of $\CH_1$ as $1$-harmonic functions. The theorem can also be interpreted as existence theorem for $1$-harmonic functions. More generally, for $\alpha>0$ the norm defined by $(\alpha\|\cdot\|_2^2 + \CS(\cdot))^\frac{1}{2}$ defines an equivalent norm on $\CD(\CS)$. In particular the closures of the set $\CD(\CS)\cap \CC_c(\FX_\Gamma)$ are identical. The orthogonal complement of $\CD_o(\CS)$ with respect to the modified inner product is given as the space $\CH_\alpha$ of $\alpha$-harmonic functions. In particular one obtains the orthogonal decomposition
\[\CD(\CS) = \CD_o(\CS) \oplus \CH_\alpha.\]
This gives that for all $\alpha>0$ the spaces $\CH_\alpha$ are isomorphic. We continue the study of harmonic functions in more detail in chapter 3.\\
We finish this section with a short discussion of harmonic functions of a graph diffusion Dirichlet form.
\begin{coro}
Let $\FX_\Gamma$ be a metric graph and $\CE$ a graph diffusion Dirichlet form. If there is a non-negligible canonical boundary point, then there exists a $1$-harmonic function in $\CD(\CS)$.
\end{coro}
\begin{proof}
The existence of such a boundary point gives by Theorem 2.24 that the form $\CE$ is not regular. The previous theorem gives the existence of such a harmonic function.
\end{proof}

\section*{Notes and remarks}
The idea behind the introduction of the forms is to have a rather general form of a Dirichlet form. This kind of Dirichlet form is new in the study of metric graphs. Using the length transformation gives us good tools to study analytical questions such as embedding properties and regularity issues. Thus all the material presented here is new and might give some new insights into related topics such as weighted Sobolev spaces, general Dirichlet forms and analysis on metric measure spaces. It would be interesting to try to compare general Dirichlet forms on graphs with graph Dirichlet forms. At least for certain global properties this should be possible. From the metric graph perspective, this chapter is also interesting as we avoid any lower bound on the edge lengths. In a certain sense, decreasing edge lengths should correspond to an increasing curvature, though the latter concept has no definition yet. The introduction of graph Dirichlet forms was mainly motivated from the theory of quantum graphs. This will be treated in the upcoming interlude. 
\chapter*{Interlude - Concerning Quantum Graphs}
In this interlude we have a look at so-called quantum graphs and show that they fit into our setting. This enables us to apply tools from the theory of regular Dirichlet forms to quantum graphs. We start with recalling some notation and basic results concerning infinite quantum graphs. As a major difference to diffusion Dirichlet forms, one considers a different subspace of $\bigoplus\limits_{e\in E} W^{1,2}(0,l(e))$ as domain of definition of the energy form on the metric graph. This subspace heavily involves the values of the components on the boundary of the intervals. A disadvantage of this is that functions which are continuous in a neighborhood of the vertices fail to belong to the domain in general. Before we can make this precise we need to recall some notation concerning quantum graphs. Note that in this interlude we just treat the case $\nu \equiv 1$, which also gives that the canonical metric and the intrinsic metric agree with the given one. Moreover a typical assumption in quantum graph theory is that there is a uniform lower edge length. In this interlude we will also make this assumption.
\begin{assumption}\index{assumption!length}
We assume that the length function $l:E\to \IR$ has a positive lower bound
\begin{equation*}
\inf\limits_{e \in E} l(e) >0.
\end{equation*}
\end{assumption}
Let $\FX_\Gamma$ be an infinite metric graph. Recall that any function $u:\FX_\Gamma \to \IR$ can be given in vertex coordinates by its components $u_x(r,k)$. We use this for the next definition.
\begin{definition}\index{trace}
For $u:\FX_\Gamma \to \IR$ we define the trace $\tr{x} u \in \IR^{\deg{x}}$ of $u$ in the point $x\in \FV$ as
\[\tr{x} (u) := (\lim_{r\to 0} u_x(r,k))_{k=1}^{\deg(x)},\]
whenever the limits exist.
The trace of $u$ is then defined as
\[\tr{} u := (\tr{x} u)_{x\in \FV}.\]
\end{definition}
In the theory of quantum graphs one relates the values of the trace of a function with the value of the trace of its derivative. The question is whether these values exist or not. In the case of a continuous function for instance the trace exists and is a multiple of the vector with constant entries. For what we want to deal with, the following lemma gives a good class of functions.
\begin{lemma}
Let $\FX_\Gamma$ be a metric graph. Then for all $x\in \FV$ the mapping
\[ \tr{x}: \prod_{e\in E} W^{1,2}(0,l(e)) \to \IR^{\deg{x}},\]
\[u\mapsto \tr{x} u\]
is continuous and the mapping
\[\tr{}: \bigoplus_{e\in E} W^{1,2}(0,l(e)) \to \bigoplus_{x\in \FV} \IR^{\deg(x)},\]
\[u\mapsto \tr{} u\]
is continuous. In particular, the mapping
\[u \mapsto (\tr{} u , \tr{} u')\]
is well-defined and continuous as mapping from $\bigoplus\limits_{e\in E} W^{2,2}(0,l(e))$ to $\bigoplus\limits_{x\in \FV} \IR^{\deg(x)} \oplus  \bigoplus\limits_{x\in \FV} \IR^{\deg(x)}$.
\end{lemma}
\begin{proof}
Everything follows from the one-dimensional Sobolev inequality, as it is uniform by assumption on the lengths of the edges.
\end{proof}
Note that since $\tr{x} (u)$ maps into $\IR^{\deg(x)}$, we can identify the  canonical basis of the latter one, with the edges incident to $x$.\medskip\\
We now follow the approach in \cite{Ku-04} to define quantum graphs. However, since we are only interested in Dirichlet forms, we restrict the class of quantum graphs under consideration. In particular we will call them Dirichlet quantum graphs. We thus follow \cite{KKVW-09,SV-11}, where all quantum graphs giving rise to Dirichlet forms have been characterized. To do so, for each $x\in \FV$ we choose a subspace $X_x$ of $\IR^{\deg(x)}$ which is additionally a Stonean sublattice, i.e.
\begin{itemize}
\item $y\in X_x$ implies $|y|\in X_x$,
\item $y\in X_x$ implies $y\wedge 1 \in X_x$,
\end{itemize}
where $|y|=(|y_1|,\dots,|y_{\deg(x)}|)$ and $y\wedge 1 = (y_1\wedge1,\dots, y_{\deg(x)}\wedge 1)$, and for each $x\in \FV$ a linear operator $L_x$ on $X_x$ which is
\begin{itemize}
\item self-adjoint, and
\item Markovian, i.e. $(L_x|y|,|y|) \leq (L_x y,y)$ and $(L_x(y\wedge 1),y\wedge 1) \leq (y\wedge 1, y\wedge 1)$ for all $y\in X_x$.
\end{itemize}
As mentioned earlier, a quantum graph is defined as a metric graph together with a form defined on a certain subspace of $\bigoplus\limits_{e\in E} W^{1,2}(0,l(e))$.
\begin{definition}\index{quantum graph}
Let $\FX_\Gamma$ be a metric graph, for each $x\in \FV$ let $X_x$ be a Stonean sublattice of $\IR^{\deg(x)}$ and $L_x$ a self-adjoint, Markovian operator on $X_x$ such that $\sup\limits_{x\in \FV} \|L_x\| < \infty$. Let $\Fh$ be the quadratic form on $L^2(\FX_\Gamma)$ acting as
\[\Fh(u)= \sum_{e\in E} \int\limits_0^{l(e)} |u_e'(t)|^2 \:dt + \sum_{x\in V} (L_x \tr{x} (u) , \tr{x} (u))\]
with domain
\[\CD(\Fh) = \{(u_e)_{e\in E} \in \bigoplus_{e\in E} W^{1,2} (0,l(e)) \mid \forall x\in V : \tr{x} (u) \in X_x\}.\]
The pair $(\FX_\Gamma, \Fh)$ is called a Dirichlet quantum graph.
\end{definition}
In \cite{Ku-04} and \cite{KKVW-09,SV-11} it was shown that such forms are in one-to-one correspondence with self-adjoint, Markovian extensions of the operator defined on $\bigoplus\limits_{e\in E} C_c^\infty(0,l(e))$ acting edgewise as the second derivative. One can show that the associated operators to these forms are given by the family of orthogonal projections $P_x$ onto $X_x^\perp$ and the family of self-adjoint operators $L_x$ via the boundary conditions for a function $u\in \bigoplus\limits_{e\in E} W^{2,2}(0,l(e))$ given as
\begin{align*}
P_x \tr{x}(u)=0 \qquad \text{and} \qquad L_x (\Id-P_x) \tr{x}(u)=(\Id-P_x) \tr{x}(u').
\end{align*}
In \cite{Ku-04} these boundary conditions were used to find self-adjoint Laplacians with associated forms. In \cite{KKVW-09,SV-11} it was characterized when an operator using these boundary conditions is associated to a Dirichlet form on finite, compact metric graphs.\medskip\\
That the form $\Fh$ is well defined follows from the bounds on the norms of $L_x$ and the lemma above. In our language the $L_x$'s are $\CE$-bounded. Apart from the fact that we restrict to this case, one may think that Dirichlet quantum graphs are much more general than the forms defined in the previous sections. However, we will show that they actually arise as a special case. On the other hand, the case of a diffusion Dirichlet form is a special case of Dirichlet quantum graph. Then they are called Kirchhoff conditions, which will arise in the next chapter again. They are a special case of the boundary conditions of the form $(X_x,L_x)$, introduced by Kuchment for quantum graphs. These continuity conditions can be expressed as $\mathrm{lin} \{\mathbf{1}_{\{1,\dots, \deg(x)\}} \} = X_x$, i.e.
\[\tr{x} u \in \mathrm{lin} \{\mathbf{1}_{\{1,\dots, \deg(x)\}} \}.\] It is easy to see that not each Stonean sublattice is of the form above. Assume that for some vertex one has chosen a Stonean sublattice which is different from the Stonean sublattice which gives rise to continuous functions. Then any function which is continuous in this vertex and nonzero cannot belong to the domain $\CD(\Fh)$. Fortunately, as we will show, we can change the graph topology and introduce a new graph, such that the forms stay the same in a precise meaning, but the new form contains also absolutely continuous functions with compact support. As was mentioned in \cite{KKVW-09}, each Stonean sublattice $X_x$ of $\IR^{\deg(x)}$ can be expressed as
\[X_x = \mathrm{lin} \{ \mathbf{1}_{C_j} \mid j=1,\dots, m\}\]
where $C_1,\dots,C_m$ is a partition of a subset of $\{1,\dots,\deg(x)\}$ with $m$ being the dimension of $X_x$. In particular such a partition gives a partition of the edges in $\bigstar(x)$, which we will also denote by the $C_j$'s. If now $e_1,e_2\in C_j$ for some $j\in \{1,\dots,m\}$ then
\[\lim_{t\to x} u_{e_1} (t) = \lim_{t\to x} u_{e_2} (t).\]
Since this holds for each $j$ and all edges corresponding to $C_j$ we obtain some kind of continuity. For edges $e\in E$ which do not belong to any of the $C_j$ we have
\[\lim_{t\to v} u_e(t)=0.\]
This corresponds to Dirichlet boundary conditions and one may call such an edge open, since the value of $u$ is prescribed to be zero at this boundary point.\medskip\\
This observation leads us immediately to our main idea. We will construct an appropriate graph, by dividing bad vertices into more vertices such that functions in the domain are actually continuous. In the spirit above, each such $C_j$ serves as a new vertex. In figure \ref{fig:qg} such a cutting procedure is shown.
\begin{figure}
\begin{minipage}[hb]{6cm}
\begin{tikzpicture}[scale=0.7]
\def\anz{6}
\newlength{\radius}\setlength{\radius}{0.1cm}
\fill (0:0) circle (\radius);
\foreach \x in {1,2,...,\anz}{\draw (\x*60:0) -- (\x*60:2.5);
%\fill (\x*60:2.5) circle  (\radius);
 \node at (\x*60:3) {${\x}$};
};
\end{tikzpicture}
\end{minipage}
\begin{minipage}[hb]{6cm}
\begin{tikzpicture}[scale=0.7]
\setlength{\radius}{0.1cm}
\fill (120:1) circle (\radius);\node at (120:0.5) {$C_1$};
\foreach \x in {1,2,3} {
\draw (120:1)--($(120:1)+(\x*60:2.5) $);
\node at ($(\x*60:3)+(120:1)$) {$\x$};
%\fill ($(120:1)+(\x:2.5)$) circle (\radius);
};
\fill (240:1) circle (\radius);\node at (240:0.5) {$C_2$};
\foreach \x in {4,5} {
\draw (240:1)--($(240:1)+(\x*60:2.5) $);
\node at ($(\x*60:3)+(240:1)$) {$\x$};
%\fill ($(240:1)+(\x:2.5)$) circle (\radius);
};
\fill (0:1) circle (\radius);\draw (0:1) -- (0:3.5);\node at (0:4) {$6$};\node at (0:0.5) {$C_3$};
%\fill (0:3.5) circle (\radius);
%\draw[densely dotted] (120:1) -- (240:1);\draw[densely dotted] (240:1)--(0:1);\draw[densely dotted] (0:1)--(120:1);
\end{tikzpicture}
\end{minipage}
\caption{Illustration of a star graph of degree six with $C_1=\{1,2,3\}$, $C_2=\{4,5\}$ and $C_3=\{6\}$.}
\label{fig:qg}
\end{figure}
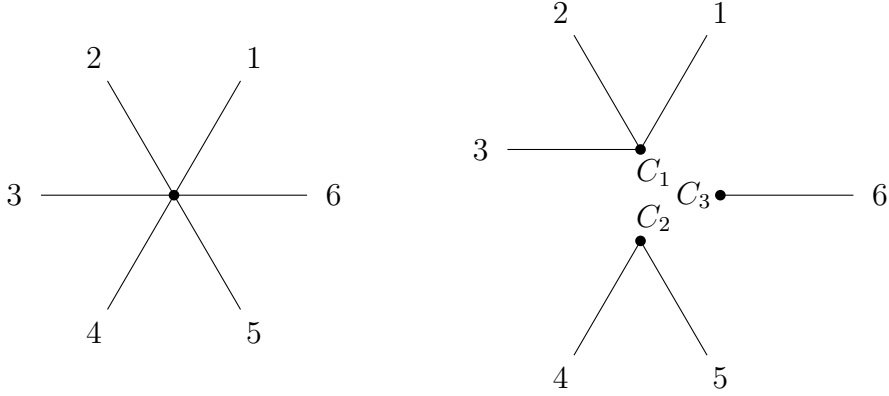
Applying this cutting procedure to each vertex gives the following theorem.
\begin{theorem}\index{quantum graph!representation}
Let $(\FX_\Gamma,\Fh)$ be a Dirichlet quantum graph. Then there exists a metric graph $\FX_{\widetilde{\Gamma}}$ and a graph Dirichlet form $\CS$ on $L^2(\FX_{\widetilde{\Gamma}})$ such that the Dirichlet spaces corresponding to $\Fh$ and $\CS$ are equivalent, i.e. there exists an algebraic isomorphism $\Phi:\CD(\Fh)\cap L^\infty(\FX_\Gamma) \to \CD(\CS) \cap L^\infty (\FX_{\widetilde{\Gamma}})$ such that
\begin{itemize}
\item $\|u\|_\infty = \|\Phi(u)\|_\infty$,
\item $\|u\|_2 = \|\Phi(u)\|_2$,
\item $\Fh(u) = \CS(u)$.
\end{itemize}
Moreover, the metric graph can be chosen such that $\CS$ is a regular Dirichlet form and the discrete part of $\CS$ is $\CE$-bounded.
\end{theorem}
\begin{proof}
We first construct the new graph $\widetilde{\Gamma}$. Since $\Fh$ is coming from a Dirichlet quantum graph, each vertex $x$ is equipped with a Stonean sublattice. Denote by $\{C_j^x\mid j=1,\dots,m_x\}$ the partition of a subset of $\{1,\dots,\deg(x)\}$ such that
\[X_x = \mathrm{lin} \{\mathbf{1}_{C_j^x} \mid j=1,\dots,m_x\}\]
and denote by
\[C_0^x = \{1,\dots,|\deg(x)|\}\setminus \bigcup_{j=1}^{m_x} C_j^x\]
the set of free vertices. With $\widetilde{E}$ being the same set of edges $E$, we define the new set of vertices as
\[\widetilde{V} = \Bigl(\bigcup_{x\in V} \bigcup_{j=1}^{m_x} \{C_j^x\}\Bigr) \cup \bigcup_{x\in V} \{C_0^x\}.\]
We define the incidence relation of the edge set $E$ and the new vertex set $\widetilde{V}$ as follows: let $e\in E$, then $e$ can be canonically identified with a base vector of the space $\IR^{\deg(x)}$ with $x = \partial^+(e)$. By the representation of $X_x$ using the $C_j^x$'s, we say $C_j^x$ is the initial vertex if $e$ as base vector belongs to the set $C_j^x$. Analogously, we define the new terminal vertex. We define now $\FX_{\widetilde{\Gamma}}$ to be the metric graph associated with the graph $\widetilde{\Gamma} =(\widetilde{V},\widetilde{E})$ equipped with the same lengths and the orientation from above. We now define the form $\CS$. Let $\CD(\CS)$ be the set of all function in $\CC(\FX_{\widetilde{\Gamma}})\cap \bigoplus_{e\in E} W^{1,2}(0,l(e))$ vanishing at all $C\in \widetilde{\FV}$ with $C\in C_0^x$. The diffusion part of the form is defined as $\CE(u)=\sum_{e\in E} \int\limits_0^{l(e)} |u'_e(t)|^2 dt$. To define the discrete part $\CQ$ we use the $L_x$'s. Since $L_x$ is defined on $X_x$, and the latter is orthogonally spanned by $\mathbf{1}_{C_j^x}$, we identify the base of $X_x$ with the $C_j^x$ and set
\[L_x(C_j^x, C_k^x)= (L_x \mathbf{1}_{C_j^x},\mathbf{1}_{C_k^x}).\]
We define the jump weights to be
\[j(C_j^x, C_k^x) = -L_x(C_j^x, C_k^x)\]
and the killing weights
\[k(C_j^x) = L_x(C_j^x, C_j^x) + \sum_{k=1}^{m_x} L_x(C_j^x, C_k^x).\]
Thus we have
\[\sum_{x\in \FV} (L_x \tr{x} (u),\tr{x} (u)) = \CQ(u)\]
by definition. Since both $\CD(\CS)$ and $\CD(\Fh)$ can be considered as closed subspaces of $\bigoplus_{e\in E} W^{1,2}(0,l(e))$, the choice of $\Phi$ as identity map from $\CD(\Fh)\cap L^\infty(\FX_\Gamma)$ to $\CD(\CS) \cap L^\infty (\FX_{\widetilde{\Gamma}})$, i.e. the mapping
\[\Phi((u_e)_{e\in E}) = (u_e)_{e\in E}\]
is well-defined. The algebraic properties of $\Phi$ are obvious as well as $\|u\|_\infty = \|\Phi(u)\|_\infty$ and $\|u\|_2 = \|\Phi(u)\|_2$. Note that for $u\in \bigoplus_{e\in E} W^{1,2}(0,l(e))$ the condition $\tr{x} (u) \in X_x$ for all $x\in V$ and the conditions $u\in \CC(\FX_{\widetilde{\Gamma}})$ and $u(C)=0$ for $C\in C_0^x$ for some $x\in V$ agree, which gives that $\Phi$ is one-to-one and onto as mapping from $\CD(\CS)$ to $\CD(\Fh)$. The equality of $\CS(u)$ and $\Fh(u)$ follows from the construction of the discrete part.\\
To derive the regularity statement we have to change the edges having free vertices as terminal or initial vertex. To do so, let $\widetilde{\FX}_{\widetilde{\Gamma}}$ be the metric space $\FX_{\widetilde{\Gamma}} \setminus (\bigcup_{x\in V} C_0^x)$. This is again a metric graph, but this time not complete. However, it is straightforward to show
\[\CD(\CS)= \overline{\CC_c(\widetilde{\FX}_{\widetilde{\Gamma}})\cap \bigoplus_{e\in E} W^{1,2}(0,l(e))}^{\|\cdot\|_\CS},\]
giving the regularity statement.
\end{proof}
We see from the proof of the previous theorem, that Dirichlet quantum graphs are coming from those Dirichlet forms on metric graphs, where the jump graph consists of an (in general) countable union of finite connected components. Each vertex in a Dirichlet quantum graph thus corresponds to a finite graph defined on a set of vertices coming from cutting a quantum graph vertex. Thus all the theory we develop for the graph Dirichlet forms is applicable to quantum graphs with Dirichlet form boundary conditions.

\chapter{Laplacians on metric graphs}
After having defined graph Dirichlet forms in the previous chapter, we come now to the associated operators. We use non-local Dirichlet form theory to define a weak Laplacian. This operator and associated equations will be reduced to a discrete problem in Theorem 3.6. This will be used to consider the interplay between the discrete and the continuous setting. In particular this approach leads to certain operators, the part Green operator and the harmonic extension operator, which are introduced and analyzed in section 2, Propositions 3.9 and 3.10. In particular, they appear when estimating adjoints of each other. This deep interplay is the content of Theorem 3.14. Connected with the harmonic extension are harmonic functions which are under analysis in section 3. There also Theorem 3.6 is generalized to a pointwise Krein type resolvent formula in Theorem 3.18. Moreover, classical results like maximum principle are considered, Theorem 3.22 and Lemma 3.23. There, a major role will play the functions which belong to the image of the harmonic extension. A third operator which is introduced in section 2 is the Kirchhoff Laplacian. The associated Dirichlet form will be analyzed in section 4. Here also the leitmotif $(C)$ appears, as the main focus there is the connecting between the continuous and the discrete. These ideas will be combined in section 6 to derive a Krein type resolvent formula, Theorem 3.47 and Theorem 3.48, connecting the continuous Laplacian on the metric graph and a weighted Laplacian on a discrete graph. Before doing so we introduce part Dirichlet forms in section 5 and concern approximating procedures there. This is essential in the estimation of the generator associated with the graph Dirichlet form, Theorem 3.44 and Theorem 3.45. Finally we finish section 6 with two results on essential self-adjointness, Theorem 3.49 and Theorem 3.50.
\section{The energy measure and functions locally in the domain}
Let $(\CS,\CD(\CS))$ be a graph Dirichlet form on a metric graph $\FX_\Gamma$ on $L^2(\FX_\Gamma,\lambda_a)$ as in the previous chapter, i.e.
\[\CS(u) = \CE(u) + \CQ(u)\]
where $\CE(u)$ is a diffusion Dirichlet form with edge weight $(b(e))_{e\in E}$, and $\CQ(u)$ is a discrete Dirichlet form with jump weight $(j(x,y))_{x,y\in \FV}$ and killing weight $(k(x))_{x\in \FV}$.\medskip\\
The energy measure as defined in the previous section is not limited to strongly local Dirichlet forms. A general observation in Dirichlet form theory, see \cite{FOT-11}, is that for $u,f\in \CD(\CS)_\infty$ the expression
\[\CS(uf,u)-\frac{1}{2}\CS(u^2,f)\]
is a linear, positive functional in $f$, and additionally it is bounded as one can estimate the right hand side by
\[\|f\|_\infty \cdot \CS(u).\]\index{Dirichlet form!energy measure}
Thus in the regular case the Riesz representation theorem gives that there exists a positive Radon measure $\mu_{\langle u\rangle}$ such that
\[\int\limits_{\FX_\Gamma} f d\mu_{\langle u\rangle}
= \CS(uf,u)-\frac{1}{2}\CS(u^2,f).\]
Using polarization for $u,v\in \CD(\CS)_\infty$ we introduce the signed Radon measure $\mu_{\langle u\rangle}$ by means of
\[\mu_{\langle u,v\rangle} = \frac{1}{2} (\mu_{\langle u+v\rangle}-\mu_{\langle u\rangle} -\mu_{\langle v\rangle}).\]
Hence we have the equality of for $f\in \CD(\CS)_c$
\[\int\limits_{\FX_\Gamma} f d\mu_{\langle u,v\rangle} = \frac{1}{2}\CS(uf,v)+\frac{1}{2}\CS(vf,u)-\frac{1}{2}\CS(uv,f).\]
A short calculation and splitting the energy measure for $\CS$ into its diffusion and its discrete part give\index{graph Dirichlet form!energy measure}
\[\mu_{\langle u,v\rangle} = \mu^\partial_{\langle u,v\rangle} + \mu^\delta_{\langle u,v\rangle}\]
with the differential part
\[d\mu^\partial_{\langle u\rangle} = |u'|^2 d\lambda_b\]
and the discrete part
\[\mu^\delta_{\langle u\rangle} = \sum_{x\in \FV}\sum_{y\in \FV} j(x,y)(u(x)-u(y))^2 \delta_x.\]
Now our aim is to extend the energy measure to a larger space of functions having the property that those functions locally agree with functions in the domain. For general Dirichlet forms the almost appropriate space is called the local space and is given as in chapter 1 as\index{Dirichlet form!local space}
\[\CD(\CS)_\loc = \{ u\in L^2_\loc(\FX_\Gamma) \mid \forall K\subset \FX_\Gamma \mbox{ open, rel. compact } \exists \phi\in \CD(\CS): \phi|_K = u|_K\}.\]
It is almost appropriate in the sense that functions in $\CD(\CS)_\loc$ locally agree with functions in $\CD(\CS)$. One more aim is to pair those functions with functions in $\CD(\CS)$ having compact support. However it might happen that due to the jump part this quantity is not well defined. For this reason we introduce the space of functions locally in the domain, which was introduced in \cite{FLW-12} in a more general framework.
\begin{definition}\index{graph Dirichlet form!local space}
The space of functions locally in the domain of $\CS$ is defined as
\[\CD(\CS)_\loc^* := \{u\in \CD(\CS)_\loc \mid \forall x\in \FV : \sum_{y\in \FV} j(x,y)(u(x)-u(y))^2 <\infty\},\]
\end{definition}
The following proposition is an obvious consequence of the definition in our case.
\begin{prop}
We have $u\in \CD(\CS)_\loc^*$ if and only if
\begin{itemize}
\item[(i)] $u$ is absolutely continuous,
\item[(ii)] $u'\in L^2_\loc(\FX_\Gamma)$
\item[(iii)] $u\in \CC_j(\FX_\Gamma)$.
\end{itemize}
Furthermore, $\mu_{\langle u \rangle}$ is a Radon measure.
\end{prop}
So far we have assumed that the form $\CS$ is regular. However we can actually drop this assumption by observing that the local spaces of the forms $\CS$ and $\CS_o$ agree, where $\CS_o$ is the restriction of $\CS$ to the $\|\cdot\|_\CS$ closure of the set $\CC_c(\FX_\Gamma)\cap\CD(\CS)$. With the space of functions locally in the domain in mind we can extend the Dirichlet form $\CS$ in the following way.
\begin{theorem}
Each function $u\in \CD(\CS)_\loc^*$ with compact support belongs to $\CD(\CS)$ and
\[\CS(u,v) := \int\limits_{\FX_\Gamma} d\mu^\partial_{\langle u,v\rangle} + \sum_{x\in \FV} \mu^\delta_{\langle u,v\rangle}\]
is well defined, whenever $u,v\in \CD(\CS)_\loc^*$ and either $u$ or $v$ has compact support.
\end{theorem}
For the proof, see the general case Theorem 3.4 in \cite{FLW-12}. We now turn our attention to the definition of some kind of a distributional operator.
\begin{definition}\index{weak Laplacian}
The weak Laplacian with respect to $\CS$ is defined as the operator
\[\widetilde{\Delta} : \dom \widetilde{\Delta} \to L^1_\loc (\FX_\Gamma)\]
on the domain given by
\[\dom \widetilde{\Delta} :=\{u\in \CD(\CS)_\loc^* \mid \exists f\in L^1_\loc(\FX_\Gamma) \forall \phi\in \CD(\CS)_c: \CS(u,\phi)=\langle f,\phi \rangle\},\]
on which it acts by
\[\widetilde{\Delta} u = f.\]
\end{definition}
By definition the weak Laplacian satisfies
\[\CS(u,v)= \langle \widetilde{\Delta} u , v\rangle\]
for all $u\in \CD(\CS)_\loc^*$ and $v\in \CD(\CS)_c$. Another application of the lemma of Du Bois-Reymond yields that this operator is well-defined in the sense that $f$ is the unique element satisfying $\widetilde{\Delta} u  = f$.
We describe the operator in the next theorem. Before we can state it we introduce the so-called Kirchhoff conditions.
\begin{definition}\index{Kirchhoff conditions}\index{discrete Laplacian}
An absolutely continuous function $u$ satisfies the $\triangle$-Kirchhoff conditions in $x\in  \FV$ if the weighted divergence
\[\partial_nu (x)= \sum_{e\sim x} b(e) u_e'(0)\]
and the weighted discrete Laplacian
\[\triangle u(x)= \sum_{y\in \FV} j(x,y)(u(x)-u(y)) + k(x)u(x).\]
of $u$ exist in the point $x\in \FV$ and are equal, i.e.
\[\partial_nu(x)= \triangle u(x)\]
holds.
\end{definition}
The following theorem tells us that the weak Laplacian is a distributional operator on the edges plus a discrete operator on the vertex set.
\begin{theorem}
Let $f\in L^1_\loc(\FX_\Gamma)$. Then $u\in \dom \widetilde{\Delta}$ is a weak solution, i.e.
\[\widetilde{\Delta} u = f\]
if and only if we have for all $e\in E$
\[-u_e'' b(e)  = f_e a(e)\]
in the distributional sense, and for all $x\in \FV$ the $\triangle$-Kirchhoff conditions
\[\partial_n u(x) = \triangle u(x)\]
are satisfied.
\end{theorem}
\begin{proof}
Let first $u$ be a weak solution, thus the equation
\[\CS(u,\phi)=\langle f,\phi \rangle\]
holds for all $\phi$ with compact support, and thus it holds for all with support on a fixed edge $e$. This gives that $\CQ(u,\phi)=0$ and the equation reduces to
\[\int\limits_0^{l(e)} u'_e \phi_e' d\lambda_b = \int\limits_0^{l(e)} f_e\phi_e d\lambda_a\]
and we get that $u$ is a distributional solution of $-u_e'' b(e)=f_e a(e)$. In order to obtain the $\triangle$-Kirchhoff conditions fix $x\in \FV$ and let $\phi$ be supported in a small neighborhood around $x$ containing no further vertex, and choose such a $\phi$ with $\phi(x)=1$ and $0\leq \phi\leq 1$. The set $\Lambda_x$ of all such functions forms a directed set by the partial order
\[\phi_1 \prec \phi_2 :\Leftrightarrow \supp \phi_2 \subset \supp \phi_1.\]
In particular, for two elements there exists an upper bound and therefore this set is directed. Thus $(\CS(u,\phi))_{\phi\in \Lambda_x}$ and $(\langle f,\phi \rangle)_{\phi\in \Lambda_x}$ are nets and furthermore
\[\lim_{\Lambda_x} \langle f,\phi \rangle=0\]
holds. Now, if $u$ is a weak solution of $\tilde{\Delta} u = f$ the limit of the net $(-\CE(u,\phi))_{\phi\in \Lambda_x}$ exists, since it equals
\[\sum_{y\sim x} j(x,y)(u(x)-u(y)) + k(x) u(x).\]
Thus by considering the subnet $(\phi_n)_{n\in \IN} \subset \Lambda_x$ such that $\phi_n(x)=1$, $\supp \phi_n \subset B_{\frac{1}{n}} (x)$ and linearly interpolated in-between, we get for this limit
\[ \lim_{n\to \infty} -\CE(u,\phi_n) =\lim_{n\to \infty} \sum_{e\sim x} n \int\limits_0^{\frac{1}{n}} (u'_e)_n (t) \:\nu(e) dt = \partial_n u(x)\]
i.e. the weak solution $u$ satisfies the $\triangle$-Kirchhoff conditions
\[\partial_n u(x) = \triangle u(x).\]
For the converse direction let $\phi\in \CD(\CS)$ with compact support. Then let $\tilde{\phi}\in \CD(\CS)$ such that $\phi- \tilde{\phi}$ vanishes on the vertex set. Since we have $-u''_e b(e) = f_e a(e)$ in the distributional sense, we have
\[\CS(u,\phi-\tilde{\phi})= \CE(u,\phi-\tilde{\phi}) = \langle f,\phi-\tilde{\phi} \rangle.\]
Thus it suffices to show the claim when $\tilde{\phi}$ is of the form $\phi_n$ as in the proof above. Using the $\triangle$-Kirchhoff conditions, for all $\epsilon>0$ there is $n \in \IN$ such that
\[|\CE(u, \phi_n) + \CQ(u,\phi_n)| = |\CE(u,\phi_n) + \triangle u(x)| =|\CE(u,\phi_n) + \partial_nu(x)| \leq \epsilon\]
since $-\CE(u,\phi_n) \to \partial_n u(x)$, as we have shown earlier in the proof. As $\epsilon$ does not depend on the choice of $\phi$ the proof is finished.
\end{proof}
The theorem will give us that all significant operators below are restrictions of the weak Laplacian -- at least in the regular case.
\begin{remark}
Given a transformation $\Phi:\FX_\Gamma \to \widetilde{\FX}_\Gamma$ one may wonder how this affects the weak Laplacian. Denoting with $\widetilde{\Delta}_\CS$ the operator coming from the original form $\CS$ and by $\widetilde{\Delta}_{\widetilde{\CS}}$ the operator coming from the image Dirichlet form $\widetilde{\CS}$, an easy calculation shows with $u\circ \Phi  =: \widetilde{u}$ that the following holds true:
\begin{itemize}
\item[(i)] $u\in \CD(\CS)^*_\loc$ if and only if $\widetilde{u} \in \CD(\widetilde{\CS})^*_\loc $,
\item[(ii)] $u\in \dom \widetilde{\Delta}_{\CS}$ if and only if $\widetilde{u} \in \dom \widetilde{\Delta}_{\widetilde{\CS}}$
\item[(iii)] $\widetilde{\Delta}_\CS u = f$ if and only if  $\widetilde{\Delta}_{\widetilde{\CS}} \widetilde{u} = \widetilde{f}$, for $f\in L_\loc^1(\FX_\Gamma)$ and $\widetilde{f}= f\circ \Phi$.
\end{itemize}
To see $(iii)$, we use the previous theorem and get that $\widetilde{\Delta}_\CS u = f$ is equivalent to
\[-u_e'' b(e)  = f_e a(e)\]
on each edge and
\[\sum_{e\sim x} b(e) u_e'(0) = \triangle u(x)\]
for each vertex. Now we apply the transformation rules for the weights $l(e),a(e),b(e)$ and get that both families of equations turn into
\[-u_e'' \tfrac{l(e)^2}{\tilde{l}(e)^2} \tilde{b}(e) = f_e \tilde{a}(e)\]
on each edge and
\[\sum_{e\sim x} \tilde{b}(e) u_e'(0)\tfrac{l(e)}{\tilde{l}(e)} = \triangle u(x).\]
By the chain rule for $\tilde{u}_e (t) = u_e(\tfrac{l(e)}{\tilde{l}(e)} t)$ we see that this is equivalent to $\widetilde{\Delta}_{\widetilde{\CS}} \widetilde{u} = \widetilde{f}$.\\
Especially point $(iii)$ is interesting for us as it gives us the possibility to choose an appropriate normal form of the graph. Choosing the canonical representation we obtain unweighted Kirchhoff conditions, or more precisely the divergence is unweighted. However the edgewise equation stays weighted. Choosing the intrinsic representation we get an unweighted equation on each edge but keep the weighted divergence.
\end{remark}
\section{The weak Laplacian under the magnifying glass}
From the previous section we see that existence and uniqueness of the weak equations are related to a discrete problem which is phrased in terms of the continuity conditions and the $\triangle$-Kirchhoff conditions. In the situation of a compact graph, the equation above results in a finite dimensional problem and thus solvability of the equation is completely determined by the discrete problem. The situation is different in the non-compact case and will be analyzed below.\medskip\\
In this section we consider weak equations of the type
\[\widetilde{\Delta} u + \alpha u =f.\]
To be more precisely we start with the following definition.
\begin{definition}\index{weak Laplacian!weak solution}
Let $f\in L^1_\loc(\FX_\Gamma)$ and $\alpha\geq0$. Then a function $u\in \dom \widetilde{\Delta}$ is called a weak $\alpha$-solution if
\[\CS(u,\phi) + \alpha \langle u,\phi\rangle = \langle f, \phi\rangle\]
holds for all $\phi\in \CD(\CS)$ with compact support.
\end{definition}
Our aim is to reduce this equation to a purely discrete problem. For this purpose we choose the intrinsic representation. Assume that $u$ is a weak $\alpha$-solution and let $e=(x,y)\in E$. Then we know from the previous section that $u_e$ is a distributional solution of
\[-u_e'' + \alpha u_e = f_e.\]
This could be solved by use of Green's functions. Let therefore $h_e^\alpha$ be the function on $[0,l_i(e)]$ defined by
\[h_e^\alpha(t) := \frac{\sinh(\sqrt{\alpha} (l_i(e)-t))}{\sinh(\sqrt{\alpha} l_i(e))}.\]
for $\alpha >0$ and
\[h_e^0 (t):= \frac{l_i(e)-t}{l_i(e)}.\]
It is clear that the functions $h_e^\alpha(t)$ and $h_e^\alpha(l_i(e)-t)$ form a fundamental system of the above differential equation with $h_e^\alpha(0)=1$ and $h_e^\alpha(l_i(e))=0$. The Green function is then defined as
\[G^\alpha_e(t,s)= -\frac{1}{{h_e^\alpha}'(l_i(e))} \begin{cases} h_e^\alpha(s)h_e^\alpha(l_i(e)-t) & \mbox{for } 0\leq t\leq s\leq l_i(e)\\ h_e^\alpha(t) h_e^\alpha(l_i(e)-s) & \mbox{for } 0\leq s\leq t \leq l_i(e) \end{cases}.\]
Thus the unique solution of $-u_e'' + \alpha u_e = f_e$ with boundary values $u_e(0)=u(x)$ and $u_e(l_i(e))=u(y)$ can be represented as
\begin{eqnarray*}
u_e(t) &=& \int\limits_0^{l_i(e)} G^\alpha_e(t,s) f_e(s) ds + u(x) h_e^\alpha(t) + u(y) h_e^\alpha(l_i(e)-t)\\
&=& h_e^\alpha(l_i(e)-t) \Bigl(u(y) - \int\limits_t^{l_i(e)} \frac{h_e^\alpha(s)}{{h_e^\alpha}'(l_i(e))} f_e(s)ds\Bigr)\\ && + h_e^\alpha(t)\Bigl(u(x) - \int\limits_0^t \frac{h_e^\alpha(l_i(e)-s)}{{h_e^\alpha}'(l_i(e))} f_e(s)ds\Bigr).
\end{eqnarray*}
From the first equality we see that the solution is given as a sum of two functions, viz, the unique solution of the equation $-u_e'' + \alpha u_e = f_e$ with zero vertex conditions given by
\[\int\limits_0^{l_i(e)} G^\alpha_e(t,s) f_e(s) ds,\]
plus a function which is a certain interpolation of a function living on the vertex set $\FV$, given in edge coordinates by
\[u(x) h_e^\alpha(t) + u(y) h_e^\alpha(l_i(e)-t).\]
This gives rise to the next definition.
\begin{remark}
Note that the set $\FV$ with the trace topology gives a topological space equipped with the discrete topology, and thus the set of all continuous functions on $\FV$ equals the set of all functions on $\FV$. In what follows we will denote functions defined on the vertex set by capital letters.
\end{remark}
\begin{definition}
Let $\FX_\Gamma$ be a metric graph and $\CS$ a graph Dirichlet form.
\begin{itemize}\index{part Green operator}
\item[(i)] The part Green operator is defined as
\[G_\alpha: L^1_\loc (\FX_\Gamma) \to \CC(\FX_\Gamma),\]
such that
\[(G_\alpha f)_e(t) = \int\limits_0^{l_i(e)} G^\alpha_e (t,s)f_e(s) ds\]
for all $e\in E$.
\item[(ii)] The $\alpha$-harmonic extension is defined as\index{harmonic extension}
\[H_\alpha: \CC(\FV) \to \CC(\FX_\Gamma),\]
such that
\[(H_\alpha U )_e(t) = U(\partial^+(e)){h_e^\alpha}(t) + U(\partial^-(e)) {h_e^\alpha}(l_i(e)-t).\]
\end{itemize}
\end{definition}
We start with the part Green operator. The reason why we call it part Green operator is that it coincides with the part of the Dirichlet form $\CS$ on the subspace of $\CD(\CS)$ where all functions are assumed to  be zero on the set of vertices. We will come back to this in chapter 4.
\begin{prop}
The part Green operator is a contraction on $L^p$ for all $1\leq p \leq \infty$. Furthermore for $f\in L^1_\loc(\FX_\Gamma)$ we have that $(G_\alpha f)_e:(0,l_i(e)) \to \IR$ is two times differentiable.
\end{prop}
\begin{proof}
We start with the case $p=1$. For $f\in L^1(\FX_\Gamma,\omega)$ we obtain for $e\in E$
\begin{eqnarray*}
\int\limits_0^{l_i(e)}|(G_\alpha f)_e(t)| \omega(e) dt &\leq& \int\limits_0^{l_i(e)}\int\limits_0^{l_i(e)} |G_e^\alpha(t,s)| |f_e(s)| \omega(e) ds\: dt\\
&\leq& \sup_{t,s\in(0,l_i(e))} |G_e^\alpha (t,s)| l_i(e) \int\limits_0^{l_i(e)} |f_e(s)|\omega(e) ds
\end{eqnarray*}
Since $h_e^\alpha(t) \leq 1$ and $|{h_e^\alpha}' (l_i(e))| \geq l_i(e)$ we obtain
\[\sup_{t,s\in(0,l_i(e))} |G_e^\alpha (t,s)| l_i(e) \leq 1\]
and thus
\[\sum_{e\in E} \|(G_\alpha f)_e\|_1 \leq \sum_{e\in E} \|f_e\|_1.\]
In the case $p=\infty$ we have for $e\in E$
\[|(G_\alpha f)_e(t)| \leq \int\limits_0^{l_i(e)} |G_e^\alpha(t,s)| ds \sup_{s} |f_e(s)| \leq \sup_{s\in (0,l_i(e))} |f_e(s)|\]
and thus $G_\alpha$ is a bounded operator on $L^\infty$. Applying the Riesz Thorin interpolation theorem gives that $G_\alpha$ is also bounded on $L^p$.\\
The second claim follows easily from the local integrability of $f$, as it allows the interchange of differentiation and integration.
\end{proof}
Our next aim is to prove similar results for $H_\alpha$. The question is for which measure $M$ is the discrete space $\ell^p(\FV,M)$ mapped into $L^p(\FX_\Gamma, \omega)$. The next proposition shows, that the measure defined by
\[M(x) := \sum_{e\sim x} l_i(e)\omega(e)\]
for all $x\in \FV$ does the job. Note that this quantity is invariant under a length transformation $\Phi$, as it equals the measure of the set $\bigstar(x)$.
\begin{prop}
For all $\alpha\geq0$ and $p\in [1,\infty]$ the operator
\[H_\alpha: \ell^p(\FV, M) \to L^p(\FX_\Gamma, \omega)\]
is bounded. Furthermore, for all $U\in \CC(\FV)$ and $e\in E$ the function $(H_\alpha U)_e : (0,l_i(e)) \to \IR$ is differentiable infinitely often.
\end{prop}
\begin{proof}
Let first $p<\infty$, then for $U\in \ell^p(\FV,M)$ we have for $e\in E$
\begin{eqnarray*}
\int\limits_0^{l_i(e)} |(H_\alpha U)_e (t)|^p \omega(e) dt &\leq& c\Bigl(U(\partial^+(e))^p \int\limits_0^{l_i(e)} h^\alpha_e (t)^p \omega(e)dt \\
&& + U(\partial^-(e))^p \int\limits_0^{l_i(e)} h^\alpha_e (l_i(e)-t)^p \omega(e)dt\Bigr) \\
&\leq&c( U(\partial^+(e))^p l_i(e)\omega(e) + U(\partial^-(e))^p l_i(e)\omega(e)).
\end{eqnarray*}
We sum over all edges $e\in E$ and recall that this is the same as summing over all vertices $x\in \FV$ and all adjacent edges to $x$, i.e.
\begin{eqnarray*}
\sum_{e\in E} \|(H_\alpha U)_e\|_p^p &=& \frac{1}{2} \sum_{x\in \FV} \sum_{e\sim x} \|(H_\alpha U)_e\|_p^p\\
&\leq & c \Bigl(\sum_{x\in \FV} |U(x)|^p \sum_{e\sim x} l_i(e)\omega(e) + \sum_{e\sim x}  |U(\partial^-(e))|^p l_i(e)\omega(e)\Bigr).
\end{eqnarray*}
The first summand equals $\sum_{x\in \FV} |U(x)|^p M(x)$ and in the second sum a reordering of the sum gives that this is equal to
\[\sum_{y\in \FV}\sum_{e\sim y}  |U(y)|^p l_i(e)\omega(e).\]
Thus we have shown that
\[\sum_{e\in E} \|(H_\alpha U)_e\|_p^p \leq c \sum_{x\in \FV} |U(x)|^p M(x).\]
For $p=\infty$ we first note that
\[|(H_\alpha U)_e (t)| \leq (H_\alpha |U|)_e (t)\]
and by the maximum principle for positive $\alpha$-harmonic functions on an interval we have
\[\sup_{t} (H_\alpha |U|)_e (t) \leq \max \{|U|(\partial^+(e)), |U|(\partial^-(e))\}\]
which gives immediately
\[\| H_\alpha U\|_\infty \leq \|U\|_\infty.\]
The second claim on the differentiability is obvious.
\end{proof}
We now summarize the previous results and discussion.
\begin{theorem}
Let $\FX_\Gamma$ be a metric graph and $\CS$ a graph Dirichlet form. Let $f\in L^1_\loc(\FX_\Gamma)$, $\alpha \geq 0$ and let $u$ be a weak $\alpha$-solution of the equation $\widetilde{\Delta} u + \alpha u = f$. Then we have
\[u = G_\alpha f + H_\alpha (u|_\FV).\]
In particular, for all $x\in \FV$ the weighted divergence of $u$ in $x$ exists and is given by
\begin{eqnarray*}
\partial_n u(x) &=& \sum_{e\sim x} \int\limits_0^{l_i(e)} h_e^\alpha(s) f_e(s)\omega(e) ds \\
&&- \sum_{e\sim x} - {h_e^\alpha}'(l_i(e))\omega(e)(u(x)-u(\partial^-(e)))\\
&& - u(x)\sum_{e\sim x} \omega(e)({h_e^\alpha}'(l_i(e)) - {h_e^\alpha}'(0)).
\end{eqnarray*}
\end{theorem}
\begin{proof}
The representation
\[u = G_\alpha f + H_\alpha (u|_\FV),\]
follows from the discussion in the beginning of the section. The previous results give the differentiability of each component $u_e$ and we obtain
\begin{eqnarray*}
u_e'(t) &=& \int\limits_0^{l_i(e)} \tfrac{\partial}{\partial t} G_e^\alpha(t,s) f_e(s) \:ds  +u(x){h_e^\alpha}'(t) - u(y){h_e^\alpha}'(l_i(e)-t)\\
&=& \frac{{h_e^\alpha}'(l_i(e)-t)}{{h_e^\alpha}'(l_i(e))}\int\limits_t^{l_i(e)} h_e^\alpha(s)f_e(s)ds-\frac{{h_e^\alpha}'(t)}{{h_e^\alpha}'(l_i(e))} \int\limits_0^{t} h_e^\alpha(l_i(e)-s) f_e(s)ds\\
&&+u(x){h_e^\alpha}'(t) - u(y){h_e^\alpha}'(l_i(e)-t).
\end{eqnarray*}
Thus, we obtain for the divergence of $u$ in $x\in \FV$
\begin{eqnarray*}
\partial_n u(x) &=& \partial_n G_\alpha f (x) + \partial_nH_\alpha(u|_\FV)\\
&=& \sum_{e\sim x} \int\limits_0^{l_i(e)} h_e^\alpha(s) f_e(s)\omega(e) ds \\
&&- \sum_{e\sim x} - {h_e^\alpha}'(l_i(e))\omega(e)(u(x)-u(\partial^-(e)))\\
&& - u(x)\sum_{e\sim x} \omega(e)({h_e^\alpha}'(l_i(e)) - {h_e^\alpha}'(0)).
\end{eqnarray*}
\end{proof}
We have already split the sum of the weighted divergence into three parts by intention, as each of them needs to be analyzed.
\begin{lemma}
The operator
\[M^{-1}\partial_n G_\alpha: L^1_\loc(\FX_\Gamma)\to \CC(\FV),\]
is given by
\[M^{-1}\partial_n G_\alpha f(x) = \tfrac{1}{M(x)}\sum_{e\sim x} \int\limits_0^{l_i(e)} h_e^\alpha (s) f_e(s) \omega(e) \: ds \]
and for all $p\in [1,\infty]$ it is a bounded operator from $L^p (\FX_\Gamma,\omega)$ to $\ell^p (\FV,M)$.
\end{lemma}
\begin{proof}
Let $f\in L^1(\FX_\Gamma,\omega)$ and $x\in \FV$. Then we have
\begin{eqnarray*}
|\partial_n G_\alpha f (x)| &\leq & \sum_{e\sim x} \int\limits_0^{l_i(e)} h_e^\alpha (s)|f_e(s)| \omega(e) \:ds\\
&\leq& \sum_{e\sim x} \int\limits_0^{l_i(e)} |f_e(s)| \omega(e) \: ds.
\end{eqnarray*}
Thus,
\begin{eqnarray*}
\|M^{-1} \partial_n G_\alpha f\|_1 &=& \sum_{x\in \FV} |M^{-1} \partial_n G_\alpha f (x)| M(x)\\
&=& \sum_{x\in \FV} |\partial_n G_\alpha f (x)|\\
&\leq & \sum_{x\in \FV} \sum_{e\sim x} \|f_e\|_1\\
&=& \|f\|_1.
\end{eqnarray*}
Let now $f\in L^\infty(\FX_\Gamma,\omega)$ and $x\in \FV$. Then
\[ |\partial_n G_\alpha f(x)| \leq \|f_e\|_\infty \sum_{e\sim x} l_i(e)\omega(e)\]
which immediately gives
\[\|M^{-1} \partial_n G_\alpha f\|_\infty \leq \|f\|_\infty.\]
The Riesz-Thorin interpolation theorem gives now the boundedness for each $p\in (1,\infty)$.
\end{proof}
The operator in the previous lemma is closely connected with the operator $H_\alpha$.
\begin{theorem}
Let $p,q\in[1,\infty]$ with $1=\tfrac{1}{p} + \tfrac{1}{q}$.Then for all $f\in \ell^p (\FV,M)$ and $g\in L^q(\FX_\Gamma,\omega)$ we have the relation
\[\langle H_\alpha f , g\rangle = \langle f, M^{-1} \partial_n G_\alpha g \rangle,\]
where $\langle\cdot,\cdot\rangle$ denotes the corresponding duality pairing.
\end{theorem}
\begin{proof}
Let $g\in L^q(\FX_\Gamma,\omega)$ and $\delta_x$ the function being zero for all $y\in \FV\setminus \{x\}$ and $\delta_x(x)=1$. We then have
\[H_\alpha^* g(x) =  \langle \tfrac{\delta_x}{M(x)},  H_\alpha^* g\rangle =\langle H_\alpha \tfrac{\delta_x}{M(x)},   g\rangle.\]
As $\delta_x$ is supported only in $\{x\}$ its harmonic extension is supported in $\bigstar(x)$ and for $e\sim x$ we obtain
\[(H_\alpha \tfrac{\delta_x}{M(x)})_e(t)= \tfrac{1}{M(x)} h_e^\alpha(t)\]
and thus
\begin{eqnarray*}
\langle H_\alpha \tfrac{\delta_x}{M(x)},   g\rangle &=& \tfrac{1}{M(x)}\sum_{e\sim x} \int\limits_0^{l_i(e)} h_e^\alpha (s) g(s) \omega(e) \:ds \\
&=&\tfrac{1}{M(x)}\partial_n G_\alpha g (x)
\end{eqnarray*}
which hold for all $x\in \FV$.
\end{proof}
The previous theorem gives the connection between graphs Dirichlet forms and discrete Dirichlet forms. This will be exploited in greater detail in the last section of this chapter. The theorem above implies the following.
\begin{coro}
Let $p\in [1,\infty)$ and $q$ such that $1=\tfrac{1}{p} + \tfrac{1}{q}$. Then the adjoint of the operator
\[H_\alpha:\ell^p (\FV,M) \to L^p(\FX_\Gamma,\omega)\]
is given by
\[M^{-1}\partial_n G_\alpha:L^q (\FX_\Gamma,\omega) \to \ell^q (\FV,M),\]
and the adjoint of the operator
\[M^{-1}\partial_n G_\alpha:L^p (\FX_\Gamma,\omega) \to \ell^p (\FV,M)\]
is given by the operator
\[H_\alpha:\ell^q (\FV,M) \to L^q(\FX_\Gamma,\omega).\]
\end{coro}
We continue with the remaining parts in the divergence.
\begin{prop}
For all $\alpha >0$ there exists $c>0$ such that
\[c^{-1} \frac{1}{l_i(e)} \leq -{h_e^\alpha}' (l_i(e)) \leq c \frac{1}{l_i(e)} \]
and
\[ c^{-1} l_i(e) \leq {h_e^\alpha}'(l_i(e)) - {h_e^\alpha}' (0) \leq C l_i(e).\]
\end{prop}
\begin{proof}
This follows immediately from
\[-{h_e^\alpha}' (l_i(e)) = \frac{\sqrt{\alpha}}{\sinh(\sqrt{\alpha} l_i(e))}\]
and
\[{h_e^\alpha}'(l_i(e)) - {h_e^\alpha}' (0) = \sqrt{\alpha} \frac{\cosh(\sqrt{\alpha} l_i(e))-1}{\sinh(\sqrt{\alpha} l_i(e))}.\]
\end{proof}
The previous proposition legitimates the following definition.
\begin{definition}\index{Kirchhoff Laplacian}
For $\alpha \geq 0$ the Kirchhoff Laplacian is the operator defined as
\[\triangle_\alpha: \CC(\FV) \to \CC(\FV),\]
acting as
\[\triangle_\alpha U(x)=  \sum_{e\sim x} - {h_e^\alpha}'(l_i(e))\omega(e)(U(x)-U(\partial^-(e))).\]
The vertex potential $M_\alpha:\FV \to (0,\infty)$ for $x\in \FV$ is defined as
\[M_\alpha (x) = \sum_{e\sim x} ({h_e^\alpha}'(l_i(e)) - {h_e^\alpha}'(0))\omega(e).\]
\end{definition}
From the definition we see that in the case of $\alpha =0$ the vertex potential vanishes.
We immediately obtain the following proposition.
\begin{prop}
For $\alpha\geq0$ we have for $x\in \FV$ and $U\in \CC(\FV)$
\[\partial_n H_\alpha U(x)=-\triangle_\alpha U(x)-M_\alpha(x) U(x).\]
\end{prop}
\begin{remark}
As mentioned in section 2.3 we have already noted that in \cite{KL-11} the weighted Laplace was defined on the set $F$ consisting of all functions $U$ such that for all $x\in \FV$ the sum
\[\sum_{y\in \FV} b(x,y) |U(y)|\]
is finite, where $b(x,y)$ denotes the edge weights. In our case this agrees with $\CC(\FV)$ as the sum is a finite sum due to the local finiteness of the metric graph.
\end{remark}
With a little abuse of notation we will also apply the Kirchhoff Laplacian to functions defined on $\FX_\Gamma$, which is possible by applying the Kirchhoff Laplacian to the restriction of the function to the set of vertices.\medskip\\
We know from the previous section that weak $\alpha$-solutions satisfy the $\triangle$-Kirchhoff conditions. Rephrased with the operators above, we obtain that for weak $\alpha$-solutions, the $\triangle$-Kirchhoff conditions are given by
\[\triangle u(x) = \partial_n u(x) = \partial_nG_\alpha f (x) - \triangle_\alpha u (x) - M_\alpha (x) u(x).\]
This is an inhomogeneous linear difference equation of second order, and using the part Green operator and the $\alpha$-harmonic extension operator we can actually solve the original problem, if we can solve this discrete equation. We finish this section with a pre-version of the Krein resolvent formula. It relates weak solutions of the graph Dirichlet form with weak solutions of a discrete Dirichlet form.
\begin{theorem}
Let $f\in L^1_\loc(\FX_\Gamma)$ and $\alpha\geq 0$. If $u\in \dom \widetilde{\Delta}$ is a weak $\alpha$-solution, i.e.
\[\widetilde{\Delta} u + \alpha u = f,\]
then its restriction to $\FV$ satisfies the pointwise discrete operator equation
\[\triangle u(x) + \triangle_\alpha u(x) +M_\alpha (x) u(x)= \partial_nG_\alpha f(x).\]
Conversely, if $U\in \CC(\FV)$ satisfies this pointwise discrete operator equation, then the function defined by
\[u(x) = G_\alpha f (x) + H_\alpha U (x)\]
is a weak $\alpha$-solution.
\end{theorem}
\begin{proof}
If $u\in \dom \widetilde{\Delta}$ is a weak $\alpha$-solution, then by Theorem 3.6 it satisfies the $\triangle$-Kirchhoff conditions
\[\partial_n u(x) = \triangle u(x)\]
for all $x\in \FV$. By Theorem 3.11 we get for the weighted divergence
\[\partial_n u(x) = \partial_n G_\alpha f - \triangle_\alpha u(x) - M_\alpha (x) u(x).\]
Both equations together yield that $u|_\FV$ satisfies the discrete equation. Conversely, if $U\in \CC(\FV)$ is a solution to the discrete problem, then by the representation formula we have that for each $e=(x,y)$ the component $u_e$ of
\[u(x) = G_\alpha f (x) + H_\alpha U (x)\]
satisfies $-u_e'' + \alpha u_e = f$ and $u_e$ has boundary values $u_e(0) = u(x)$ and $u_e(l(e))= u(y)$. Furthermore the discrete equation above are the $\triangle$-Kirchhoff conditions for the function $u$. Thus it is a weak solution of $\widetilde{u} + \alpha u = f$ by Theorem 3.6.
\end{proof}
So far we have obtained an existence result, as we have reduced the continuous equation to a discrete one. In the next section we consider the homogeneous equation which is connected with the uniqueness of the equation.
\section{Harmonic functions}
In this section we study solutions of the weak equation where there right hand side $f$ is identically zero. This motivates the following definition.
\begin{definition}\index{harmonic functions!$\alpha$-sub / superharmonic}
Let $\FY\subset \FX_\Gamma$ be open and $\alpha \geq0$.
Then a function $u\in \CD(\CS)_\loc^*$ is called $\alpha$-subharmonic in $\FY$ if
\[\CS(u,\phi) + \alpha \langle u,\phi\rangle \leq 0\]
holds for all nonnegative $\phi\in \CD(\CS)$ with compact support in $\FY$. Conversely, if $-u$ is $\alpha$-subharmonic in $\FY$ we call $u$ $\alpha$-superharmonic. If $u$ is both, $\alpha$-subharmonic and $\alpha$-superharmonic in $\FY$, we call it $\alpha$-harmonic in $\FY$. In this case $u\in \dom \widetilde{\Delta}$ and $u$ is a weak $\alpha$-solution in $\FY$ of the equation $\widetilde{\Delta}u + \alpha u =0$.
\end{definition}
As mentioned in the definition, a function $u\in \CD(\CS)_\loc^*$ which is $\alpha$-harmonic is already in the domain of the weak Laplacian. This is not true for general $\alpha$-sub- and superharmonic functions. From the definition we get that if $u\in \CD(\CS)_\loc^*$ is $\alpha$-superharmonic, then by the Riesz representation theorem there exists a positive Radon measure $\mu$ such that
\[\CS(u,\phi) + \alpha \langle u,\phi\rangle = \mu(\phi).\]
Sub- and super-harmonic functions play an important role in potential theory. It is an important observation that we can restrict to a certain subclass of them, given in terms of the $\alpha$-harmonic extension. We begin by calculating the Radon measure which is connected with the weak Laplacian of the $\alpha$-harmonic extension of functions defined on the set of vertices.
\begin{prop}\index{Kirchhoff conditions! sub-$\triangle$ / super-$\triangle$}
Let $U\in \CC(\FV)$. Then we have for $u:= H_\alpha U$
\[\CS(u,\phi) +\alpha \langle u, \phi\rangle = \mu(\phi)\]
where $\mu$ is the signed Radon measure defined by
\[\mu(\phi)= \sum_{x\in \FV} (\triangle_\alpha U(x) + M_\alpha (x)U(x) + \triangle U(x)) \delta_x(\phi).\]
In particular, the  $\alpha$-harmonic extension  of $U\in \CC(\FV)$ is $\alpha$-subharmonic (resp. $\alpha$-superharmonic) if and only if it satisfies the sub-$\triangle$-Kirchhoff conditions (resp. super-$\triangle$-Kirchhoff conditions) given by
\[ \triangle_\alpha u(x)+ M_\alpha (x) u(x) + \triangle u(x) \leq 0 \mbox{ (resp. $\geq$)}. \]
\end{prop}
\begin{proof}
Let $x\in \FV$ and $\phi_x\CD(\CS)_c$ be a function with support in the star neighborhood of $x$. Then by integration by parts and since $H_\alpha U$ is $\alpha$-harmonic on each edge, we obtain that $\mu$ has the desired representation. The remaining claims thus follow.
\end{proof}
The proposition gives also a one-to-one correspondence between $\alpha$-harmonic functions defined on the continuous graphs and $1$-harmonic functions of the discrete problem, i.e.
\[ M_\alpha (x)^{-1}(\triangle_\alpha + \triangle) u(x)+ u(x) =0.\]
Unfortunately there is no one-to-one correspondence between continuous and discrete subharmonic functions. However as we will show next, when considering $\alpha$-subharmonic functions it is sufficient to study those which are harmonic on each edge. It is basically a generalization of the maximum principle on an interval.
\begin{prop}\index{harmonic extension!maximum principle}
Let $\alpha\geq 0$ and $u$ be $\alpha$-subharmonic. Then $H_\alpha(u|_\FV)$ is $\alpha$-subharmonic and $H_\alpha(u|_\FV)\geq u$.
\end{prop}
\begin{proof}
We first show the second claim. We set $v:=H_\alpha(u|_\FV)$. By definition $v$ is $\alpha$-harmonic on each edge and thus $u-v$ is an $\alpha$-subsolution on each edge with $(u-v)|_\FV \equiv 0$. Thus by the minimal principle of the ODE we obtain $u-v \leq 0$. To show the first claim, let $\phi \in \CC^\infty_c(\FX_\Gamma)$ supported in a small neighborhood of a vertex $x\in \FV$. We have to show that
\[\CS(v,\phi) + \alpha (v,\phi) \leq 0\]
which is in turn equivalent to
\[-\partial_n v(x) + \triangle v(x) \leq 0\]
due to the previous proposition.
By definition we have $\triangle v = \triangle u$. Since $u\leq v$ and $u = v$ on $\FV$ we have
\[\frac{v_e(t)- v(x)}{t} \geq \frac{u_e(t)- u(x)}{t}\]
and thus
\[\frac{1}{t} \int\limits_0^t v_e'(t) \:dt \geq \frac{1}{t} \int\limits_0^t u_e'(t) \:dt\]
hence $\partial_n v(x) \geq \partial_n u(x)$.
This yields
\[-\partial_n v(x) + \triangle v(x) \leq - \partial_n u(x) + \triangle u(x) \leq 0\]
where the last inequality follows since $u$ is $\alpha$-subharmonic.
\end{proof}
We now come to a first uniqueness result, which is actually a version of the maximum principle. It involves ideas in the discrete setting from \cite{KL-10}.
\begin{theorem}\index{harmonic functions!$L^p$ maximum principle}
Let $\alpha>0$, $1\leq p<\infty$ and assume that $u\in L^p(\FX_\Gamma,\omega)$ is $\alpha$-harmonic. Assume that for all rays $\Fp=(p_0,p_1,\dots)$ we have
\[\sum_{n=1}^\infty \omega(p_k,p_{k+1})l_i(p_k,p_{k+1})=\infty,\]
then $u\equiv 0$.
\end{theorem}
\begin{proof}
Since $u$ is $\alpha$-harmonic it satisfies the $\triangle$-Kirchhoff conditions, that is
\[\triangle_\alpha u(x)+ M_\alpha u(x) + \triangle u(x) = 0.\]
Choose for abbreviation $b(x,y)>0,c(x)>0$ such that
\[\triangle_\alpha u(x)+ M_\alpha u(x) + \triangle u(x) = \sum_{y\in \FV} b(x,y) (u(x)-u(y)) + c(x) u(x).\]
Assume that there is some $x_0$ with $u(x_0)<0$. Since
\[0\leq \sum_{y\in \FV} b(x_0,y) (u(x_0)-u(y)) + c(x_0) u(x_0)\]
there has to be an $x_1\sim x_0$ with $u(x_1)<u(x_0)$ and inductively a sequence $(x_n)_n$ of connected points with $u(x_n)<u(x_0) <0$. Then with $e_n=(x_n,x_{n+1})$ we have
\begin{eqnarray*}
\|u\|_p^p &\geq & \sum_{n=1}^\infty \|u_{e_n}\|_p^p\\
&=&\sum_{n=1}^\infty \int\limits_0^{l_i(e_n)} |u(x_n)\tfrac{\sinh ( \sqrt{\alpha}(l_i(e_n)  - t) )}{\sinh ( \sqrt{\alpha} l_i(e) )} + u(x_{n+1}) \tfrac{\sinh( \sqrt{\alpha} t )}{\sinh ( \sqrt{\alpha} l_i(e) )}|^p \omega(e_n) dt\\
&\geq&C \sum_{n=1}^\infty  |u(x_0)| \Bigl(\frac{\sinh \tfrac{l_i(e_n)}{2}}{\sinh l_i(e_n)}\Bigr)^p \int\limits_{-\frac{l_i(e_n)}{2}}^{\frac{l_i(e_n)}{2}} (\cosh t)^p \omega(e_n) dt\\
&\geq& C \sum_{n=1}^\infty l_i(e_n)\omega(e_n)
\end{eqnarray*}
where we have used that $\cosh t \geq 1$ and $\tfrac{\sinh \tfrac{l_i(e_n)}{2}}{\sinh l_i(e_n)}\geq C$ when $l_i(e_n)\leq 1$. Due to our assumption on the measures and the lengths the right hand side diverges, which contradicts $u\in L^p$.\\
Analogously we show that the case $u(x_0)>0$ cannot appear. Thus we have shown $u\equiv 0$.
\end{proof}
Let us note that the assumption is somehow natural as it is actually equivalent to all unbounded sets having infinite measure. If the intrinsic measure satisfies $\inf\limits_{e\in E} \omega(e) >0$, then metric completeness implies the assumption in the theorem.\medskip\\
We now turn to a local version of the maximum principle which will be of use in the next section.
\begin{lemma}\index{harmonic functions!local maximum principle}
Let $\FX_\Gamma$ be a metric graph and $\CS$ a graph Dirichlet form. Let $\FY\subset \FX_\Gamma$ precompact be given. Assume that the function $u$ on $\FX_\Gamma$ satisfies
\begin{itemize}
\item[(i)] $u$ is $\alpha$-superharmonic in $\FY$ for some $\alpha>0$,
%\item[(ii)] $u$ is bounded from below on $\FY$,
\item[(ii)] $u\geq 0$ on $\FX_\Gamma \setminus \FY$.
\end{itemize}
Then $u\equiv 0$ or $u>0$ on each connected component of $\FY$. In particular $u\geq 0$.
\end{lemma}
\begin{proof}
First we see that if $u$ satisfies these assumptions, then also $H_\alpha (u|_\FV)$ does. Furthermore we know that $u\geq H_\alpha (u|_\FV)$. Thus it again  suffices to show it for the case $H_\alpha (u|_\FV) = u$, which we will now assume.\\
Choose $x_0\in \FX_\Gamma$ such that $u$ attains a negative minimum there. As such an $u$ cannot attain a minimum in the interior of an edge, we will assume that $x_0\in \FV$. Thus we have for all $y\sim x_0$ that $u(x_0) \leq u(y)$ when $y\in \FY$ and also for $y\in \FX_\Gamma \setminus \FY$ as $u$ is nonnegative there. Since $u$ is $\alpha$-superharmonic it satisfies the super-$\triangle$-Kirchhoff conditions in $x_0$
\[\triangle_\alpha u(x_0) + M_\alpha (x_0) u(x_0) + \triangle u(x_0) \geq 0.\]
As all differences $u(x_0)- u(y)$ and $u(x_0)$ are non-positive, we obtain a contradiction and get that $u(x_0)$ is nonnegative. By an irreducibility argument we obtain  $u\equiv 0$ or $u>0$ on each connected component of $\FY$.
\end{proof}
The existence of certain harmonic functions is closely related to  the geometry of the graph and the global behavior of the parameters involved in the graph Dirichlet form. Unfortunately, the connection is hard to establish for general forms. It is however possible to do this for the special class of spherically symmetric graphs. Recall from Definition 1.14, that a spherically symmetric graph has a root $o\in \FX_\Gamma$ with the property that the inward and outward degree function is constant on the spheres $S_r(o)$. The next proposition shows an average property of spherically symmetric graphs if the weights of $\CS$ are additionally radial. Recall from Definition 1.16, that a weight $\omega(x,y)$ is called radial with respect to $o$ if
\[d(x_1,o)=d(x_2,o) \mbox{ and } d(y_2,o)=d(y_2,o) \Longrightarrow \omega(x_1,y_1)=\omega(x_2,y_2).\]
\begin{prop}
Let $\FX_\Gamma$ be a spherically symmetric graph and $\CS$ a graph Dirichlet form on $L^2(\FX_\Gamma, \nu)$ such that the weights $j,k,\nu$ are radial. If for $\alpha\geq0$ there exists an $\alpha$-harmonic function on $\FX_\Gamma$, then there exists a radial $\alpha$-harmonic function, i.e. a function which is constant on the spheres with respect to the root $o\in \FX_\Gamma$.
\end{prop}
\begin{proof}
The proof follows easily by direct calculation. Let $(r_j)_j$ be the sequence associated to the radial vertex set and let $S_j$ be the number of edges whose terminal vertex has distance $r_j$ to the root, whereas $S_{r_j}$ denotes the sphere with radius $r_j$. Let $u$ be $\alpha$-harmonic. Define $v:\FX_\Gamma \to \IR$ by averaging the function $u$, i.e. let $e\in E$ and let $j$ such that $d(\partial^-(e),o) =r_j$ then we define
\[v_e (t):= \frac{1}{S_{j}} \sum_{{\tilde{e}}\mid d(o,\partial^-(e))= r_j} u_{\tilde{e}}(x).\]
We easily conclude by linearity that $v_e$ fulfills $v_e'' = \alpha v_e$ for all $e\in E$. We are left to show continuity and the $\triangle$-Kirchhoff laws. Let $x\in S_{r_j}$ and assume that neither $\deg_+(r_j)$ nor $\deg_-(r_j)$ are zero, since these cases are trivial. Then for $e,f\in E$ with $\partial^-(e)=x$ and $\partial^+(f)=x$ we have by continuity of $u$
\[u_e(l_j-) = u_f(0+)\]
and hence
\[\frac{1}{\deg_-(r_j)}\sum_{e\mid \partial^-(e)=x} u_e(l_j-) = \frac{1}{\deg_+(r_j)}\sum_{f\mid \partial^+(f)=x} u_f(0+) \]
and thus by summing over all $x\in S_{r_j}$ we obtain
\[\frac{1}{\deg_-(r_j)} \sum_{e\mid \partial^-(e)\in S_{r_j}} u_e(l_j-) = \frac{1}{\deg_+(r_j)} \sum_{f\mid \partial^+(f)\in S_{r_j}} u_f(0+).\]
Multiplying this equation with $\frac{1}{S_j}$ and recall that $S_j\deg_+(r_j)= S_{j+1} \deg_+(r_{j+1})$ we get that
\[v_e(l_e-)=v_f(0+)\]
for all $e\sim f$. Moreover we get that for $x\in S_{r_j}$ the equality \[v(x)= \frac{1}{S_{r_j}} \sum_{y\in S_{r_j}} u(y)\]
holds. The calculations for the $\triangle$-Kirchhoff laws are analogous. Using linearity of the divergence and the discrete Laplacian we obtain for $x\in S_{r_j}$
\[\partial_n v(x) = \frac{1}{S_{r_j}} \sum_{y\in S_{r_j}} \partial_n u(y) \mbox{ and } \triangle v(x) = \frac{1}{S_{r_j}} \sum_{y\in S_{r_j}} \triangle u(y)\]
and hence $v$ also satisfies the $\triangle$-Kirchhoff conditions.
\end{proof}
\section{The Kirchhoff Dirichlet form}
In this section we have a closer look at the discrete operator $\triangle_\alpha$. First we recall some notation from the previous sections which are of use below. The Kirchhoff Laplacian was defined by means of the functions $h_\alpha$ for $\alpha \geq 0$. For $\alpha >0$ it is explicitly given by
\[\triangle_\alpha u(x)= \sum_{e=(y,x)\sim x} \frac{\omega(x,y) \sqrt{\alpha}}{\sinh(\sqrt{\alpha} l_i(x,y))}(u(x)-u(y))\]
and for $\alpha =0$ by
\[\triangle_\alpha u(x)= \sum_{e=(y,x)\sim x} \frac{\omega(x,y)}{ l_i(x,y)}(u(x)-u(y)).\]
The discrete functions $M_\alpha:\FV \to (0,\infty)$ are given by
\[M_\alpha (x)= \sum_{e=(x,y)\sim x} \omega(x,y)\sqrt{\alpha} \frac{\cosh(\sqrt{\alpha} l_i(x,y)) - 1}{\sinh(\sqrt{\alpha} l_i(x,y))}\]
for $\alpha>0$ and $M=M_0$ is given as
\[M(x)= \sum_{e=(x,y) \sim y} l_i(x,y) \omega (x,y) = \lambda_\omega(\bigstar(x)).\]
The formal weighted Laplacian which is connected with the jump and the killing weights is given as
\[\triangle u(x) =\sum_{y\in \FV} j(x,y)(u(x)-u(y)) + c(x) u(x).\]
We are now armed for this section and start with a proposition on discrete Dirichlet forms, taken from \cite{HKLW-12}.
\begin{prop}
For all  $u\in \CC(\FV)$ and $v\in \CC_c(\FV)$ we have
\[\langle \triangle_\alpha  u,v\rangle = \langle u, \triangle_\alpha v\rangle,\]
where all sums converge absolutely and agree with
\[\frac{1}{2} \sum_{x,y\in \FV}  - {h_e^\alpha}'(l_i(x,y))\omega(e)(u(x)-u(y))(v(x)-v(y))\]
\end{prop}
It will be more convenient to consider the form introduced above in the space $\ell^2(\FV,M_\alpha)$. Note that we set $M_0=M$.
\begin{theorem}
For all $\alpha \geq 0$ the form defined by
\[\CQ_\alpha(u) = \frac{1}{2} \sum_{x,y\in \FV}  - {h_e^\alpha}'(l_i(x,y))\omega(e)(u(x)-u(y))^2,\]
is a discrete Dirichlet form on the space $\ell^2(\FV,M_\alpha)$ with domain
\[\CD(\CQ_\alpha)=\{u \in \ell^2(\FV,M_\alpha) \mid \CQ_\alpha(u) < \infty \}.\]
\end{theorem}
The form $\CQ_\alpha$ will be called the Kirchhoff Dirichlet form. \index{Kirchhoff Dirichlet form} The theorem is purely on discrete Dirichlet forms and can be found in \cite{HKLW-12}.
 The next lemma shows that for varying $\alpha$ the spaces are isomorphic.
\begin{lemma}
For all $\alpha>0$ the Dirichlet spaces $\CQ_0+\CQ$ on $\ell^2(\FV,M)$ and $\CQ_\alpha +\CQ$ on $\ell^2(\FV,M_\alpha)$ are isomorphic.
\end{lemma}
\begin{proof}
This follows as we have by proposition 3.14 that for all $\alpha>0$ there exists $C>0$ such that
\[C^{-1} M \leq M_\alpha \leq C M\]
and
\[C^{-1} \CQ_0(u) \leq \CQ_\alpha (u) \leq C \CQ_0(u).\]
\end{proof}
Before we derive the connection of the graph Dirichlet form $\CS$ and the Kirchhoff Dirichlet form we continue with the following decomposition.
\begin{theorem}
Let $\FX_\Gamma$ be a metric graph and $\CS$ a graph Dirichlet form. Then for all $\alpha>0$ we have the orthogonal decomposition
\[\CD(\CS)= \Bigl(\bigoplus_{e\in E} W^{1,2}_o((0,l_i(e)),\omega)\Bigr) \oplus \CH^\FV_\alpha\]
with respect to the inner product $\CS(\cdot,\cdot) + \alpha(\cdot,\cdot)$, where $\CH^\FV_\alpha$ denotes the space of all functions which are $\alpha$-harmonic on the edges.
\end{theorem}
\begin{proof}
We denote by $\CD(\CS)_\FV$ the space of all functions of $\CD(\CS)$ which are zero on the set of vertices. By the Sobolev inequality this space is given as the closure of $\CD(\CE)\cap (\bigoplus\limits_{e\in E} \CC_c^\infty(0,l_i(e)))$ with respect to $\|\cdot\|_\CE$, as the discrete part vanishes on this space. Thus we have
\[\CD(\CS)_\FV = \bigoplus_{e\in E} W^{1,2}_o((0,l_i(e),\omega)).\]
Considering its orthogonal complement, we have $u\in\CD(\CS)^\bot$ if and only if
\[ \CE(u,\phi) + \alpha (u,\phi)=0\]
for all $\phi \in \bigoplus\limits_{e\in E} \CC_c^\infty (0,l_i(e))$ which gives the claim on $\CH_\alpha^\FV$.
\end{proof}
Basically the space $\bigoplus\limits_{e\in E} W^{1,2}_o((0,l_i(e),\omega))$ corresponds to a decoupled graph, i.e. the disjoint union of all edges of the graph with Dirichlet boundary conditions. In particular the discrete part $\CQ$ vanishes on this space. Thus all information of the original graph should be contained in $\CH_\alpha^\FV$. This is the content of the next theorem.
\begin{theorem}
Let $\FX_\Gamma$ be a metric graph and $\CS$ a graph Dirichlet form.
Then the restriction mapping
\[\cdot|_\FV: \CH_\alpha^\FV \to \CD(\CQ_\alpha+\CQ)\]
is an isometric isomorphism with inverse $H_\alpha$. In particular we have
\[\alpha\|u\|_{L^2}^2+ \CE(u)= \|u|_\FV\|_{\ell^2}^2 + \CQ_\alpha(u).\]
Furthermore, the space $\CH_\alpha$ is mapped bijectively to the space of $1$-harmonic functions of $\CQ+\CQ_\alpha$, i.e. functions $U$ satisfying
\[\triangle_\alpha U(x) + \triangle U(x) + M_\alpha (x) U(x) =0.\]
\end{theorem}
\begin{proof}
By the representation formulae of $\alpha$-harmonic functions we obtain for each component $u_e$ on the edge $e=(x,y)$
\[u_e(t) = u(x) \tfrac{\sinh\sqrt{\alpha}(l_i(e)-t)}{\sinh\sqrt{\alpha} l_i(e)} + u(y) \tfrac{\sinh\sqrt{\alpha}t}{\sinh \sqrt{\alpha}l_i(e)}\]
and thus
\[u_e'(t)=- u(x)\sqrt{\alpha} \tfrac{\cosh\sqrt{\alpha}(l_i(e)-t)}{\sinh \sqrt{\alpha}l_i(e)} + u(y)\sqrt{\alpha} \tfrac{\cosh\sqrt{\alpha}t}{\sinh \sqrt{\alpha}l_i(e)}.\]
A short calculation gives
\[\alpha \|u_e\|_{L^2}^2 + \|u_e'\|_{L^2}^2 = u(x)^2 A_1 + 2 u(x)u(y) A_2 + u(y)^2 A_3\]
where
\[A_1 = \alpha \int\limits_0^{l_i(e)} \tfrac{\sinh^2\sqrt{\alpha}(l_i(e)-t)+ \cosh^2\sqrt{\alpha}(l_i(e)-t)}{\sinh^2 \sqrt{\alpha}l_i(e)} dt\]
\[A_2 = \alpha \int\limits_0^{l_i(e)} \tfrac{\sinh\sqrt{\alpha}(l_i(e)-t)\sinh\sqrt{\alpha}t - \cosh\sqrt{\alpha}(l_i(e)-t)\cosh\sqrt{\alpha}t}{\sinh^2 \sqrt{\alpha}l_i(e)} dt\]
\[A_3 = \alpha \int\limits_0^{l_i(e)} \tfrac{\sinh^2\sqrt{\alpha}t+ \cosh^2\sqrt{\alpha}t}{\sinh^2\sqrt{\alpha} l_i(e)} dt.\]
Evaluating the integrals we obtain
\[A_1=A_3 = \sqrt{\alpha}\tfrac{\cosh \sqrt{\alpha}l_i(e)}{\sinh\sqrt{\alpha}l_i(e)}\]
\[A_2 = -\sqrt{\alpha}\tfrac{1}{\sinh l_i(e)}.\]
Reorganizing this sum gives
\[\alpha\|u_e\|_2^2 + \|u_e'\|_2^2 = (u(x)^2 + u(y)^2) \sqrt{\alpha}\tfrac{\cosh \sqrt{\alpha}l_i(e) - 1}{\sinh \sqrt{\alpha}l_i(e)} + \tfrac{\sqrt{\alpha}}{\sinh \sqrt{\alpha}l_i(e)} (u(x) - u(y))^2 ,\]
and hence by summing over all $e\in E$ we get
\begin{eqnarray*}
\alpha \|u\|_{L^2}^2 + \CE(u) &=& \sum_{e\in E} (\alpha\|u_e\|_{L^2}^2 + \|u_e'\|_{L^2}^2)\omega(e)\\
&=&\tfrac{1}{2} \sum_{x,y\in \FV }(u(x)^2 + u(y)^2)\omega(x,y) \sqrt{\alpha}\tfrac{\cosh \sqrt{\alpha}l_i(x,y) - 1}{\sinh \sqrt{\alpha}l_i(x,y)}\\
&&+ \sum_{x,y\in \FV}\omega(x,y)\tfrac{\sqrt{\alpha}}{\sinh \sqrt{\alpha}l_i(x,y)} (u(x) - u(y))^2\\
&=& \tfrac{1}{2} (\sum_{x \in \FV }u(x)^2 M_\alpha(x) + \sum_{y \in \FV }u(y)^2 M_\alpha(y))\\
&&+ \sum_{x,y\in \FV}\omega(x,y)\tfrac{\sqrt{\alpha}}{\sinh \sqrt{\alpha}l_i(x,y)} (u(x) - u(y))^2\\
&=& \|u\|_{\ell^2}^2 + \CQ_\alpha(u).
\end{eqnarray*}
This gives the in particular part. All remaining claims thus follow.
\end{proof}
Note that the definition of harmonicity of discrete functions is analogous as in our case. For a definition and further reading see the papers \cite{HK-10} and \cite{HKLW-12}. In the rest of this section we investigate certain connections between the continuous and the discrete world. We begin with regularity.
\begin{theorem}
The graph Dirichlet form $\CS$ is regular if and only if the discrete Dirichlet form $\CQ_0+\CQ$ is. In particular, a diffusion Dirichlet form is regular if and only if the associated Kirchhoff Dirichlet form is.
\end{theorem}
\begin{proof}
We use the characterization of regularity by harmonic functions in section 2.6. Thus $\CS$ is not regular if and only if the space $\CH_1$ is nontrivial. This space is non-trivial if for some $\alpha>0$ the discrete Dirichlet form $\CQ+\CQ_\alpha$ has a nontrivial $1$-harmonic function. As $\CQ+\CQ_\alpha$ and $\CQ + \CQ_0$ are isomorphic, such a function exists if and only if the Dirichlet form $\CQ+\CQ_0$ admits a nontrivial $1$-harmonic function. By a similar decomposition as in section 2.6, which was given in \cite{HKLW-12}, we have
\[\CD(\CQ + \CQ_0) = \CD_o(\CQ+\CQ_0) \oplus \Ch_1\]
where $\Ch_1$ denotes the space of $1$-harmonic functions of $\CQ+\CQ_0$. Thus we see that regularity of $\CQ+\CQ_0$ corresponds to a trivial space $h_1$, which yields the claim.
\end{proof}
The appearance of harmonic functions with certain integrability, resp. summability, is also in the focus of analysis. For this reason we consider further mapping properties of $H_\alpha$. We already know that as operator on $\CD(\CQ_0+\CQ)$ it has a bounded inverse, given as the restriction mapping. As we know that $H_\alpha$ maps $\ell^p$ into $L^p$ we may ask whether it is also invertible on its image with bounded inverse. This is the content of the next theorem.
\begin{theorem}\index{harmonic extension!invertibility}
Let $\FX_\Gamma$ be a metric graph. Then the following holds true.
\begin{itemize}
\item[(i)] For $1\leq p \leq 2$ the operator $H_\alpha: \ell^p(\FV,M) \to L^p(\FX,\omega)$ has a bounded inverse on its image.
\item[(ii)] For $2< p < \infty$ the operator $H_\alpha: \ell^p(\FV,M^p) \to L^p(\FX,\omega)$ has a bounded inverse on its image, where $M^p(x):= \sum\limits_{e\sim x} l_i(e)^{p-2}\omega(e)$.
\item[(iii)] For $p=\infty$ the operator $H_\alpha: \ell^p(\FV,M) \to L^p(\FX,\omega)$ has a bounded inverse on its image.
\end{itemize}
\end{theorem}
\begin{proof}
We start with a Sobolev inequality. For a function $u\in W^{2,2}(0,L)$ we have the estimate
\[ L^{2p} |f'(x)|^p \leq c \Bigl(L \int\limits_0^L |f(t)|^p dt + L^{2p+1} \int\limits_0^L |f''(t)|^p dt\Bigr)\]
and thus
\begin{eqnarray*}
|f(x)|^p &\leq & C\Bigl( L^{-1} \int\limits_0^L |f(t)|^p dt + L^{p-1} \int\limits_0^L |f'(t)|^p dt\Bigr)\\
&\leq& C\Bigl( L^{-1} \int\limits_0^L |f(t)|^p dt + L^{-p-1} (L^2 \int\limits_0^L |f(t)|^p dt + L^{2p+2} \int\limits_0^L |f''(t)|^p dt)\Bigr).
\end{eqnarray*}
If now $f$ additionally satisfies $f'' = \alpha f$ we obtain
\[|f(x)|^p \leq C (L^{-1} + L^{-p+1} + L^{p+1}) \int\limits_0^L |f(t)|^p dt.\]
Turning back to the graph case, let $u \in H_\alpha(\ell^p(\FV,M))$ and $1\leq p \leq 2$. Then we have from the result above
\begin{eqnarray*}
\sum_{x\in \FV} |u(x)|^p M(x) &\leq& C \sum_{x\in \FV} \sum_{e\sim x} (1+l_i(e)^{-p+2} + l_i(e)^{p+2})\omega(e) \int\limits_0^{l_i(e)} |u_e(t)|^p dt \\
&\leq& C \|u\|_p^p
\end{eqnarray*}
as we have assumed $p\leq 2$ and $l_i(e)\leq 1$. For $p\geq 2$ we analogously have
\begin{eqnarray*}
\lefteqn{\sum_{x\in \FV} |u(x)|^p (\sum_{e\sim x} l_i(e)^{p-2}\omega(e))}\\
&\leq& C \sum_{x\in \FV} \sum_{e\sim x} (l_i(e)^{p-2} + 1 + l_i(e)^{2p})\omega(e) \int\limits_0^{l_i(e)} |u_e(t)|^p dt\\
&\leq& C \|u\|_p^p.
\end{eqnarray*}
Part (iii) is obvious.
\end{proof}
The theorem above could be read as follows. Whenever $\CS$ admits an $\alpha$-harmonic function in $L^p$, then the discrete form $\CQ_\alpha + \CQ$ admits a harmonic function in $\ell^p$ as well. This follows since $H_\alpha$ maps harmonic functions to harmonic functions. Another property of interest is stochastic completeness or conservation of probability. Though we will not give its definition here, we remark that this is equivalent to all nonnegative and bounded $\alpha$-harmonic functions being zero. We this in mind we can prove the following theorem, which treats the case of a diffusion Dirichlet form.
\begin{theorem}\index{theorem!$L^p$-Liouville type}
A diffusion Dirichlet form $\CE$ is stochastically complete if and only if the associated Kirchhoff Dirichlet form $\CQ_\alpha$ is.
\end{theorem}
\begin{proof}
In \cite{St-94}, the equivalence of stochastic completeness with the absence of nontrivial, nonnegative and bounded $1$-harmonic functions was shown. Secondly, Proposition 3.20 says that $u$ is a nontrivial, nonnegative and bounded $1$-harmonic function of $\CS$ if and only if its restriction to the vertex set is a nontrivial, nonnegative and bounded $1$-harmonic of the Kirchhoff Dirichlet form, i.e. we have
\[ M_\alpha (x)^{-1}(\triangle_\alpha + \triangle) u(x)+ u(x) =0.\]
Finally, in \cite{KL-11} it was shown in the setting of discrete Dirichlet forms, that the absence of nontrivial, nonnegative and bounded $1$-harmonic functions is equivalent to stochastic completeness. Putting these three pieces together yields the claim.
\end{proof}
In \cite{Folz,Huang} the authors show a volume growth criterion, which is due to Grigor'yan \cite{Gr-99} for manifolds, for discrete graphs. This was an open problem that received quite some attention, see \cite{Wo-10,GHM,MUW-12}. Their basic strategy is to construct an appropriate metric graph. There they can apply the Grigor'yan criterion in the setting of strongly local Dirichlet spaces, see \cite{St-94}. Finally they translate this result back, which yields the desired result. The theorem above gives us the same possibility. The usual problem is to relate the volume growths in both settings with each other. As the measure of the discrete problem is nothing but the measure of star neighborhoods, we see that we can be optimistic to deduce the same results as in \cite{Folz,Huang}. This will be further discussed in \cite{HW}.\\
We continue with an improvement on the properties of the restriction map compared to the results in chapter 2.
\begin{coro}
The restriction map
\[\cdot|_\FV:\CD(\CE) \to \ell^2(\FV,M)\]
is continuous.
\end{coro}
\begin{proof}
For $u\in \CD(\CE)$ we decompose the function as $u=u_o+u_h$. Then $u|_\FV = u_h |_\FV$ and thus
\begin{eqnarray*}
\|u|_\FV\|_{\ell^2}^2 &\leq& \|u_h\|_{\ell^2}^2 + \CQ_0(u_h)\\
&\leq& C (\|u_h\|_{L^2}^2 + \CE(u_h))\\
&\leq & C \|u\|_\CE
\end{eqnarray*}
\end{proof}
The previous corollary improves the former result with the vertex measure $m(x)=\min\limits_{e\sim x} \omega(e)l(e)$. Thus we can also improve the result on the relative boundedness.
\begin{coro}
Assume the discrete part $\CQ$ of the graph Dirichlet form is a bounded form on $\ell^2(\FV,M)$, then $\CQ$ relatively $\CE$-bounded.
\end{coro}
The next theorem gives a Sobolev type inequality. In chapter $2$ we have proven it for a diffusion under the assumption that the graph is connected. We now prove it under the assumption of irreducibility for a general graph Dirichlet form. Before we state and prove it we start with two lemmata concerning the part on the vertices.
\begin{lemma}
Let $\FX_\Gamma$ be a metric graph. Then for all $u\in  \bigoplus\limits_{e\in E} W^{1,2}_0((0,l_i(e)),\omega)$ the inequality
\[\|u\|_2^2 \leq  \CE(u)\]
holds. In particular, the Sobolev inequality holds if and only if we have the strong Sobolev inequality
\[\|u\|^2_\infty \leq C \CE(u).\]
\end{lemma}
\begin{proof}
For all $e\in E$ we have
\[\|u_e\|_\infty^2 \leq l_i(e) \int\limits_0^{l_i(e)} |u'_e(t)|^2 \: dt\]
and therefore as $l_i(e)\leq 1$ we have
\[\|u\|_2^2 \leq \sum_{e\in E} l_i(e)^2 \|u_e'\|_2^2 \leq \CE(u).\]
\end{proof}
The next lemma is a certain version of an isoperimetric inequality.
\begin{lemma}
Let $\FX_\Gamma$ be a metric graph. Then the strong Sobolev inequality holds for $ \bigoplus\limits_{e\in E} W^{1,2}_0((0,l_i(e)),\omega)$ if and only if we have
\[l_i(e)\leq C \omega(e).\]
\end{lemma}
\begin{proof}
Assume the strong Sobolev inequality holds. Let $e'\in E$ and choose $u\in  \bigoplus\limits_{e\in E} W^{1,2}((0,l_i(e)),\omega)$ such that $u_{e'}(\tfrac{l_i(e)}{2}) =1$ and linearly interpolated on the rest of the interval. Then by the Sobolev inequality we have
\[1 \leq  C \CE(u) = \tfrac{\omega(e)}{l_i(e)}\]
giving that $l_i(e) \leq C \omega(e)$. To prove the converse we have for a function $u\in  \bigoplus\limits_{e\in E} W^{1,2}_0((0,l_i(e)),\omega)$ and $e\in E$ that
\[\|u_e\|_\infty^2 \leq l_i(e) \int\limits_0^{l_i(e)} |u'_e(t)|^2 dt \leq C \omega(e) \int\limits_0^{l_i(e)} |u'_e(t)|^2 dt\]
which yields claim by summing over $e\in E$.
\end{proof}
\begin{theorem}\index{theorem!Sobolev embedding!global version}
Let $\FX_\Gamma$ be a metric graph and $\CS$ a graph Dirichlet form. Then the Sobolev inequality holds for $u \in \CD(\CS)$ if and only if \[l_i(e) \leq C \omega (e)\]
for all $e\in E$ and we have the discrete Sobolev inequality
\[\|U\|_\infty^2 \leq C ( \|U\|_{\ell^2}^2 + \CQ_0(U) + \CQ(U))\]
for all $U\in \CD(\CQ+\CQ_0)$.
\end{theorem}
\begin{proof}
Decompose $u\in \CD(\CS)$ as $u_o+ u_h$. Then we know from the previous lemma that the validity of the Sobolev inequality is equivalent to $l_i(e) \leq \omega(e)$. For $u_h$ we have
\begin{eqnarray*}
\|u_h\|_\infty ^2 &=& \sup_{x\in \FV} |u(x)|^2\\
&\leq& C( \|u_h\|_{\ell^2}^2 + \CQ(u_h) + \CQ_0(u_h))\\
&\leq& C ( \|u_h\|_{L^2}^2 + \CS(u_h))
\end{eqnarray*}
where the first equality follows from the maximum principle. If now the discrete inequality is valid and $l_i(e)\leq C \omega (e)$ holds true we have
\begin{eqnarray*}
\|u\|_\infty^2 &\leq & 2 (\|u_o\|^2_\infty + \|u_h\|^2_\infty)\\
&\leq& C (\CE(u_o) + \|u_h\|_{L^2}^2 + \CS(u_h))\\
&\leq& C  (\|u\|_{L^2}^2 + \CS(u)).
\end{eqnarray*}
The converse direction is proven easily by applying the Sobolev inequality separately to $\bigoplus\limits_{e\in E} W^{1,2}((0,l_i(e)),\omega)$ and $\CH_1^\FV$.
\end{proof}

\section{Part Dirichlet forms and approximation}
In this section we are interested in assigning to each open subset $\FY$ an appropriate operator which is connected with the weak Laplacian. This is done via the tool of part processes and goes as follows. Let $\FY$ be a subset of $\FX_\Gamma$ and consider the set
\[\CD(\CS)_{\FY,o} = \{u\in \CD(\CS) \mid u|_{\FX_\Gamma\setminus \FY} \equiv 0\}.\]\index{Dirichlet form!part}
We consider this space as a subspace of $L^2(\FY,\omega)$. Applying general results of part Dirichlet forms to our setting, cf. \cite{CF-12} chapter 4, we derive that if $\FY$ is open and relatively compact then $\CS_{\FY,o}$ is a regular Dirichlet form on $L^2(\FY,\omega)$.\\
Closely connected to part Dirichlet forms is the following version of Dirichlet's problem in a relatively compact open subset $\FY\subset \FX_\Gamma$: For $f\in L^2(\FX_\Gamma,\omega)$, $w\in \CD(\CS)_\loc^*$ and $\alpha >0$, does there exist a function $u$ such that
\[u=w \mod \CD(\CS)_{\FY,o},\]
and
\[\CS(u,\phi) + \alpha (u,\phi) = (f,\phi)\]
for all $\phi \in \CD(\CS)_{\FY,o}$ with compact support, where the first condition means that we have $u=w+w_o$ for some $w_o\in \CD(\CS)_{\FY,o}$. As $\FY$ is a subgraph we can consider the weak Laplacian on $\FY$. We immediately obtain that the Dirichlet problem means that $u$ is a weak $\alpha$-solution of the equation $\widetilde{\Delta} u + \alpha u = f$ with the additional boundary value $u=w$ outside of $\FY$.\medskip\\
Substituting $v= u-w$ we get that $v\in \CD(\CS)_{\FY,o}$ and the equation above becomes
\[\CS(v,\phi) + \alpha (v,\phi) = - \CS(w,\phi) + (f-\alpha w, \phi).\]
As by Cauchy-Schwarz
\[|\CS(w,\phi)| \leq \CS(w) \|\phi\|_\CS\]
and
\[|(f-\alpha w, \phi)| \leq \|f-\alpha w\|_2 \|\phi\|_\CS\]
we get that the linear functional
\[f_w(\phi) := -\CS(w,\phi) + (f-\alpha w, \phi)\]
is bounded on the Hilbert space $\CD(\CS)_{\FY,o}$ equipped with the energy norm. Thus by the Riesz representation theorem there exists a $v\in \CD(\CS)_{\FY,o}$ such that
\[\CS(v,\phi) + \alpha (v,\phi) = f_w (\phi)\]
giving the existence and the uniqueness of a solution of the Dirichlet problem by $u = v+w$. Note that if $w\equiv 0$ this solution agrees with the associated resolvent of $\CS_{\FY,o}$, i.e. \[u= (\Delta_{\FY,o} + \alpha)^{-1} f\]
where $\Delta_{\FY,o}$ denotes the $L^2$ generator of the Dirichlet form $\CS_{\FY,o}$ which is uniquely defined by the identity
\[\CS(u,\phi) = (\Delta_{\FY,o} u , \phi) \quad \forall \phi\in \CD(\CS)_{\FY,o}.\]
In particular for $\FY = \FX_\Gamma$ we obtain the generator $\Delta_o$ of the regular Dirichlet form $\CS_o$. Our aim is to show that we can approximate the resolvent of $\CS_o$ by  the resolvents of the part Dirichlet forms $\CS_{\FY,o}$. Note that this is true for general Dirichlet forms and can, for example, be found in the works \cite{Stolli-95,StolliV-96}. We include a proof following \cite{KL-11}. The key ingredient for that is the domain monotonicity principle which will basically follow from the minimum principle.
\begin{lemma}\index{weak Laplacian!domain monotonicity}
Let $\FY_1,\FY_2 \subset \FX_\Gamma$ be precompact open subsets with $\FY_1 \subset \FY_2$. Then for any $x\in \FY_1$
\[(\Delta_{\FY_1 ,o} + \alpha)^{-1} f(x) \leq (\Delta_{\FY_2,o} + \alpha)^{-1} f(x)\]
for all $f\in L^2(\FX_\Gamma,\omega)$ with $f\geq 0$ and $\supp f \subset \FY_1$.
\end{lemma}
\begin{proof}
We consider the function
\[u:= (\Delta_{\FY_2,o} + \alpha)^{-1} f - (\Delta_{\FY_1,o} + \alpha)^{-1} f\]
where we extend the resolvents on the complements of $\FY_1$ and $\FY_2$ by zero. Then by definition $u$ is $\alpha$-superharmonic on $\FY_1^\complement$ and $u$ attains its minimum on $\FY_1$ as it is precompact. Furthermore $u$ is $\alpha$-harmonic on $\FY_1$. Thus by the minimum principle we obtain $w\geq 0$.
\end{proof}
The domain monotonicity has the desired approximation result as consequence. We state it without proof, the interested reader may find one in \cite{KL-11}.
\begin{lemma}
Let $\FX_\Gamma$ be a metric graph and $\CS$ a graph Dirichlet form. Let $(\FY_n)_n$ be an exhausting sequence of $\FX_\Gamma$, i.e. each $\FY_n$ is open, relatively compact, $\overline{\FY}_n \subset \FY_{n+1}$ and $\bigcup\limits_{n\in \IN} \FY_n = \FX_\Gamma$. Then,
\[(\Delta_{\FY_n ,o} + \alpha)^{-1} f \to  (\Delta_{o} + \alpha)^{-1} f \]
in $L^2$-norm for all $f\in L^2(\FX_\Gamma, \omega)$.
\end{lemma}
Next we show the monotone convergence of solutions.
\begin{lemma}\index{weak Laplacian!monotone convergence of solutions}
Let $(u_n)_n$ be a sequence of nonnegative functions in $\CD(\CS)_\loc^*$, $u\in L_\loc^1(\FX_\Gamma)$ and $f\in L^1_\loc(\FX_\Gamma)$ such that $u_n \leq u_{n+1}$, $u_n\to u$ pointwise and $\widetilde{\Delta} u_n \to f$ in the sense
\[\lim_{n\to \infty} \CS(u_n,\phi) = (f,\phi)\]
for all $\phi\in \CD(\CS)_c$. Then $u\in \CD(\CS)_\loc^*$ and $\widetilde{\Delta} u = f$.
\end{lemma}
\begin{proof}
Applying the definition of the weak Laplacian we need to show for the sequence $u_n$
\[\CS(u,\phi) = (f,\phi) = \lim_{n\to \infty} \CS(u_n,\phi).\]
As $\phi$ is compactly supported, we have
\begin{eqnarray*}
\lefteqn{\CS(u_n,\phi)}\\
 &=& \sum_{e\in E}\omega(e)\Bigl( \int\limits_0^{l_i(e)} -(u_n)_e(t) \phi''(t) \:dt + u_n(\partial^-(e))\phi'(\partial^-(e)) -  u_n(\partial^+(e))\phi'(\partial^+(e))\Bigr) \\
&& + \sum_{x\in \FV} \phi(x)\Bigl(u_n(x)k(x) + \sum_{y\in \FV} j(x,y)(u_n(x)-u_n(y))\Bigr).
\end{eqnarray*}
As
\[\Bigl|\int\limits_0^{l_i(e)} (u_n)_e(t) \phi''(t) \:dt\Bigr| \leq \|\phi''\|_\infty \int\limits_0^{l_i(e)} |u(t)| dt\]
we get by monotone convergence that
\[\int\limits_0^{l_i(e)} (u_n)_e(t) \phi''(t) \:dt \to \int\limits_0^{l_i(e)} u_e(t) \phi''(t) \:dt.\]
Thus we also get that for all $e\in E$ we have $u_e\in W^{1,2}(0,l_i(e))$. In the remaining summands the only part where convergence is not obvious is
\[\sum_{y\in \FV}j(x,y)u_n(y) \to \sum_{y\in \FV}j(x,y)u(y).\]
That this converges follows since the whole term and all other summands converge. Monotone convergence again shows the convergence. Thus
\[\lim_{n\to \infty} \CS(u_n,\phi) = \CS(u,\phi)\]
and the claim follows.
\end{proof}
These three principles will be used in the next chapter to obtain certain information on the generators of the resolvents of the Dirichlet form $\CS$ acting on different $L^p$ spaces.
\section{The $L^p$ generators}
In this section our aim is to derive a characterization of the $L^p$-generators associated with $\CS$. We start with the $L^2$-generators of the form $\CS$ and the regular form $\CS_o$. The representation theorem due to Friedrichs, see for instance \cite{K-95}, says that there exists an operator $\Delta_\CS$ such that
\[\CS(u,v)= (u,\Delta_\CS v)\]
for all $u\in \CD(\CS)$ and $v\in \CD(\Delta_\CS)$. Similarly there exists an operator $\Delta_o$ such that
\[\CS_o(u,v)= (u,\Delta_o v)\]
for all $u\in \CD(\CS)_c$ and $v\in \CD(\Delta_o)$.
We start with the characterization of $\Delta_\CS$.
\begin{theorem}\index{graph Dirichlet form!$L^2$-generator}
Let $\FX_\Gamma$ be a metric graph and $\CS$ a graph Dirichlet form. Then the $L^2$ generator of $\CS$ is given by the restriction of the weak Laplacian to the set of all $u\in \bigoplus\limits_{e\in E} W^{2,2}((0,l_i(e)), \omega(e))$ satisfying the flux conditions
\[\frac{\partial u}{\partial \phi}  + \CQ(u,\phi)=0\]
for all $\phi \in \CH_1$, where
\[\frac{\partial u}{\partial \phi} = \sum_{e\in E} \bigl(u_e'(l_i(e))\phi_e(l_i(e)) - u_e'(0)\phi_e(0)\bigr).\]
\end{theorem}
\begin{proof}
We use the decomposition $\phi = \phi_o + \phi_h$ for $\phi \in \CD(\CS)$. Then for all $u\in \CD(\CS)$ we have by definition of the weak Laplacian
\[\CS(u,\phi) - (\widetilde{\Delta} u,\phi)=\CS(u,\phi_o) - (\widetilde{\Delta} u,\phi_o) +\CS(u,\phi_h) - (\widetilde{\Delta} u,\phi_h) = \CS(u,\phi_h) - (\widetilde{\Delta} u,\phi_h)\]
whenever it makes sense. Furthermore by edgewise partial integration we have
\[\CS(u,\phi_h) = \sum_{e\in E} u_e'(l_i(e))\phi_e(l_i(e)) - u_e'(0)\phi_e(0) + (\widetilde{\Delta} u,\phi) + \CQ(u,\phi_h).\]
As $\phi = \phi_h$ on $\FV$ it holds that
\[\CS(u,\phi_h) - (\widetilde{\Delta} u,\phi_h) = \frac{\partial u}{\partial \phi}  + \CQ(u,\phi).\]
If $\widetilde{\Delta} u \in L^2$ and the flux conditions are satisfied, then the above is well-defined and we obtain
\[\CS(u,\phi) - (\widetilde{\Delta} u,\phi) =0= \CS(u,\phi) - (\Delta_\CS u,\phi).\]
Conversely assume that $u\in \dom \Delta_\CS$ then by Friedrich's representation theorem we have
\[\CS(u,\phi) = (\Delta_\CS u,\phi)\]
for all $\phi \in \CD(\CS)$. In particular this holds for all $\phi \in \CD(\CS)_c$ and we have
\[\CS(u,\phi) = (\Delta_\CS u,\phi) = (\widetilde{\Delta} u, \phi),\]
i.e.
\[\widetilde{\Delta} u = \Delta_\CS u.\]
As $\Delta_\CS u \in L^2$ by definition the calculations above are again well-defined and we obtain
\[0= \CS(u,\phi) - (\widetilde{\Delta} u, \phi) = \frac{\partial u}{\partial \phi}  + \CQ(u,\phi),\]
which gives that the flux conditions are satisfied.
\end{proof}
Using edgewise partial integration, we get a characterization of the generator of $\CS_o$.
\begin{theorem}
Let $\FX_\Gamma$ be a metric graph. Then the $L^2$ generator of $\CS_o$ is given by the restriction of the weak Laplacian to
\[\dom \Delta_{\CS_o} = \CD(\CS_o) \cap \bigoplus\limits_{e\in E} W^{2,2}((0,l_i(e)),\omega(e)).\]
\end{theorem}
We now turn our attention to the resolvents on $L^p$. As we have mentioned in the previous section, the $L^2$-resolvent of $\CS$ can be characterized as the unique element $(\Delta_\CS + \alpha \mathrm{id})^{-1}f$ in $\CD(\CS)$ such that
\[\CS((\Delta_\CS + \alpha \mathrm{id})^{-1}f,v) + \alpha((\Delta_\CS + \alpha \mathrm{id})^{-1}f,v)= (f,v)\]
holds for all $v\in \CD(\CS)$. As $\CS$ is a Dirichlet form, the resolvent is a contraction on $L^2\cap L^p$, and thus the restriction of $(\Delta_\CS + \alpha \mathrm{id})^{-1}$ to $L^2\cap L^p$ could be extended to a contraction on $L^p$. Unfortunately, in the case that $\CS$ is not regular, the determination of the $L^p$ generators is not tractable. Thus we just treat the case that $\CS=\CS_o$. In the following we denote the associated $L^p$ generator by $\Delta_p$ for short. We first show that $\Delta_p$ is a restriction of the weak Laplacian $\widetilde{\Delta}$.
\begin{prop}
Let $p\in [1,\infty]$, then for any $g\in L^p(\FX_\Gamma,\omega)$ the function $u:=(\Delta_p + \alpha)^{-1} g$ belongs to $\dom \widetilde{\Delta}$ and the equation
\[\widetilde{\Delta} u + \alpha u = g\]
holds true for $\alpha >0$.
\end{prop}
\begin{proof}
From the previous theorem we know that in the case $p=2$ the operator $\Delta_2=\Delta_\CS$ is a restriction of $\tilde{\Delta}$.\\
Assume now the general case of $1\leq p < \infty$ and let $\FX_\Gamma = \bigcup_n K_n $ be an exhaustion of $\FX_\Gamma$. Let $g\in L^p$, then $u_n := (\Delta_p + \alpha)^{-1} (g|_{K_n})$ converges to $u$ in $L^p$-sense. Moreover, $u_n\in L^2$, since the restrictions are in $L^2$ and also $u_n = (\Delta_2 + \alpha)^{-1} (g|_{K_n})$. Since the resolvents are positivity preserving, $(u_n)$ is monotonously increasing. Using the result already obtained for $p=2$ we get
\[(\tilde{\Delta} + \alpha)u_n =(\tilde{\Delta} + \alpha) (\Delta_2 + \alpha)^{-1} (g|_{K_n}) = (g|_{K_n}).\]
By construction we have for all $e\in E$ that $(u_n)_e \to u_e$ in $W^{2,1}(0,l(e))$. Thus the derivative of each component of $u$ converges uniformly on the edges, and $u$ converges uniformly on compacta. Thus $u\in \dom \widetilde{\Delta}$ and $ -u_e'' + \alpha u_e = g_e$ in the sense of distributions. Hence we have shown $\Delta_p u = \tilde{\Delta} u$ for all $u\in \dom \Delta_p$.\\
\end{proof}
The previous proposition gives the following information on the generators.
\begin{theorem}\index{graph Dirichlet form!$L^p$ generator}
Let $p\in [1,\infty]$, then $\Delta_p$ is a restriction of the weak Laplacian, i.e. we have
\[\Delta_p u = \widetilde{\Delta} u\]
for all $u\in \dom (\Delta_p)$.
\end{theorem}
\begin{proof}
For $u\in \dom (\Delta_p)$ the function $g:= (\Delta_p + \alpha) u$ is in $L^p(\FX_\Gamma,\omega)$ and the previous proposition shows that
\[(\widetilde{\Delta} + \alpha) u = g = (\Delta_p + \alpha ) u\]
implying the claim.
\end{proof}
Under an additional assumption we get the following characterization of the $L^p$ generators.
\begin{theorem}
Assume that for all rays $\Fp=(p_1,p_,\dots)$ we have
\[\sum_{k=1}^\infty l_i(p_k,p_{k+1}) \omega(p_k,p_{k+1}) = \infty.\]
Then for $1\leq p < \infty$ the domain of $\Delta_p$, the $L^p$ generator of $\CS$, is given as
\[\dom (\Delta_p) = \{ u\in \dom (\widetilde{\Delta}) \mid u\in L^p(\FX_\Gamma,\omega), \widetilde{\Delta} u \in L^p(\FX_\Gamma,\omega)\}.\]
\end{theorem}
\begin{proof}
We already know from the previous theorem that the domain of the $L^p$ generator is a subset of the given set. It remains to show the converse inclusion. For such a $u$ we have that $g:=\widetilde{\Delta}u + \alpha u$ belongs to $L^p$ and therefore the function $\tilde{u} := (\Delta_p + \alpha)^{-1} g$ belongs to $\dom \Delta_p$. Furthermore $\tilde{u}$ solves $\widetilde{\Delta} \tilde{u} + \alpha \tilde{u} = g$, and therefore $u-\tilde{u}$ is an $\alpha$-harmonic $L^p$-function. By the assumption on the paths, the minimum principle immediately gives $u-\tilde{u} =0$. This finishes the proof.
\end{proof}
The condition on the rays is the same as all rays having infinite volume in general representation. The theorem has an interesting corollary.
\begin{coro}
Under the assumptions of the previous theorem, for all $u \in L^2(\FX_\Gamma,\omega)$ with $\widetilde{\Delta} u \in L^2(\FX_\Gamma,\omega)$ we have that
\[\CS(u) < \infty.\]
\end{coro}
Using the above results, we now turn to the connection between the continuous Laplacians and the discrete Kirchhoff Laplacians. The following result is essentially an $L^p$ version of the so-called Krein resolvent formula. Before we state it we remark the following. The resolvent of the Dirichlet form $\CS$ restricted to the set $\bigoplus_{e\in E} W^{1,2}_o((0,l_i(e)),\omega)$ is nothing but the part Green operator $G_\alpha$. The resolvent formula gives a description of the difference of the original resolvent associated with $\CS$ and the part Green operator. Before we state and prove it we need to note the following. In section 3 we have introduced $M^{-1}\partial_n G_\alpha$ as adjoint of the operator $H_\alpha:\ell^2(\FV,M) \to L^2(\FX_\Gamma,\omega)$. Below we need the adjoint of $H_\alpha:\ell^2(\FV,M_\alpha) \to L^2(\FX_\Gamma,\omega)$ which is given as $M_\alpha^{-1}\partial_n G_\alpha$. We easily obtain that
\[M^{-1}\partial_n G_\alpha = \tfrac{M_\alpha}{M} M_\alpha^{-1}\partial_n G_\alpha.\]
\begin{theorem}\index{graph Dirichlet form!resolvent formula}
Let $\FX_\Gamma$ be a metric graph and $\CS$ a graph Dirichlet form. Then for all $f\in L^2$ we have
\[(\Delta_\CS + \alpha)^{-1} f = G_\alpha f + H_\alpha (\triangle + \triangle_\alpha + 1)^{-1} M_\alpha^{-1}\partial_n G_\alpha f.\]
\end{theorem}
\begin{proof}
For all $\alpha>0$ and $f\in L^2$ we have
\[\CS((\Delta_\CS + \alpha)^{-1} f,\phi) + \alpha((\Delta_\CS + \alpha)^{-1} f, \phi)_{L^2}= (f,\phi)_{L^2}\]
for all $\phi\in \CD(\CS)$ and
\[\CS(G_\alpha f, \phi_o) +\alpha (G_\alpha f,\phi_o)_{L^2} = (f,\phi_o)_{L^2}\]
for all $\phi_o\in \CD(\CS)$ with $\phi|_\FV \equiv 0$. Thus, the operator $(\Delta_\CS + \alpha)^{-1} - G_\alpha$ is the projection of the resolvent onto the space of functions which are $\alpha$-harmonic on each edge, in particular
\[((\Delta_\CS + \alpha)^{-1} - G_\alpha) f = H_\alpha ((((\Delta_\CS + \alpha)^{-1} - G_\alpha)f)|_\FV).\]
We set $R_H = ((\Delta_\CS + \alpha)^{-1} - G_\alpha)$. Subtracting the second equation from the first gives
\[\CS(R_H f,\phi_h) + \alpha (R_H f,\phi_h)_{L^2} = (f,H_\alpha\phi|_\FV)_{L^2},\]
where $\phi = \phi_o + \phi_h$ and $\phi_h = H_\alpha \phi |_\FV$. On the other hand using the Kirchhoff Dirichlet form, the same expression equals
\[\CS(R_H f,\phi_h) + \alpha (R_H f,\phi_h) = (R_H f, \phi_h)_{\ell^2}+ \CQ_\alpha(R_H f, \phi_h) + \CQ(R_H f, \phi_h).\]
Furthermore consider the adjoint of the operator $H_\alpha: \ell^2(\FV,M_\alpha) \to L^2(\FX_\Gamma,\omega)$ and we have
\[(f,H_\alpha\phi|_\FV)_{L^2} = (M_\alpha^{-1}\partial_n G_\alpha f, \phi_h)_{\ell^2}.\]
Combining the three equations above we have
\[(R_H f, \phi_h)_{\ell^2}+ \CQ_\alpha(R_H f, \phi_h) + \CQ(R_H f, \phi_h) = ( M_\alpha^{-1}\partial_n G_\alpha f, \phi_h)_{\ell^2}.\]
As this holds for arbitrary $\phi_h\in \CD(\CQ + \CQ_\alpha)$, we have by the characterization of resolvents the equality
\[ (R_H f)|_\FV = (\triangle + \triangle_\alpha + 1)^{-1} M_\alpha^{-1}\partial_n G_\alpha f.\]
Applying $H_\alpha$ to each side we obtain
\[H_\alpha ((R_H f)|_\FV) = H_\alpha (\triangle + \triangle_\alpha + 1)^{-1} M_\alpha^{-1}\partial_n G_\alpha f.\]
\end{proof}
In the case of general $p$ we obtain the resolvent formula only under an additional assumption.
\begin{theorem}
Assume that for all rays $\Fp=(p_1,p_,\dots)$ we have
\[\sum_{k=1}^\infty l_i(p_k,p_{k+1}) \omega(p_k,p_{k+1}) = \infty.\]
Then for $1\leq p < \infty$ we have for all $f\in L^p$
\[(\Delta_p + \alpha)^{-1} f = G_\alpha f + H_\alpha (\triangle_\alpha + \triangle +1 )^{-1} M_\alpha^{-1} \partial_n G_\alpha f,\]
where $(\triangle_\alpha + M_\alpha)^{-1}$ is the resolvent of the Kirchhoff Dirichlet form on $\ell^p(\FV,M)$.
\end{theorem}
\begin{proof}
First we have with $u,f\in L^p$
\[(\Delta_p + \alpha)^{-1} f =u\]
if and only if
\[\widetilde{\Delta} u + \alpha u = f\]
by the characterization of the $L^p$ generator. Furthermore
\[\widetilde{\Delta} u + \alpha u = f\]
if and only if
\[(\triangle_\alpha+\triangle+M_\alpha) U = \partial_n G_\alpha f\]
with $U = u |_\FV$. We multiply this equation with $M_\alpha^{-1}$ and we obtain
\[M_\alpha(x)^{-1}(\triangle_\alpha+\triangle)U (x)+ U(x) = M_\alpha(x)^{-1}\partial_n G_\alpha f(x).\]
We know that the right hand side is in $\ell^p(\FV,M_\alpha)$ and the left hand side is the formal discrete Laplacian applied to $U$. We now apply a result from \cite{KL-10} to our situation. To do so, we set $m(x):= M_\alpha (x)$ and $\widetilde{L} U(x) := M_\alpha(x)^{-1}(\triangle_\alpha+\triangle)U (x)$. The result form \cite{KL-10} is the following characterization of operators of discrete Dirichlet forms:\\
Let $\CQ$ be the discrete Dirichlet form on $\ell^2(V,m)$ associated with $\widetilde{L}$ and assume that for all rays $(p_1,p_2,\dots)$ we have
\[\sum_{k=1}^\infty m(p_k)=\infty\]
then the generators $L_p$ of the resolvents on $\ell^p$ are given as
\[\dom \triangle_p = \{U\in \ell^p \mid \widetilde{L} U \in \ell^p\}\]
and
\[L_p U= \widetilde{L} U = (\frac{1}{m(x)} \sum_{y} b(x,y) (U(x)-U(y)) + \frac{c(x)}{m(x)} U(x))_{x\in V}.\]
To apply this result we need to check that $\sum_{k=1}^\infty M(p_k) = \infty$ for all rays $(p_1,p_2,\dots)$. But this is implied by the assumption on the paths. Thus we have in our case
\[U= (\triangle_\alpha + \triangle + 1)^{-1} M_\alpha^{-1} \partial_n G_\alpha f.\]
As by definition $U\in \ell^p(\FV,M_\alpha)$, we obtain that $H_\alpha U \in L^p(\FX_\Gamma,\omega)$. Applying again the characterization on the $L^p$ generators of $\CS$, we obtain the claim.
\end{proof}
We finish this chapter on the Laplacians with a look at essential self-adjointness. We first consider the case of vanishing $\CQ$ and complete $\FX_\Gamma$.
\begin{theorem}\index{graph diffusion Dirichlet form!essential selfadjointness}
Let $\FX_\Gamma$ be intrinsically complete. Then $\Delta_\CE$ is essentially self-adjoint on $(\dom \Delta_\CE)_c$.
\end{theorem}
\begin{proof}
We denote by $\Delta_c$ the restriction of $\Delta_\CE$ to $(\dom \Delta_\CE)_c$. Assume that $\Delta_c$ is not essentially selfadjoint, then the range $(\Delta_c + 1) (\dom \Delta_\CE)_c$ is not dense in $L^2$ and thus there exists $u\in L^2$ such that for all $\phi\in (\dom \Delta_\CE)_c$ we have
\[(u,(\Delta_c + 1) \phi) =0.\]
Thus on each edge we have that $u_e$ is $\CC^\infty$ and satisfies $u_e'' = u_e$, in particular $(u''_e)_e \in L^2$. Now let $\phi \in \CD_c(\Delta_\CE)$ be supported in a neighborhood of a vertex $x\in \FV$. Performing partial integration on each edge we obtain
\[\sum_{e\sim x} u_e' (0) \phi_e(0) = \sum_{e\sim x} u_e(0) \phi_e'(0)\]
for all such $\phi$. Note that such a $\phi$ is continuous and satisfies the Kirchhoff conditions. Choosing $\phi$ such that $\phi=1$ in a small neighborhood of $x$, we obtain that $u$ satisfies Kirchhoff conditions in $x$ and thus
\[0 = \sum_{e\sim x} u_e(0) \phi_e'(0)\]
for all such $\phi$ which implies that $u_e(0) = u_{e'}(0)$ for all $e,e'\sim x$. In particular we have that $u\in \CD(\CE)_\loc$ and therefore
\[\CE(u,\phi) + (u,\phi)=0\]
for all $\phi \in \CD(\Delta_\CE)_c$. By continuity we can extend this equation to all $\phi \in \CD(\CE)_c$.\\
Let $\psi:\IR\to [0,1]$ be $\CC^\infty$ such that $\psi=1$ on $[0,1]$ and $\psi =0$ on $[2,\infty)$. Fix $y\in \FX_\Gamma$ and define $\phi_n(x):=\psi(\frac{d_i(x,y)}{n})$. Hence $0\leq \phi_n\leq 1$ and $\phi_n$ converges locally uniformly to $1$ as $n\to \infty$. Since $d$ agrees with the intrinsic metric, we have $|d'|\leq 1$. Furthermore
\[|\phi_n'| \leq  |\psi'(\tfrac{d}{n})| \tfrac{d'}{n} \leq \tfrac{\|\psi'\|_\infty}{n}\]
which also implies that $\phi_n\in \CD(\CE)$ with compact support and therefore $\phi_n^2 u \in \CD(\CE)_c$. Thus we have
\[0\geq  -(\phi_n^2 u,u)= \CE(\phi_n^2u,u) = \int\limits_{\FX_\Gamma} \phi_n^2 u'^2 + 2\phi_n u \phi_n' u'\: d\omega,\]
and using Cauchy-Schwarz
\[\int\limits_{\FX_\Gamma} \phi_n^2 |u'|^2\: d\omega \leq 2 \int\limits_{\FX_\Gamma} \phi_n |u| |\phi_n'| |u'|\:d\omega \leq  2 \|\phi_n\|_\infty \|u\|_2 \|\phi_n u'\|_2.\]
Hence
\[\|\phi_n u'\|_2 \leq 2n^{-1} \|\psi'\|_\infty \|u\|_2\]
which implies by Fatou's lemma that $\CE(u)=0$, in particular $u_e=  u''_e =0$, on each edge since $u_e$ is a strong solution.
\end{proof}
\begin{remark}
In \cite{Ku-04} the previous theorem was shown for graphs with edge lengths uniformly bounded from below and forms $\CQ$ which are uniformly bounded. There, cut-off functions are constructed which are constant in vertex neighborhoods such that they don't change the boundary conditions. Due to the uniform lower bound on the edge lengths their gradient is also uniformly bounded. This method surely doesn't work in our case. Nevertheless in our case we only treat the case $\CQ\equiv 0$.\\
On the other hand, a result from \cite{Ca-08}, says that if there is a boundary, then the operator is not essentially self adjoint. Thus, we see that the previous theorem is a characterization of essential self-adjointness for the case $\CQ\equiv 0$.
\end{remark}
We now treat the more general case of $\CQ\neq 0$.
\begin{theorem}\index{graph Dirichlet form!essential selfadjointness}
Assume that for all rays $\Fp =(p_1,p_2,\dots)$ we have
\[\sum_{k=1}^\infty l_i(p_k,p_{k+1})(\omega(p_k,p_{k+1})) =\infty.\]
Then the generator of the form $\CS$ is essentially self-adjoint on $(\dom \Delta_\CS)_c$.
\end{theorem}
\begin{proof}
Let us again denote the operator $\widetilde{\Delta}$ defined on the set $(\dom \Delta_\CS)_c$ by $\Delta_c$. Then as above one can show that $\Delta_c^*$ is given as the weak Laplacian on the set
\[\{u\in \dom \widetilde{\Delta} \mid u\in L^2, \widetilde{\Delta}u \in L^2\}.\]
As this set coincides with domain of the $L^2$ generator of $\CS$ by assumption on the rays, we obtain essential self-adjointness.
\end{proof}
\section*{Notes and remarks}
Our main motivation of this chapter was coming from the theory of non-local Dirichlet forms, in particular the article \cite{FLW-12}. It is still unclear how to construct an intrinsic metric in the general case of a graph Dirichlet form $\CS$. In the case when the reference measure is given as the sum of the Lebesgue measure on edges and delta measures on the vertices, it is very easy to give an intrinsic metric under certain assumptions on the parameters. This was done in \cite{MUW-12}. However, in the general case this is not possible as the energy measure is not absolutely continuous to the Lebesgue measure. To apply the arguments used in \cite{MUW-12}, the absolute continuity is necessary.\\
There are several papers concerning the interplay of continuous and discrete graphs. Only to mention a few, see \cite{Cattaneo,CW-07,Pa-06,Po-12}. However, these paper mostly  assume equilateral edge lengths or a lower bound on the edge lengths. It seems that we have found the right measure on the set of vertices. It would be interesting if the results which assume strong conditions on the edge lengths may carry over to our situation. One result in this direction is given as the Krein resolvent formula, which was known only in the case of equilateral edge lengths and uniform bounded vertex degree.\\
Just recently, see \cite{Huang,Folz}, one studied discrete graphs by means of metric graph. More precisely, given a discrete graph, the authors construct a metric graph. Then they apply a criterion of stochastic completeness to this setting and translate it back to the discrete graph. It seems that our approach also works here due to Theorem 3.32. In general one may try to study problems on discrete graphs by means of appropriate metric graphs. This thesis is yet another factor making this hope a realistic one. 
\chapter{Functions of finite energy}
In this final chapter we consider the graph Dirichlet form extended to the space of functions of finite energy. This space is introduced in the first section. Surprisingly, we can equip it with a norm which turns this space into a Hilbert space, Proposition 4.2. Connected with this issue are properties like recurrence. This is then studied, see Propositions 4.9 and 4.10, where we also consider potential theoretic related issues. Afterwards, in section 3 we return to the special case of diffusion Dirichlet forms. Using the tools developed in chapter 2, we study in Theorem 4.16 and Proposition 4.18 recurrence in terms of the metric properties. In the last section we turn back to our leitmotif $(C)$, i.e. using similar ideas as in chapter 3, we relate recurrence of the graph Dirichlet form to a certain discrete one in Theorem 4.23 and Corollary 4.24. The latter one is in close connection to what is known in general Dirichlet form theory as trace Dirichlet form.
\section{Functions of finite energy}
In this section we introduce the space of all functions such that the energy functional $\CS$ is finite and certain subspaces of it. Recall that the Dirichlet form $\CS$ on $L^2(\FX_\Gamma,\lambda_a)$ is given as $\CS= \CE + \CQ$, where $\CE$ is the diffusion part given by
\[\CE(u) = \int\limits_{\FX_\Gamma} |u'|^2 d\lambda_b\]
and $\CQ$ is the discrete part given by
\[\CQ(u)=\frac{1}{2}\sum_{x,y\in \FV} j(x,y)(u(x)-u(y))^2+\sum_{x\in \FV} k(x) u(x)^2.\]
By definition the quadratic functional $\CS(u)$ makes sense for all absolutely continuous functions if we allow for $\CS$ the value infinity.
\begin{definition}\index{functions of finite energy}
The space of all functions of finite energy of $\CS$ is the space $\widetilde{\CD}(\CS)$ consisting of all absolutely continuous functions such that $\CS(u)<\infty$.
\end{definition}
It is appropriate in this chapter to work in the canonical representation, as we essentially deal with the functional $\CS$.\medskip\\
As the constant functions are in $\widetilde{\CD}(\CS)$ the energy functional $\CS$ fails to be a norm. However we can equip this space with a norm. To do so choose $o\in \FX_\Gamma$ arbitrary but fixed and set for $u\in \widetilde{\CD}(\CS)$
\[ \|u\|_o^2 = |u(o)|^2 + \CS(u).\]
\begin{prop}\index{theorem!Sobolev embedding!finite energy version}
The space $\widetilde{\CD}(\CS)$ with the norm $\|\cdot\|_o$ is a Hilbert space and is locally, continuously embedded into the space of continuous functions, i.e. for all compact sets $\FY\subset \FX_\Gamma$ there exists a constant $C>0$ such that we have for all $x\in \FY$
\[|u(x)|^2 \leq C \|u\|_o^2,\]
in particular we have
\[\widetilde{\CD}(\CS)=\{u\in \CC(\FX_\Gamma)\mid \CQ(u)<\infty \mbox{ and } u'\in \bigoplus_{e\in E}L^2(0,l_c(e))\}.\]
Furthermore, for all $o'\in \FX_\Gamma$ the norm $\|\cdot\|_{o'}$ defines an equivalent norm on $\widetilde{\CD}(\CS)$.
\end{prop}
\begin{proof}
We only show the inequality, the remaining claims are easy consequences of it. Let $\FY$ be a compact subset of $\FX_\Gamma$. Assume first that $o$ and $x$ belong to the same connected component. Thus there exists a path $\Fp_o^x$ such that
\begin{eqnarray*}
|u(x)|^2 &\leq& |u(o)|^2 + |u(x) - u(o)|^2\\
&\leq & |u(o)|^2 + |\int\limits_{\Fp_o^x} u'(t)\ dt|^2 \\
&\leq & |u(o)|^2 + C \CE(u)
\end{eqnarray*}
where we have used the Cauchy-Schwarz inequality. Here, the constant $C>0$ depends on the canonical lengths of the edges in the path. Now, assume that $o$ and $x$ are not on the same connected component, but that there exist a jump path $x_1, \dots, x_n$ with $j(o,x_1)>0$ and $j(x_n,x)>0$. We then obtain
\begin{eqnarray*}
|u(x)|^2 &\leq& C(|u(o)|^2 + \sum_{k=1}^{n-1} \frac{j(x_k,x_{k+1})}{j(x_k,x_{k+1})}|u(x_k)- u(x_{k+1})|^2 + |u(x) - u(o)|^2)\\
&\leq& C (|u(o)|^2 +  \CQ(u))
\end{eqnarray*}
where we have applied Cauchy-Schwarz to get the constant $C>0$ which this time depends on the jump weight $j$. For the general case of $o$ and $x$ there exists by irreducibility a combination of topological and jump paths connecting both.
\end{proof}
We continue with a characterization of convergence in $\widetilde{\CD}(\CS)$. We follow the arguments given in \cite{Sch-11} in the case of discrete Dirichlet forms.
\begin{prop}
Let $u\in \widetilde{\CD}(\CS)$ and $(u_n)_n$ be a sequence in $\widetilde{\CD}(\CS)$. Then $u_n\to u$ with respect to $\|\cdot\|_o$ if and only if $u_n$ converges pointwise to $u$ and
\[\limsup_{n\to\infty} \CS(u_n) \leq \CS(u).\]
\end{prop}
\begin{proof}
Let $u_n\to u$ in $\widetilde{\CD}(\CS)$. Then by the previous proposition it converges pointwise and clearly $\CS(u_n)\to \CS(u)$. Thus the first direction is obvious. To prove the converse direction, let $(u_n)$ be as stated. Then $(u_n)$ is bounded in $\widetilde{\CD}(\CS)$. By weak compactness of balls in Hilbert spaces every subsequence has a weakly convergent subsequence. By pointwise convergence all those limits must coincide and we get that $u_n\to u$ weakly in $\widetilde{\CD}(\CS)$. Moreover
\begin{eqnarray*}
0 &\leq & \CS(u-u_n) + (u(o)- u_n(o))^2\\
&\leq& \CS(u) + u(o)^2 + \CS(u_n) + u_n(o)^2 - 2\langle u,u_n\rangle_o
\end{eqnarray*}
and thus $u_n\to u$ in $\widetilde{\CD}(\CS)$.
\end{proof}
One can associate such a space of functions of finite energy to general Dirichlet forms. Abstractly, this space is defined to be the space of all measurable functions $u$ on $\FX_\Gamma$ such that for all $n\geq 1$ the truncated function $-n\vee u \wedge n \in \CD_\loc(\CS)$ and $\CS(u) = \lim\limits_{n\to \infty} \CS(-n\vee u \wedge n) < \infty$. This space is then known as the reflected Dirichlet space, see \cite{CF-12}. The next theorem shows that the abstract definition coincides in our situation with the just defined space of functions of finite energy.
\begin{theorem}\index{reflected Dirichlet space}
Let $\FX_\Gamma$ be a metric graph and $\CS$ be a graph Dirichlet form. Then the space of functions of finite energy $\widetilde{\CD}(\CS)$ agrees with the reflected Dirichlet space.
\end{theorem}
\begin{proof}
One inclusion is trivial and the second follows from the propositions above. Choose $u$ from the reflected Dirichlet space, then as $-n \vee u \wedge n$ converges pointwise to $u$ and by assumption
\[\limsup_{n\to \infty}\CS(-n \vee u \wedge n) = \lim_{n\to\infty}\CS( -n \vee u \wedge n) = \CS(u),\]
the proposition above gives the other inclusion as  $-n \vee u \wedge n \in \widetilde{\CD}(\CS)$, since otherwise the limit would be infinity. Thus we have that $u\in \widetilde{\CD}(\CS)$.
\end{proof}
By definition, for all edge weights $\nu$ we have the identity
\[\CD(\CS) = \widetilde{\CD}(\CS)\cap L^2(\FX_\Gamma,\nu),\]
where the right hand side is called active reflected Dirichlet space in the abstract theory. We now consider the reflected Dirichlet space to $\CD(\CS)_o$.
\begin{prop}
The reflected Dirichlet space of $\CD(\CS)_o$ coincides with $\widetilde{\CD}(\CS)$.
\end{prop}
\begin{proof}
This follows immediately from the fact the local spaces of $\CD(\CS)$ and $\CD(\CS)_o$ coincide.
\end{proof}
The proposition shows that the intersection of the reflected Dirichlet space of $\CD(\CS)_o$ and $L^2$ does not coincide with $\CD(\CS)_o$ in general. This leads to the concept of the extended Dirichlet space.
\begin{reminder}\index{Dirichlet form!extended}
Let $(\CS,\CD)$ be a Dirichlet form on some $L^2(X,m)$ space. The family $\CD(\CS)_e$ consisting of all $m$-measurable functions $u$ on $X$ with $|u|<\infty$ $m$-a.e. and there exists an $\CS$-Cauchy sequence of functions $(u_n)$ in $\CD$ such that $u_n \to u$ $m$-a.e., is called the extended Dirichlet space. Such a sequence is called an approximating sequence.
\end{reminder}
Due to this definition we analyze the extended Dirichlet spaces of $\CS$ and $\CS_o$ in the following theorem.
\begin{theorem}\index{graph Dirichlet form!extended Dirichlet space}
The extended Dirichlet space of $\CS$ is a closed subspace of $(\widetilde{\CD}(\CS),\|\cdot\|_o)$ and it is given by
\[\overline{\CD(\CS)}^{\|\cdot\|_o}.\]
The extended Dirichlet space of $\CS_o$ is a closed subspace of $(\widetilde{\CD}(\CS),\|\cdot\|_o)$ and it is given by
\[\overline{\CD(\CS)_c}^{\|\cdot\|_o} = \widetilde{\CD}(\CS)_o.\]
\end{theorem}
\begin{proof}
Let $(u_n)_n \subset \CD(\CS)$ converging to $u\in \CD(\CS)_e$, i.e. $u_n \to u$ a.e. and $(u_n)_n$ is $\CS$-Cauchy. Thus there is $o'\in \FX_\Gamma$ with $u_n(o')\to u(o')$. By the Cauchy property we have
\[\limsup_{n\to \infty} \CS(u_n) \leq \CS(u).\]
Thus $u_n\to u$ in $\|\cdot\|_o$. \\
Conversely $\|\cdot\|_o$ implies convergence on compacta, and thus pointwise convergence and by definition of this convergence we also get the Cauchy property.\\
Analogously we get that
\[(\CD(\CS)_o)_e = \overline{\CD(\CS)_o}^{\|\cdot\|_o},\]
and the claim follows from the density of $\CD_c(\CS)$ and as $\|\cdot\|_\CS$ convergence forces $\|\cdot\|_o$ convergence.
\end{proof}
Due to the theorem above we have the inclusions
\[\widetilde{\CD}(\CS)_o \subset \CD(\CS)_e \subset \widetilde{\CD}(\CS),\]
and only the space in the middle depends on the canonical measure. Note that in general these inclusions are strict. Thus we are left to study when some of these inclusions are equalities. By general theory of Dirichlet forms one gets that the set of all square integrable functions of the extended Dirichlet space coincide with the original Dirichlet space. This immediately implies the following.
\begin{coro}
The form $\CS$ on $L^2(\FX_\Gamma, \nu)$ is regular if and only if
\[\widetilde{\CD}(\CS)_o = \CD(\CS)_e.\]
\end{coro}
\section{Recurrence and transience}
We continue in this section with the investigations of the just introduced spaces with some potential theory. Basically, given a Dirichlet form $\CS$  one might ask whether for $u\in L^1(\FX_\Gamma,\nu)$, $u\geq 0$, the limit
\[Gu(x):= \lim_{N\to \infty} \int\limits_0^N P_t u(x) \ dt\]
is finite $\nu$-a.e. or not. It turns out, that the finiteness of this limit is independent of the particular choice of $u$. In the case of finiteness, the semigroup is called transient, whence if not it is called recurrent.
It turns out that transience of a Dirichlet space is equivalent to the extended Dirichlet space being a Hilbert space with respect to the energy functional $\CS$. It is immediate that this is in turn equivalent that the constant functions do not belong to the extended Dirichlet space. As we consider the forms $\CS$ and $\CS_o$ separately, this leads to the following definition.
\begin{definition}\index{recurrence}
The Dirichlet space $(\CS,\CD(\CS))$ is called (totally) recurrent if \[\widetilde{\CD}(\CS)_o = \widetilde{\CD}(\CS),\]
and $(\CS,\CD(\CS))$ is called $\nu$-recurrent if
\[\CD(\CS)_e = \widetilde{\CD}(\CS).\]
\end{definition}
As antonym to recurrence and $\nu$-recurrence we will use the expressions transient and $\nu$-transient.\medskip\\
We start with the following characterization of transience in terms of the Poincaré inequality.\index{transience}
\begin{prop}\index{inequality!Poincaré}
The Dirichlet space $(\CS, \CD(\CS))$ is transient if and only if the norms $\|\cdot\|_o$ and $\CS(\cdot)$ on the space $\widetilde{\CD}(\CS)_o$ are equivalent. In particular, the latter inequality means that Poincaré's inequality holds, i.e. for all $x\in \FX_\Gamma$ there is a constant $C_x>0$ such that for all $u\in \widetilde{\CD}(\CS)_o$ we have
\[|u(x)|^2 \leq C_x \CS(u).\]
\end{prop}
The proof is immediate. It is clear that an analogous version could be formed for $\nu$-transience. One further observation is that whenever a non-trivial killing term appears, the form in question is automatically transient. Thus we assume from now on that the killing weights $k$ of $\CS$ are identically zero.
\begin{assumption}
The killing term vanishes, i.e. $k\equiv 0$.
\end{assumption}
As in the case of non-regularity, we can decompose the space into a transient part and a space of harmonic, respectively almost harmonic functions. Given a graph Dirichlet form we have the orthogonal decomposition
\[\widetilde{\CD}(\CS) = \widetilde{\CD}(\CS)_o \oplus \CM \]
where $\CM$ is the space of all monopoles of finite energy in the vertex $o$. A monopole in a point $x\in \FX_\Gamma$ is a function $u\in \CD(\CS)_\loc$ such that
\[\CS(u,\phi)= \phi(x)\]
for all $\phi\in \CD(\CS)_c$. For discrete Dirichlet forms it was shown in \cite{JP,Sch-11} that transience is equivalent to the existence of monopoles of finite energy. The same holds in our setting.
\begin{prop}\index{monopole of finite energy}
The form $\CS_o$ is transient if and only if there exists a monopole of finite energy.
\end{prop}
\begin{proof}
Assume that $\CS_o$ is transient, using the decomposition
\[\widetilde{\CD}(\CS) = \widetilde{\CD}(\CS)_o \oplus \CM, \]
then a function $u\in \CM$ satisfies
\[ u(o) \phi(o) + \CS(u,\phi)=0\]
for all $\phi\in \CD_c(\CS)$. Thus we have a monopole of finite energy as $\CM$ is nontrivial. Assume conversely that a monopole of finite energy exists for some $x\in \FV$, then we can normalize it such that $u(x)=1$ and we have
\[\phi(x) + \CS(u,\phi)=0.\]
Since $\|\cdot\|_o$ and $\|\cdot\|_x$ are equivalent also the space of monopoles in the point $o$ is non-trivial, which yields the claim.
\end{proof}
We turn now to the study of harmonic and superharmonic functions in the extended Dirichlet space. We closely follow \cite{S-94}.\index{harmonic functions}
\begin{lemma}
The space of harmonic functions $\CH_0$ is closed in $\widetilde{\CD}(\CS)$ and for $u\in \CH_0$ and $v\in \widetilde{\CD}(\CS)_o$ we have
\[\CS(u,v)=0.\]
\end{lemma}
\begin{proof}
To show the closedness let $(u_n)_n\subset \CH_0$ and $u \in \widetilde{\CD}(\CS)$ such that $u_n\to u$ in $\|\cdot\|_o$. Then we have by definition of harmonicity
\[\CS(u_n, v)=0\]
for all $v\in \widetilde{\CD}(\CS)$ with compact support. In particular we have
\[\CS(u,v)= \CS(u-u_n,v) + \CS(u_n,v)= \CS(u-u_n,v)\leq \|u-u_n\|_o\CS(v)\]
and by assumption on $(u_n)_n$ the left hand side is arbitrarily small and hence $u$ is harmonic. The second claim follows by extending this equality to all $v\in \widetilde{\CD}(\CS)_o$ by continuity of the functional $\CS$.
\end{proof}
We show now an $\CS$-orthogonal decomposition. This is also known as the Royden decomposition.\index{theorem!Royden decomposition}
\begin{theorem}
For all $u\in \widetilde{\CD}(\CS)$ there exists unique $u_o\in \widetilde{\CD}(\CS)_o$ and $u_h \in \CH_0$ such that
\[u= u_o + u_h\]
and
\[\CS(u)= \CS(u_o) + \CS(u_h).\]
\end{theorem}
\begin{proof}
We assume that $\CS$ is transient as in the recurrent case $\CH_0$ is trivial. Let $d^2 = \inf_{v\in \widetilde{\CD}(\CS)_o} \CS(u-v)$ and $(v_n)_n \subset \widetilde{\CD}(\CS)_o$ such that
\[\CS(u-v_n) \to d^2.\]
As $\CS$ is a positive quadratic functional the sequence $(v_n)_n$ is an $\CS$-Cauchy sequence and by transience and the Poincaré inequality which therefore holds for functions in $\widetilde{\CD}(\CS)_o$ this sequence is also pointwise a Cauchy sequence. As $\widetilde{\CD}(\CS)_o$ is closed there exists a minimizer $u_o\in \widetilde{\CD}(\CS)_o$ with $v_n\to u_o$. We set $u_h = u-u_o$ and need to show that $u_h$ is harmonic. To this extent we know that for arbitrary $v\in\widetilde{\CD}(\CS)$ with compact support the function
\[t\mapsto \CS(u_h + t v)\]
has a minimum in $t=0$. As
\[\CS(u_h + tv) = \CS(u) + 2t \CS(u_h, v) + t^2 \CS(v)\]
the derivative of the above function in $t=0$ equals $2\CS(u,v)$, which has to be zero by the minimum property. As $v\in \widetilde{\CD}(\CS)_c$ was arbitrary we have shown that $u_h$ is harmonic.\\
To show the uniqueness of the decomposition assume there are $u_o'\in \widetilde{\CD}(\CS)_o$ and $u_h' \in \CH_0$ such that
\[u=u_o + u_h = u_o' + u_h'.\]
This gives that $u_o-u_o' = u_h - u_h'$ is harmonic and in $\widetilde{\CD}(\CS)_o$ which implies by transience that it is zero.
\end{proof}
Next we want to prove a Riesz type decomposition for superharmonic functions. To do so we recall some important facts from the potential theory of regular Dirichlet forms, see chapter 2 from \cite{FOT-11}. An analogue of Theorem 2.2.1 for transient Dirichlet forms from \cite{FOT-11} says that a function $u\in \widetilde{\CD}(\CS)_o$ is superharmonic, i.e.
\[\CS(u,\phi)\geq 0, \quad \forall \phi\in \CD(\CS)_c, \phi \geq 0\]
if and only if there exists a positive Radon measure $\mu$ on $\FX_\Gamma$ such that
\[\CS(u,\phi)= \mu(\phi).\]
In particular such a $u$ is nonnegative. One then sets $U_0\mu = u$, where $U_0$ is the potential operator of zeroth order. In this equivalence one direction is obvious, whereas the other one follows from the Riesz representation theorem for positive functionals on $\CC(\FX_\Gamma)$. Together with the previous lemma we get the following theorem.
\begin{theorem}\index{theorem!Riesz decomposition}
Let $\CS$ be transient. Then for all non-negative superharmonic functions $u\in \widetilde{\CD}(\CS)$ there exists $u_o\in \widetilde{\CD}(\CS)_o$ and $u_h\in \CH_0$ such that
\begin{itemize}
\item[(i)] $u_o$ is a nonnegative superharmonic function,
\item[(ii)] there exists a nonnegative Radon measure $\mu$ such that $u_o= U_0 \mu$
\item[(iii)] $u_h$ is a nonnegative harmonic function.
\end{itemize}
\end{theorem}
\begin{proof}
By the Royden decomposition we get $u_o\in \widetilde{\CD}(\CS)_o$ and $u_h\in \CH_0$ with $u=u_o+u_h$. In particular $u_o$ is a superharmonic function of the regular Dirichlet form $\CS_o$ and therefore nonnegative. In particular this yields the existence of the Radon measure $\mu$. To show the last part, note that $u \wedge u_o$ belongs to $\widetilde{\CD}(\CS)_o$ and thus we have
\[\CS(u- u\wedge u_o) =\CS( (u-u_o)_+) \leq \CS(u-u_o).\]
As $u_o\in \widetilde{\CD}(\CS)_o$ is the minimizer of $\CS(u-v)$ on $\widetilde{\CD}(\CS)$ we obtain that $u\wedge u_o = u_o$ and therefore
\[0\leq u-u_o = u_h\]
yielding the nonnegativity of $u_h$.
\end{proof}
The theorem has two interesting corollaries.
\begin{coro}
For all $u\in \CH_0$ there exist nonnegative $u_+, u_- \in \CH_0$ such that
\[u = u_+ - u_-.\]
\end{coro}
\begin{proof}
In the case of recurrence nothing is to proof, as all harmonic functions are constant.\\
As $u$ is harmonic, both functions $-u \vee 0$ and $ u \wedge 0$ are superharmonic and the Royden decomposition gives the existence of nonnegative $v_+,v_- \in \widetilde{\CD}(\CS)_o$ and $u_+,u_- \in \CH_0$ such that
\[ -u \vee 0 = v_- - u_-,\quad  u\wedge 0 = v_+ - u_+.\]
As we have
\[u= - (u\wedge 0 - u \vee 0) = -v_-- v_+ +u_+- u_-\]
we get by the uniqueness of the Royden decomposition that $v_+ = v_-$ and $u = u_+ - u_-$.
\end{proof}
The next one is a version of Virtanen's Theorem.
\begin{coro}
For all $h\in \CH_0$ there exists a sequence of bounded functions $(h_n)_n\in \CH_0$ such that $\|h_n-h\|_o \to 0$ as $n\to \infty$.
\end{coro}
\begin{proof}
Again we only consider the transient case. Due to the previous corollary we may assume that $h\geq 0$. For $n\in \IN$ we set $g_n:= h \wedge n$. Then $g_n$ is superharmonic and by the characterization of convergence in $\widetilde{\CD}(\CS)$ we have $g_n \to h$ in $\|\cdot\|_o$ norm. The Royden decomposition gives again the existence of nonnegative functions $s_n \in \widetilde{\CD}(\CS)_o$ and $h_n\in H_0$ such that $g_n = s_n + h_n$. As $s_n$ and $h_n$ are $\CS$ orthogonal we have
\[\CS(h- g_n) = \CS(s_n) + \CS(h-h_n)\]
which gives that $\CS(s_n)\to 0$ and $\CS(h-h_n)\to 0$ as the left hand side goes to zero as well. Since $\CS$ is a scalar product on $\widetilde{\CD}(\CS)$ we infer that $s_n$ converges also pointwise to zero. Thus $h_n\to h$ in $\|\cdot\|_o$ norm.
\end{proof}
\section{The diffusion part}
In this section we focus on the diffusion part $\CE$ again. In \cite{St-94} it was shown that for strongly local Dirichlet forms recurrence is equivalent to the strong notion of parabolicity. More precisely, it says that each positive superharmonic function is constant. Recall that we have only proven that recurrence is equivalent to that each superharmonic function of finite energy is constant. We start with the result taken from \cite{St-94} in our framework. Note that we use in this theorem the intrinsic metric and measure, in particular we assume that the graph is intrinsically complete.
\begin{theorem}
Let $\FX_\Gamma$ be a connected, complete metric graph and $\CE$ a diffusion Dirichlet form. If for some point $x_o\in \FX_\Gamma$ we have
\[\int\limits_1^\infty \frac{r}{\omega(B_r(x_o))} dr = \infty\]
then $\FX_\Gamma$ is parabolic.
\end{theorem}
Using the graph structure and if we assume that the length function is radial with respect to some $x_o\in \FX_\Gamma$ we deduce the following corollary.
\begin{coro}\index{graph diffusion Dirichlet form!parabolicity}
Under the assumptions of the theorem, assuming additionally that the length function is radial and define
\[A(S_k,l_k) = \frac{  \sum_{j=1}^k \omega(S_j) l_j}{ \sum_{j=1}^k l_j}\]
i.e. the weighted arithmetic mean of the spheres with weights $l_j$. Then if
\[\sum_{k=1}^\infty \frac{l_{k}}{ A(S_k,l_k)} = \infty\]
then $\CE$ is parabolic. In particular this is the case if the  sequence $A(S_k,l_k)$ of arithmetic means is bounded from above.
\end{coro}
\begin{proof}
An easy calculation shows
\begin{eqnarray*}
\int\limits_0^\infty \frac{r}{V(r)} \:dr &=& \sum_{k=0}^\infty \int\limits_{r_k}^{r_{k+1}} \frac{r}{V(r)} \:dr\\
&\geq& \sum_{k=0}^\infty \frac{(r_{k+1} + r_k)(r_{k+1}-r_k)}{\sum_{j=1}^{k+1} \omega(S_j) l_j}\\
&\geq& \sum_{k=0}^\infty \frac{l_{k+1}}{ A(S_{k+1},l_{k+1})}
\end{eqnarray*}
The in particular part follows since the graph is assumed to be complete.
\end{proof}
The intrinsically incomplete case is again the hard one and again we consider the canonical metric. We already know from section 2.2 that we can extend each function of finite energy continuously to the canonical boundary. Furthermore functions $u\in \widetilde{\CD}(\CS)_o$ have canonical boundary value zero. This leads to the following proposition as in this case the constant functions cannot be approximated.
\begin{prop}
Let $\FX_\Gamma$ be a metric graph with graph Dirichlet form $\CS$. \begin{itemize}
\item[(i)] If $\FX_\Gamma$ is canonically incomplete, then it is transient.
\item[(ii)] Let $\nu$ be a graph measure. If there exists a canonical boundary point with infinite measure, then it is $\nu$-transient.
\end{itemize}
\end{prop}
We finish this section with a discussion, when the non-local part of a graph diffusion could be considered as a small perturbation concerning recurrence. We say that $\CQ$ is $\CE$-form-bounded if the inequality\index{form-boundedness}
\[\CQ(u) \leq C \CE(u),\]
holds for all absolutely continuous functions $u$ with $u' \in L^2(\FX_\Gamma)$. Note that from this definition we see that $\FX_\Gamma$ has to be connected, since otherwise the right hand side may vanish for a particular choice of $u$, whereas the left hand side does not. Thus we assume that the graph $\FX_\Gamma$ is connected. We start with a necessary condition in terms of the $0$-th order capacity, defined below.
\begin{lemma}
Let $\FX_\Gamma$ be a metric graph and $\CS$ a graph Dirichlet form. If $\CQ$ is $\CE$-form-bounded, then there is a constant $C>0$, such that
\[\sum_{y\in \FV}  j(x,y) + k(x)  \leq C\ {\mathrm{cap}_0}(\{x\},\supp \CQ) .\]
holds for all $x\in \supp \CQ$. Here the relative $0$-th order capacity of $\{x\}$ with respect to $\supp \CQ$ is given by
\[{\mathrm{cap}_0}(\{x\},\supp \CQ) = \inf\{ \CE(\phi) \mid \phi\in \widetilde{\CD}(\CE),\ \phi(x)=1, \phi|_{\supp \CQ \setminus \{x\}}=0\}.\]
Furthermore we have
\[{\mathrm{cap}_0}(\{x\},\supp \CQ)\leq \sum_{e\sim x} \frac{1}{l_c(e)},\]
where equality holds if and only if all neighbors $y$ of $x$ belong to $\supp \CQ$.
\end{lemma}
\begin{proof}
We proceed as in section 2.5. For $x\in \FV$ we choose a function $\phi_x$ of finite energy such that $\phi(x)=1$ and $\phi(y)=0$ for all $y\in \supp \CQ$ with $y\neq x$. Then
\[\CQ(\phi_x)= \sum_{y\in \FV} j(x,y) + c(x) \leq C \CE(\phi_x).\]
Taking the infimum over all such $\phi_x$ we obtain the first claim. The second one follows as the minimizer is harmonic outside of $\FV$, and thus linear.
\end{proof}
To get a sufficient condition, take $x,y\in \FV$ with $j(x,y)>0$. Then
\[j(x,y) |u(x)-u(y)|^2 \leq j(x,y)l_c(x,y)\int\limits_{p_x^y} |u'(t)|^2 \: dt\]
where $p_x^y$ denotes a path from $x$ to $y$. Summing over all $x,y\in \FV$ gives
\[\sum_{x,y\in \FV} j(x,y) |u(x)-u(y)|^2 \leq \sum_{x,y\in \FV} j(x,y)l_c(x,y)\sum_{e\in \Fp_x^y}  \int\limits_{\FX_e} |u'(t)|^2 \: dt.\]
We change the order of summation to get that the right hand side equals
\[\sum_{e\in E} p(e) \int\limits_{\FX_e} |u'(t)|^2 dt\]
where the quantity
\[p(e):=\sum_{x,y \in \FV}j(x,y)l_c(x,y)\chi(\Fp_x^y,e),\]
with $\chi(\Fp_x^y,e)=1$ if $e\in \Fp_x^y$ and $0$ else. Note that $p$ depends on the special choices of paths connecting $x,y$ with $j(x,y)>0$. We will call such a special choice of paths a path covering.
\begin{prop}
If there is a path covering such that there is $C>0$ with
\[p(e) \leq C\]
for all $e\in E$, then $\CQ$ is $\CE$-form bounded.
\end{prop}
We discuss the above criteria by the following example, which also gives that this criterion is optimal in a certain sense.
\begin{example}
Let $\FX_\Gamma$ be a connected metric graph and $\CS$ a Dirichlet form such that $j(x,y)>0$ if and only if there exists an edge connecting $x$ and $y$. Thus we have a canonical path covering given by the edges and the condition of the discussion above reduces to
\[j(x,y)l_c(x,y)\leq C.\]
On the other hand, let $\FX_\Gamma$ be a tree and let $x,y\in \FV$, such that $x\sim y$. Then $\FX_\Gamma\setminus\{\{x\},\{y\}\}$ consists of three connected components, including the corresponding open edge $(x,y)=(\partial^+(e),\partial^-(e))$. We denote the connected component with boundary $x$ by $\FY_x$, and $\FY_y$ analogously. Then we define the function $u$ as
\[u(z):= \begin{cases} 1&, z\in \FY_x \\  0&, z\in \FY_y\end{cases}\]
and on the edge $(x,y)$ we define its linear interpolation. Then, if $\CQ$ is $\CE$-form-bounded we get
\[\CQ(u) = j(x,y)\]
and
\[\CE(u)= l_c(x,y)^{-1}.\]
In particular we obtain
\[j(x,y) l_c(x,y)\leq C\]
uniformly in $x,y$.
\end{example}
Within the following example we compare $\CE$-boundedness and $\CE$-form-boundedness.
\begin{example}
Let $\FX_\Gamma = \IR$, $\FV = \IZ\setminus\{0\}$ and $j(x,y)>0$ if and only if $|x|=|y|$. Then if $\sum_{n\in \IZ} j(-n,n) n < \infty$ then $\CQ$ is $\CE$-form bounded, whence it suffices that $j(-n,n)\leq C$ uniformly to obtain $\CE$-boundedness.
\end{example}
\section{The trace Dirichlet form}
Our aim in this section is to relate the extended Dirichlet space of $\CS$ with a discrete Dirichlet form. We will proceed as in the chapter of the Kirchhoff Dirichlet form. In general Dirichlet form theory this procedure is known as trace Dirichlet form. Recall that the Kirchhoff Dirichlet form with $\alpha =0$ is acting as
\[\CQ_0(u) = \sum_{x,y\in \FV} \frac{\omega(x,y)}{l_i(x,y)} (u(x)-u(y))^2.\]
Together with the discrete part $\CQ$ from $\CS$ the space of functions defined on the vertices where $\CQ$ and $\CQ_0$ is finite will appear below and is denoted in analogy of the graph Dirichlet forms by
\[\widetilde{\CD}(\CQ + \CQ_0).\]
We start with the decomposition theorem.
\begin{theorem}\index{trace Dirichlet form}
Let $\FX_\Gamma$ be a metric graph and $\CS$ a graph Dirichlet form. Then we have the orthogonal decomposition
\[\widetilde{\CD}(\CS) = \widetilde{\CD}(\CS)_\FV \oplus \CH_0^\FV\]
where
\[\widetilde{\CD}(\CS)_\FV = \bigoplus_{e\in E} \{u_e \mid u_e' \in L^2(0,l(e)), u_e(0)=u_e(l(e))=0\}\]
and $\CH_0^\FV$ denotes the space of all continuous functions of finite energy which are linear on each edge. Furthermore the mapping
\[\cdot|_\FV: \CH_0^\FV \to \widetilde{\CD}(\CQ + \CQ_0)\]
is an isometric isomorphism with inverse $H_0$.
\end{theorem}
\begin{proof}
Clearly the space of all functions in $\widetilde{\CD}(\CS)$ which are zero on the set of vertices is a closed subspace by the Sobolev embedding. In particular the edgewise decomposition is valid. Considering its orthogonal complement we have $u\in \widetilde{\CD}(\CS)_\FV^\bot$ if and only if
\[\CE(u,\phi)= 0\]
for all $\phi \in \bigoplus_{e\in E} \CC_c^\infty (0,l(e))$, and hence $u$ is linear on each edge. For $u\in \CH_0^\FV$ we have for $e=(x,y)$ the representation
\[u_e(t)= u(x) \frac{l(e)-t}{l(e)} + u(y) \frac{t}{l(e)}\]
and therefore
\[\CE(u)= \sum_{x,y\in \FV} \tfrac{\omega(x,y)}{l(x,y)} (u(x)-u(y))^2.\]
\end{proof}
Before we study the relation to the Kirchhoff Dirichlet form, we immediately deduce the following corollary.
\begin{coro}
The graph Dirichlet form $\CS$ is recurrent if and only if the discrete Dirichlet form $\CQ+\CQ_0$ is recurrent. In particular, a diffusion Dirichlet form is recurrent if and only if the associated Kirchhoff Dirichlet form is.
\end{coro}
We now want to study the difference between the trace and the Kirchhoff Dirichlet form. To this extent we introduce the so-called Feller measure, see \cite{CF-12}. For $\alpha >0$ and $V_1,V_2\in \CC(\FV)$ we set
\[ U_\alpha (V_1,V_2)= \alpha (H_\alpha V_1, H_0 V_2).\]
By definition $U_\alpha$ is bilinear and on its diagonal for $x,y\in \FV$  given by
\[U_\alpha(\delta_x,\delta_y) = \begin{cases} 0 &: x\neq y, x\not\sim y\\ - \sqrt{\alpha}\frac{\omega(x,y)}{\sinh \sqrt{\alpha} l_i(x,y)} + \frac{\omega(x,y)}{l_i(x,y)} &: x\sim y\\ \sum\limits_{e\sim x} \sqrt{\alpha} \omega(e) \frac{\cosh\sqrt{\alpha} l_i(e)}{\sinh \sqrt{\alpha}l_i(e)} - \frac{\omega(e)}{l_i(e)}&: x=y \end{cases}.\]
As the diagonals of $\CQ_0$ and $\CQ_\alpha$ are given by
\[\CQ_0(\delta_x,\delta_y) = \begin{cases} 0&: x\neq y, x\not\sim y\\ -\frac{\omega (x,y)}{l_i(x,y)} &: x\sim y\\ \sum\limits_{e\sim x} \frac{\omega(e)}{l_i(e)} &: x=y \end{cases}\]
and
\[\CQ_\alpha(\delta_x,\delta_y) = \begin{cases} 0&: x\neq y, x\not\sim y\\ -\sqrt{\alpha}\frac{\omega (x,y)}{\sinh \sqrt{\alpha} l_i(x,y)} &: x\sim y\\ \sum\limits_{e\sim x} \sqrt{\alpha} \frac{\omega(e)}{\sinh \sqrt{\alpha} l_i(e)} + \sqrt{\alpha} \omega(e) \frac{\cosh\sqrt{\alpha}l_i(e) - 1}{\sinh \sqrt{\alpha}l_i(e)} &: x=y \end{cases}\]
we obtain the equality of $\CQ_0 + U_\alpha = \CQ_\alpha$ on $\CC_c(\FV)$.
\section*{Notes and remarks}
The main purpose of this section was to determine the reflected and extended Dirichlet spaces of graph Dirichlet forms. We have seen that both appear as Hilbert spaces. This was already known for discrete Dirichlet forms. For general Dirichlet forms, it seems not to be true to equip this space with an inner product. Thus here we have an example which is not discrete. The explicit description of these spaces allowed us to perform certain potential theoretic procedures. In particular, it seems that a lot of results which hold for discrete graphs can also be shown for our model. In the last section we were able to calculate the trace Dirichlet and related Feller measures form explicitly. This gives a new characterization of recurrence in terms of discrete Dirichlet forms.

\backmatter
\bibliographystyle{amsalpha}
	\bibliography{Diss_bib}
%	\nocite*
\printindex
\end{document}